%% file: main.tex
\documentclass{lmcs} 
\pdfoutput=1
\usepackage[utf8]{inputenc}

\usepackage{lastpage}
\lmcsdoi{21}{1}{7}
\lmcsheading{}{\pageref{LastPage}}{}{}%
{Nov.~30,~2022}{Jan.~28,~2025}{}


\input{init}

\usepackage{enumitem}

\begin{document}

\title{Quantitative Equality in Substructural Logic via Lipschitz doctrines} 


\author[F.~Dagnino]{Francesco Dagnino\lmcsorcid{0000-0003-3599-3535}}[a]	
\author[F.~Pasquali]{Fabio Pasquali\lmcsorcid{0000-0002-3599-5525}}[b]	

\address{DIBRIS, Università di Genova, Italy}	
\email{francesco.dagnino@dibris.unige.it}  

\address{DIMA, Università di Genova, Italy}	
\email{pasquali@dima.unige.it}  





\begin{abstract}
Substructural logics naturally support a quantitative interpretation of formulas, as they are seen as consumable resources.  Distances are the quantitative counterpart of equivalence relations: they measure how much two objects are similar, rather than just saying whether they are equivalent or not. Hence, they provide the natural choice for modelling equality in a substructural setting.  In this paper, we develop this idea,  using the categorical language of Lawvere's doctrines.  We work in a minimal fragment of Linear Logic enriched by graded modalities, which are needed to write a resource sensitive substitution rule for equality, enabling its quantitative interpretation as a distance.  We introduce both a deductive calculus and the notion of Lipschitz doctrine to give it a sound and complete categorical semantics.  The study of 2-categorical properties of Lipschitz doctrines provides us with a universal construction, which generates examples based for instance on metric spaces and quantitative realisability. Finally, we show how to smoothly extend our results to richer substructural logics, up to full Linear Logic with quantifiers. 
\end{abstract}

\maketitle

\input{intro}

\input{background}

\input{monoidal}

\input{graded-modality}

\input{new-graded-logic}

\input{qet}

\input{categorical}

\input{altredoc}

\input{future}

\section*{Acknowledgment}
We are very grateful to all the anonymous referees who reviewed this paper. Their suggestions and comments have been really useful for improving our work. 

%
%

\bibliographystyle{alphaurl}
\bibliography{biblio}

\end{document}

%% file: init.tex
\usepackage{xspace}
\usepackage{amssymb,amsmath,bbold}
\usepackage{stmaryrd} 
\usepackage[mathscr]{euscript}
\usepackage[all]{xy}
\usepackage{mathtools}

\def\eg{{\em e.g.}\xspace} \def\ie{{\slshape i.e.}\xspace}
\def\cf{cf.~} 

\def\dfn#1{\emph{#1}} 

\newcommand{\refToSect}[1]{Section~{\ref{sect:#1}}}
\newcommand{\refToFig}[1]{Figure~{\ref{fig:#1}}}
\newcommand{\refToDef}[1]{Definition~{\ref{def:#1}}}
\newcommand{\refToDefItem}[2]{\refToDef{#1}({\ref{def:#1:#2}})}
\newcommand{\refToThm}[1]{Theorem~{\ref{thm:#1}}}
\newcommand{\refToProp}[1]{Proposition~{\ref{prop:#1}}}

\newcommand{\refToLem}[1]{Lemma~{\ref{lem:#1}}}
\newcommand{\refToCor}[1]{Corollary~{\ref{cor:#1}}}
\newcommand{\refToEx}[1]{Example~{\ref{ex:#1}}}
\newcommand{\refToExItem}[2]{\refToEx{#1}(\ref{ex:#1:#2})} 
\newcommand{\refToRem}[1]{Remark~{\ref{rem:#1}}}
\newcommand{\refToDia}[1]{Diagram~{\ref{dia:#1}}}

\newcommand{\Tuple}[1]{ {\langle {#1} \rangle} }

\def\ple#1{\ensuremath{\Tuple{{#1}}}}

\def\vuoto{}
\newcommand{\FUN}[4]{\ensuremath{{#2}#1{#3}\rightarrow{#4}}}
\newcommand{\fun}[3]{\relax\def\testa{#1}\relax\ifx\testa\vuoto
  \relax\FUN{}{}{{#2}}{{#3}}\else\relax\FUN{:}{{#1}}{{#2}}{{#3}}\fi}

\newcommand{\PW}[1]{\ensuremath{\mathop{\mathscr{P}{#1}}}}
\newcommand{\pw}[1]{\relax\def\testa{#1}\relax\ifx\testa\vuoto
  \relax\PW{}\else\relax\PW{\left(#1\right)}\fi}

\newcommand{\N}{\mathbb{N}}
 
\newcommand{\RPos}{[0,\infty)}

\def\tt{\ensuremath{\textsf{t\kern-.3ex t}}}
\def\ff{\ensuremath{\textsf{f\kern-.3ex f}}}
\let\Land\wedge

\let\ForalL\forall \def\Forall#1.{\ForalL_{#1}}
\let\ExistS\exists \def\Exists#1.{\ExistS_{#1}}
\newcommand{\Lbang}{\mathsf{!}} 
\newcommand{\Ltensor}{\otimes}
\newcommand{\Lone}{\mathbf{1}}
\newcommand{\Lzero}{\mathbf{0}}
\newcommand{\LLand}{\mathbin{\&}} 
\newcommand{\Ltop}{\top}

\newcommand{\Leq}{\eqcirc} 

\DeclareFontFamily{OT1}{pzc}{}
\DeclareFontShape{OT1}{pzc}{m}{it}{<->s*[1.30]pzcmi7t}{}
\DeclareMathAlphabet{\mathpzc}{OT1}{pzc}{m}{it}
\def\ct#1{\ensuremath{\mathpzc{#1}}}
\def\Ct#1{\ensuremath{\mathbf{#1}}}

\def\id#1{\ensuremath{\mathrm{id}_{#1}}}
\def\Id#1{\ensuremath{\mathrm{Id}_{#1}}}
\def\op{^{\mbox{\normalfont\scriptsize op}}}
\def\blank{\mathchoice{\mbox{--}}{\mbox{--}}
{\mbox{\scriptsize--}}{\mbox{\tiny--}}}

\newcommand{\CC}{\ct{C}\xspace}
\newcommand{\D}{\ct{D}\xspace}

\newcommand{\oneAr}[3]{\ensuremath{{#1}:{#2}\rightarrow{#3}}}
\newcommand{\twoAr}[3]{\ensuremath{{#1}:{#2}\Rightarrow{#3}}}
\def\tdot{\textbf{.}}
\def\lsta{\vrule depth4pt width0pt}
\def\lstb{\vrule height5pt width0pt}
\def\arnta[#1]{\ar[#1]|-*=0[@]{\lsta\tdot}}
\def\arntb[#1]{\ar[#1]|-*=0[@]{\lstb\tdot}}
\newcommand{\TNat}[3]{\ensuremath{{#1}:{#2}
 \stackrel{\makebox{\kern-.3ex\tdot}}\rightarrow{#3}}}

\newcommand{\Sub}[2][{}]{\relax\def\testa{#2}\relax
  \ensuremath{\mathsf{Sub}_{#1}\relax\ifx\vuoto\testa\relax\else
    \left(#2\right)\fi}}

\newcommand{\Set}{\ct{Set}\xspace}

\newcommand{\Pos}{\ct{Pos}\xspace}

\newcommand{\DC}{\Ct{Dtn}\xspace} 
\newcommand{\PD}{\Ct{PD}\xspace} 
 
\newcommand{\FOD}{\Ct{FOD}\xspace} 
\newcommand{\HD}{\Ct{HD}\xspace} 
\newcommand{\ED}{\Ct{ED}\xspace}
\newcommand{\MD}{\Ct{LD}\xspace}

\newcommand{\GMD}[1]{\ensuremath{\MD_{#1}}\xspace } 
\newcommand{\EGMD}[1]{\ensuremath{\Ct{LLD}_{#1}}\xspace } 
\newcommand{\ElGMD}[1]{\ensuremath{\Ct{ELD}_{#1}}\xspace }

\newcommand{\ELLL}[1]{\ensuremath{\Ct{QHD}_{#1}}\xspace } 
\newcommand{\LLL}[1]{\ensuremath{\Ct{FLD}_{#1}}\xspace } 

\newcommand{\GMEFun}{\shopr{R_{Lip}}} 
\newcommand{\GEMFun}{\shopr{U_{Lip}}} 
\newcommand{\ELFun}{\shopr{I_{EL}}} 
\newcommand{\LEFun}{\shopr{R_{EL}}} 

\newcommand{\ee}{\epsilon}

\def\enab#1{\ensuremath{#1}\xspace}
\newcommand{\cmd}{\enab{\mathsf{K}}}

\newcommand{\counit}{\enab{\epsilon}}

\edef\cdrestoreat{
\noexpand\catcode\lq\noexpand\@=\the\catcode\lq\@}\catcode\lq\@=11
\def\unita{\@ifnextchar[{\UNita}{\UNita[\cmd]}}
\def\UNita[#1]{\enab{\varsigma^{#1}}}
\def\Forget{\@ifnextchar[{\FOrget}{\FOrget[\cmd]}}
\def\ForgetLift{\@ifnextchar[{\FOrgetlift}{\FOrgetlift[\cmd]}}
\def\FOrget[#1]{\ensuremath{\mathit{U}^{#1}}}
\def\FOrgetlift[#1]{\ensuremath{\iota^{#1}}}
\cdrestoreat

\newcommand{\unit}{\enab{\eta}}

\newcommand{\Doc}{\ensuremath{\mathit{P}}\xspace}
\newcommand{\aDoc}{\ensuremath{\mathit{Q}}\xspace}

\newcommand{\reidx}[1]{_{#1}} 
\newcommand{\DReIdx}[2]{\ensuremath{{#1}\reidx{#2}}}
\newcommand{\terDoc}{\mathit{1}} 
\newcommand{\order}{\le}

\def\shopr#1{\ensuremath{\mathrm{#1}}\xspace}

\newcommand{\fmul}{\ast}
 
\newcommand{\funit}{\kappa}

\newcommand{\krpmul}{\circ}
\newcommand{\krpunit}{\epsilon}
\newcommand{\WW}{\mathbb{W}}

\newcommand{\RR}{\ensuremath{R}\xspace} 
\newcommand{\RSet}{|R|}
\newcommand{\rord}{\preceq} 
\newcommand{\rplus}{+}
\newcommand{\rmul}{\cdot}
\newcommand{\rzero}{0} 
\newcommand{\rone}{1} 
\newcommand{\res}{r}
\newcommand{\ares}{s} 
\newcommand{\bres}{q} 
\newcommand{\rmod}[1]{\flat_{#1}} 
\newcommand{\armod}[1]{\sharp_{#1}} 

\newcommand{\NN}{\mathbb{N}} 
\newcommand{\RRPos}{\mathbb{R}_{\geq0}} 

\newcommand{\delem}{\mathsf{d}} 
\newcommand{\adelem}{\mathsf{e}} 
\newcommand{\dist}{\rho}
\newcommand{\adist}{\sigma}
\newcommand{\bdist}{\tau} 
\newcommand{\dmul}{\boxtimes} 
\newcommand{\Met}[2]{\ensuremath{\ct{L}_{\ple{#1,#2}}}} 
\newcommand{\In}[2]{\ensuremath{#1_{!}}\xspace} 
\newcommand{\Des}[1]{\ensuremath{\ct{Des}_{#1}}} 
\newcommand{\DMet}[2]{\mathit{D}_{\ple{#1,#2}}}
\newcommand{\DIn}[1]{\ensuremath{#1_{!}}\xspace}

\newcommand{\lip}{Lipschitz\xspace} 
\newcommand{\pl}{primary linear\xspace} 
\newcommand{\Pl}{Primary linear\xspace} 

\newcommand{\LIP}{\ct{Met}_L}
\newcommand{\LIPDoc}{\ensuremath{\mathit{Lip}}\xspace}

\newcommand{\app}{.}

\def\tracks#1{\text{ tracks }}

\newif\ifcomments
\commentstrue

\def\RB#1{\mathchoice
  {\rotatebox[origin=c]{180}{$#1$}}
  {\rotatebox[origin=c]{180}{$#1$}}
  {\rotatebox[origin=c]{180}{$\scriptstyle#1$}}
  {\rotatebox[origin=c]{180}{$\scriptscriptstyle#1$}}}
\def\Ex{\RB{E}\kern-.3ex}
\def\Al{\RB{A}\kern-.6ex}

\usepackage{mathpartir} 
\usepackage{prooftree} 

\newcommand {\Infer} [3] [{}] {
  \inferrule*[%
    left={\rn{#1}},%
    vcenter%
  ]%
  {#2}{#3}
}
\newcommand{\rn}[1]{\textsc{[#1]}} 
\newcommand{\qrn}[1]{\textbf{\small (#1)}} 

\newcommand{\lcons}[3][]{{#2} \vdash_{#1} {#3}}
\newcommand{\form}{\phi} 
\newcommand{\aform}{\psi}
\newcommand{\bform}{\theta} 
\newcommand{\hps}{\Gamma}
\newcommand{\ahps}{\Delta} 
\newcommand{\trm}{t}
\newcommand{\atrm}{u} 
\newcommand{\btrm}{v} 
\newcommand{\var}{x}
 
\newcommand{\subst}[3]{#1[#2/#3]} 
\newcommand{\srt}{\sigma}
\newcommand{\asrt}{\tau}

\newcommand{\Var}{\mathit{V}}
\newcommand{\Trm}{\mathit{Trm}}
\newcommand{\Wff}{\mathit{Wff}} 
\newcommand{\sem}[1]{\llbracket #1 \rrbracket}  
\newcommand{\RPLL}{PLL$_\RR$\xspace}
\newcommand{\ERPLL}{L\RPLL\xspace} 
\newcommand{\ERLL}{LL$_\RR$\xspace} 
\newcommand{\IERLL}{ILL$_\RR$\xspace}

\newcommand{\gr}{\mathsf{gr}} 
 
\newcommand{\gar}[1]{|#1|}

\newcommand{\PP}{\pw{} } 
 
\newcommand{\Prop}{Prp} 
\newcommand{\TT}{\mathcal{T}}

\newcommand{\pr}{\mathrm{pr}} 

\newcommand{\qsplus}{+}
\newcommand{\qszero}{0} 
 
\newcommand{\mvplus}{+}
\newcommand{\mvzero}{0}
\newcommand{\mvmin}{-}
\newcommand{\mvscalar}{\lambda} 
\newcommand{\onepred}{\mathsf{u}}

\newcommand{\lcab}{\mathsf{b}}

%% file: intro.tex

\section{Introduction}
\label{sect:intro}

Equality is probably the most elementary and non-trivial predicate one can consider in logic. 
It enables a very basic task, that is,  reasoning about the identity of expressions 
and thus of the objects they denote. 
Equality in First Order Logic is nowadays fairly well understood from both a syntactic and a semantic perspective. 
On the syntactic side, 
equality can be described as a binary predicate $\trm\Leq\atrm$, relating two terms $\trm$ and $\atrm$, which has to satisfy a couple of properties: 
every term has to be equal to itself (\emph{reflexivity}) and, 
whenever we know that $\trm$ and $\atrm$ are equal and a formula $\form$ holds on $\trm$, then it holds on $\atrm$ as well (\emph{substitutivity}). 
In other words, 
the former establishes that some trivial identities always hold, 
while the latter  allows one to \emph{transport} a property $\form$ from $\trm$ to $\atrm$ when they are equal; 
indeed, substitutivity is also known as the transport property of equality. 
More formally,
the equality predicate has to satisfy the following entailments: 
\[
\lcons{\Ltop}{\trm\Leq\trm}\qquad\qquad
\lcons{\subst{\form}{\trm}{\var}\LLand\trm\Leq\atrm}{\subst{\form}{\atrm}{\var}} 
\]
where $\Ltop$ is the true predicate, $\LLand$ is the conjunction and $\subst\form\trm\var$ denotes the substitution of the term $\trm$ for the variable $\var$ in the formula $\form$. 

A simple semantic counterpart of the above description of equality  is given by \emph{equivalence relations}. 
Indeed, considering a set $X$ of objects/individuals one is interested in, 
a property/formula on $X$ can be seen as a function \fun\alpha{X}{\{0,1\}}, which says, for each object $x$ in $X$, whether or not it has the property $\alpha$. 
Then, equality can be safely modelled by an equivalence relation on $X$, namely, 
a binary function \fun{\rho}{X\times X}{\{0,1\}} enjoying 
the usual reflexivity, symmetry and transitivity properties,
provided that formulas are modelled by those $\alpha$ 
compatible with $\rho$. 
That is, 
if $x$ and $y$ are equivalent according to $\rho$ (\ie $\rho(x,y) = 1$), then either both satisfy $\alpha$ or both do not (\ie $\alpha(x) = \alpha(y)$). 

The understanding of equality becomes much less clear when moving to \emph{substructural} logics, such as, (fragments of) Linear Logic \cite{Girard87}. 
The core idea of these logics is to see formulas as consumable resources in a proof, which therefore cannot  be freely deleted or duplicated. 
To achieve this, some structural rules of usual deductive calculi are removed, specifically weakening and contraction rules, which are precisely those enabling deletion and duplication of formulas. 
As a consequence, usual true $\Ltop$ and conjunction $\LLand$ are replaced by $\Lone$ and $\Ltensor$, respectively, which behaves just as monoid-like operations (in particular, $\Ltensor$ is not idempotent and $\Lone$ is not a top element). 

Even though the  
model  sketched above cannot work in general for the substructural case, 
since it obviously validates both weakening and contraction, 
from a semantic perspective, 
we can still have a quite natural understanding of equality, by taking a \emph{quantitative} point of view.
Intuitively, a quantitative semantics
measures how much  a property holds rather than just saying whether it holds or not. 
As a paradigmatic example, let us model a quantitative property on a set $X$ by a function \fun{\alpha}{X}{[0,\infty]} taking values in the set of extended non-negative real numbers (usually considered with the reversed order). 
Intuitively, the value $\alpha(x) = \epsilon$ means that $x$ has the property $\alpha$ up to an error $\epsilon$, hence $\alpha$ surely holds on $x$ when $\alpha(x) = 0$.
Following this intuitive reading of quantitative properties, 
conjunction is given by the pointwise extension of addition of real numbers, which, being non-idempotent,  is capable of properly propagate errors. 
In this setting, 
the natural generalisation of equivalence relations are \emph{distances}, that is, binary functions 
\fun{\rho}{X\times X}{[0,\infty]} satisfying a quantitative version of reflexivity, symmetry and transitivity: 
\[
0\ge\rho(x,x) \qquad 
\rho(x,y)\ge \rho(y,x) \qquad 
\rho(x,y) + \rho(y,z) \ge \rho(x,z) 
\]
Note that the quantitative transitivity property is nothing but the usual triangular inequality. 

Now the question is: how can we introduce equality in substructural logics so that this quantitative semantics is supported? 
A first naive attempt 
could be to take the linear version of equality rules in First Order Logic, 
introducing a binary predicate $\trm\Leq\atrm$ which satisfies the linear reflexivity and substitutivity properties written below: 
\[
\lcons{\Lone}{\trm\Leq\trm}\qquad\qquad
\lcons{\subst{\form}{\trm}{\var}\Ltensor\trm\Leq\atrm}{\subst{\form}{\atrm}{\var}} 
\]
These entailments are precisely the usual ones where standard connectives $\Ltop$ and $\LLand$ are replaced by their linear versions $\Lone$ and $\Ltensor$. 
However smooth, this approach raises the unexpected issue that 
the equality predicate can be used an arbitrary number of times. 
Indeed, as already observed in \cite{Hodges1996LogicFF}, the following entailments are provable from the two properties above: 
\[
\lcons{\trm\Leq\atrm}{\Lone}\qquad\qquad
\lcons{\trm\Leq\atrm}{(\trm\Leq\atrm) \Ltensor (\trm\Leq\atrm)} 
\]
The former says that equality can be deleted, while the latter that it can be duplicated. 
In other words, this shows that weakening and contraction rules, which in general are banned from substructural logics, are instead admissible for equality.
This fact is undesirable per se in a substructural setting, as it is against its goal of controlling the use of formulas,
but it is even worse from the semantic perspective, 
as the quantitative interpretation of equality is broken. 
Indeed, given a distance
 \fun{\rho}{X\times X}{[0,\infty]} interpreting equality, the second entailment above forces it to satisfy the inequality
$\rho(x,y) \ge \rho(x,y) + \rho(x,y)$, 
 which is true only when $\rho(x,y)$ is either $0$ or $\infty$. This means that $\rho$ is, de facto, an equivalence relation.

In this paper we design a novel approach to equality in substructural logic, which is capable of preserve its natural quantitative interpretation as a (non-trivial) distance. 
On the one hand, this treatment of equality, enabling more control on its usage, is more in the spirit of substructural logics and, 
on the other, it could also shed a new light on the connection between such logics and \emph{quantitative reasoning}, defining a set of formal tools supporting it. 
Indeed, quantitative methods are increasingly used in many different domains, such as 
differential privacy \cite{ReedP10,BartheKOB13,XuCL14,BartheGHP16}, 
denotational semantics \cite{BaierM94,AmorimGHKC17}, 
program/behavioural metrics  \cite{DesharnaisGJP04,BreugelW05,ChatzikokolakisGPX14,Gavazzo18,DalLagoGY19,DalLagoG22}, 
and rewriting \cite{Gavazzo23}. With respect to usual techniques, quantitative methods better cope with the imprecision arising when one reasons about the behaviour of complex software systems also interacting with physical processes. 
The main formal system supporting these techniques is Quantitative Equational Logic, introduced and developed by 
Mardare et al. \cite{MardarePP16,MardarePP17,MardarePP21,Adamek22,BacciMPP18}. 
This is a logical system with all structural rules, specifically designed to reason about real-valued metrics, using ad-hoc symbols representing approximations of such metrics.
Nevertheless, 
quantitative methods are also naturally connected to substructural logics, however, such a connection is still less explored, hence, 
in this paper, we investigate it systematically, developing both syntax and (categorical) semantics of a (minimal) substructural logic with \emph{quantitative equality}.

Our analysis starts by observing that, 
in analogy with what happens in the specific case of program metrics \cite{CrubilleDL17,Gavazzo18}, 
the usual substitutive property 
 is not resource sensitive as 
it does not take into account resources needed to perform the substitution:
the equality $\trm\Leq\atrm$ has to be substitutive \emph{for any} formula $\form$, while the ``cost'' of substituting $\trm$ by $\atrm$ may be different for different formulas, since, for instance, it may depend on the number of occurrences of $\trm$. 
Hence, we need to rephrase the substitutive property to explicitly consider such resources. 
To this end, we work in (a fragment of) Linear Logic enriched with \emph{graded modalities} \cite{BreuvartP15}. 
Roughly, one adds to the logic a family of unary connectives $\Lbang_\res$, where $\res$ is 
taken from a structure of resources, which make explicit ``how much'' a formula can be used. 

Using graded modalities, we can rephrase the substitutive property as follows: 
\[
\lcons{\subst{\form}{\trm}{\var}\Ltensor\Lbang_\res (\trm\Leq\atrm)}{\subst{\form}{\atrm}{\var}}
\] 
where $\res$ depends on $\form$ (in a way that will be described in the paper) and encodes the amount of resources needed to derive the substitution. 
In this way, 
the equality predicate cannot be duplicated,
hence it admits a non-trivial quantitative interpretation. 

The main mathematical tool we will use throughout the paper are \emph{Lawvere's doctrines} \cite{LawvereF:adjif}, which provide an elegant categorical/algebraic understanding of logic. 
The key feature of this approach 
is that theories are 
written as functors, named doctrines, 
whose domain category
models contexts and terms, and the functor maps each object  to  a poset that models formulas on that object ordered by logical entailment. 
In this framework, logical phenomena can be explained algebraically, abstracting away from details and focusing only on essential features. 
For instance, replicability of standard equality in substructural logics has a neat algebraic explanation: 
such an equality is defined by a left adjoint, as pioneered by Lawvere~\cite{LawvereF:adjif,LawvereF:equhcs}, and, as we will show, predicates defined in this way are always replicable.  
This shows that a quantitative equality \emph{cannot be given by a left adjoint}, 
however, thanks to the language of doctrines, we manage to compare in a rigorous way  quantitative equality with the standard one, proving they share other fundamental structural properties. 
More precisely, as proved in \cite{EPR}, standard equality in non-linear doctrines has a coalgebraic nature and  quantitative equality, generalising the standard one, is coalgebraic as well. 

The rest of the paper is organised as follows. 
\refToSect{background} recalls known facts on equality in doctrines and presents their extension  to the linear setting,
discussing the non-quantitative nature of such an equality. 
\refToSect{graded} presents our framework for quantitative equality in a (graded) substructural logic. 
In \refToSect{graded-doc} we first define \RR-graded doctrines, 
which are  doctrines modelling the $(\Ltensor,\Lone)$-fragment of Linear Logic enriched by \RR-graded modalities, where \RR is an ordered  semiring of resources, 
and introduce  \RR-\lip doctrines, namely, \RR-graded doctrines with quantitative equality. 
Then, in \refToSect{grade-syntax}, we present a core deductive calculus for quantitative equality, based on that in \cite{BreuvartP15}, 
describing its sound and complete categorical semantics in \RR-\lip doctrines in \refToSect{grade-sem}. 
\refToSect{qet} compares our approach to Quantitative Equational Theories (QETs) \cite{MardarePP16,MardarePP17}, providing some examples of theories in our calculus, which cannot be expressed as QETs. 
\refToSect{categorical} analyses 2-categorical properties of \RR-\lip doctrines. 
In \refToSect{comonad} we show the coalgebraic nature of  \RR-\lip doctrines by proving they 
arise as coalgebras for a 2-comonad on the 2-category of \RR-graded doctrines. 
This provides us with a universal construction yielding an \RR-\lip doctrine  from an \RR-graded one, and we use it to generate semantics for the calculus. 
In \refToSect{rel-eq} we relate quantitative equality with the usual one  defined by left adoints, formally proving that the former indeed refines the latter. 
Then, in \refToSect{altredoc} we extend previous results to richer fragments of Linear Logic up to full LL with quantifiers. 
We conclude in \refToSect{future} discussing related and future work.\\

This paper is the journal version of \cite{DagninoP22}. 
Here, we give  more background on doctrines (\refToSect{doctrines}) and more details on standard equality in the substructural setting (\refToSect{monoidal}). 
Moreover, we provide a categorical comparison between our quantitative equality and the usual one by left adjoints (\refToSect{rel-eq}). 
Finally, we also present more examples of theories (\refToSect{qet}) and report all proofs of our results.

%% file: background.tex

\section{Preliminaries on (linear) doctrines} 
\label{sect:background}

Doctrines provide a simple categorical framework to study several kinds of logics. 
In this section, we first recall basic notions about doctrines for standard (\ie non-linear) logics, then we introduce the class of doctrines modelling the fragment of linear logic we will be concerned with.

\subsection{Doctrines in a nutshell} 
\label{sect:doctrines} 
 
Denote by \Pos the category of posets and monotone functions. 
A left adjoint to a monotone function $\fun{g}{K}{H}$ is a monotone  function $\fun{f}{H}{K}$ such that for every $x$ in $K$ and $y$ in $H$, both $y\le gf(y)$ and $fg(x)\le x$ hold. Equivalently $y\le g(x)$ if and only if $f(y)\le x$. 

A \dfn{doctrine} is a pair \ple{\CC,\Doc} where \CC
is a category with finite products and $\fun{\Doc}{\CC\op}{\Pos}$ is a
functor. The category \CC is named the \dfn{base} of the
doctrine. For $X$ an object in \CC the poset $\Doc(X)$ is called
\dfn{fibre over} $X$. For \fun{f}{X}{Y} an arrow in \CC, the
monotone function \fun{\DReIdx{\Doc}{f}}{\Doc (Y)}{\Doc (X)} is called
\dfn{reindexing along} $f$.
We will often refer to a doctrine $\ple{\CC,\Doc}$ using only the functor $ \fun{\Doc}{\CC\op}{\Pos}$.

A doctrine \fun{\Doc}{\CC\op}{\Pos} is \dfn{primary} if all fibres have finite meets and these are preserved by reindexing. The top element of $\Doc(A)$ will be denoted as $\top_A$, while binary meets by $\wedge_A$. We shall drop the subscripts when these are clear.

There is a large variety of examples of primary doctrines (see \cite{JacobsB:catltt, OostenJ:reaait, PittsCL}), in this paper we will exemplify definitions considering the following.

\begin{exas}\label{ex:running}
\begin{enumerate}
\item\label{ex:running:3}
Let $L$ be a many-sorted signature with sorts $\srt,\asrt,...$, 
function symbols $f,g,...$ and 
predicate symbols $p,q,...$. 
Let $\TT$ be a theory in the $(\LLand,\Ltop)$-fragment of First Order Logic over  $L$. The category $\ct{Ctx}_L$ has 
contexts $\vec{\srt} = \ple{x_1:\srt_1,\ldots,x_n:\srt_n}$ as objects  
and lists of terms $ \ple{\trm_1,\ldots,\trm_k}:\vec{\srt}\to \ple{y_1:\asrt_1,\ldots,y_k:\asrt_k}$ as arrows, where each
$\trm_i$ has sort $\asrt_i$ in the context $\vec{\srt}$. Binary products are given by context concatenation, up to renaming of variables.
As detailed in \cite[Example~2.2]{MaiettiME:quofcm} the 
 \dfn{syntactic doctrine $\Prop_{\TT}$} based on  $\ct{Cxt}_L$ maps each context $\vec{\srt}$ to the poset reflection of the collection of  
well-formed formulas over $L$ in the $(\LLand,\Ltop)$-fragment with free variables in $\vec{\srt}$ preordered by the entailment in $\TT$.
That is, 
$[\form]\le [\aform]$ iff $\lcons[\vec{\srt}]{\form}{\aform}$ is provable in $\TT$ ( In other words, the poset $\Prop_{\TT}(\vec{\srt})$ is the Lindenbaum-Tarski algebra of well-formed formulas over $L$ in the $(\LLand,\Ltop)$-fragment with free variables in $\vec{\srt}$).
The doctrine $\Prop_{\TT}$ is primary where conjunctions give finite meets. 

\item\label{ex:running:1}
Suppose $H$ is a poset, the functor $\PP_H:\Set\op\to\Pos$ sends a set $A$ to the set of functions $H^A$ ordered pointwise. 
For a function $f:X\to A$ and $\alpha$ in $\PP_H(A)$ the function $\PP_H(f)(\alpha)$ is the composite $\alpha f$ in $\PP_H(X)$. 
The doctrine $\PP_H$ is primary when $H$ is a meet-semilattice.
For $H=\{0,1\}$ the doctrine $\PP_{H}$ is the contravariant powerset functor that will be denoted simply by $\PP$.

\item\label{ex:running:2}Suppose that $A=(|A|, \app , K, S)$ is a (partial) combinatory algebra as in \cite{OostenJ:reaait}. 
The doctrine $\mathcal{R}_A:\Set\op\to\Pos$ maps each set $X$ to the poset reflection of the preorder $\PP(|A|)^X$ where 
for $\alpha,\beta:X\to \PP(|A|)$, $\alpha\le \beta$  whenever there is $a$ in $|A|$ such that, for every $x$ in $X$ and  $b$ in $\alpha(x)$, the application $a\app b$ is defined and belongs to $\beta(x)$. One says that $a$ realises $\alpha\le \beta$.
The action of $\mathcal{R}_A$ on functions is given by pre-composition. 
$\mathcal{R}_A$ is primary \cite{OostenJ:reaait}.

\end{enumerate}
\end{exas}

Doctrines are the objects of the 2-category \DC where  
1-arrows 
\oneAr{\ple{F,f}}{\ple{\CC,\Doc}}{\ple{\D,\aDoc}} 
consist of 
a finite product preserving functor 
$\fun{F}{\CC}{\D}$ and 
a natural transformation 
$\TNat{f}{\Doc}{\aDoc F\op}$ as in the diagram 
\[
\xymatrix@C=7.5em@R=1em{
{\CC\op}\ar[rd]^(.4){P}_(.4){}="P"
\ar[dd]_{F\op}^{}="F"
&\\
 & {\ct{Pos}}\\
{\D\op}\ar[ru]_(.4){Q}^(.4){}="R"&
\ar"P";"R"_{f\kern.5ex\cdot\kern-.5ex}="b"
}
\]
the 2-arrows 
\twoAr{\theta}{\ple{F,f}}{\ple{F',f'}}
are natural transformations $\TNat{\theta}{F}{F'}$ such that 
$f_X\order_{FX}\DReIdx{\aDoc}{\theta_X}\circ f'_X$ for
every object $X$ in \CC, \ie
\[
\xymatrix@C=12em@R=1.5em{
{\CC\op}\ar[rd]^(.4){P}_(.4){}="P"
\ar@<-1ex>@/_/[dd]_{F\op}^{}="F"\ar@<1ex>@/^/[dd]^{(F')\op}_{}="G"&\\
 & {\ct{Pos}}\\
{\D\op}\ar[ru]_(.4){Q}^(.4){}="R"&
\ar@/_/"P";"R"_{f\kern.5ex\cdot\kern-.5ex}="b"
\ar@<1ex>@/^/"P";"R"^{\kern-.5ex\cdot\kern.5ex f'}="c"
\ar"G";"F"_{.}^{\theta\op}\ar@{}"b";"c"|{\le}}
\]
The composition 
of two consecutive 1-arrows \ple{F,f} and \ple{G,g} is given by 
$\ple{G,g}\circ \ple{F,f} = \ple{GF, (gF\op) f}$, while 
composition of 2-arrows is that of natural transformations. 

Primary doctrines are the objects of the 2-full 2-subcategory \PD of \DC, where 
a 1-arrow from $\Doc$ to $\aDoc$ is a 1-arrow $\ple{F,f}$ in \DC such that each component of $f$ preserves finite meets. 

\begin{rem}\label{rem:modelli}
The 2-categorical structure of doctrines has a relevant logical meaning: 
objects correspond to theories, 
1-arrows to models/translations between theories 
and 
2-arrows to homomorphisms of models/translations. 
Hence, the hom-category of two doctrines $\Doc$ and $\aDoc$ corresponds to the category of models of $\Doc$ into $\aDoc$. 
 
Under this interpretation, each 2-category of doctrines 
corresponds to a class of theories with their models. 
For instance, the 2-category \PD corresponds to  theories in the $(\Land,\top)$-fragment of First Order Logic. 
Thus, given a theory $\TT$ in this fragment over a signature $L$ as in \refToExItem{running}{3},  a categorical semantics of $\TT$ into a primary doctrine  $\Doc$ is an object of $\PD(\Prop_{\TT},\Doc)$.
\end{rem}

A categorical description of equality in terms of left adjoints goes back to Lawvere \cite{LawvereF:equhcs}. 
More recently, Maietti and Rosolini \cite{MaiettiME:eleqc}   
identify primary doctrines as the essential framework where equality can be studied. 
Primary doctrines with equality  are called elementary  \cite{MaiettiME:eleqc}
and are defined as follows.

\begin{defi}\label{def:pelementary}
A primary doctrine \fun{\Doc}{\CC\op}{\Pos} is \dfn{elementary} if 
for every $A$ in \CC, 
there is  an element $\delta_A$ in $\Doc(A\times A)$ such that 
\begin{enumerate}
\item $\top_A \order \DReIdx{\Doc}{\Delta_A}(\delta_A)$ 
\item for all $X$ in \CC and $\alpha$ in $\Doc(X\times A)$ it holds that 
$\DReIdx{\Doc}{\ple{\pi_1,\pi_2}}(\alpha) \land \DReIdx{\Doc}{\ple{\pi_2,\pi_3}}(\delta_A) \order \DReIdx{\Doc}{\ple{\pi_1,\pi_3}}(\alpha)$
\end{enumerate}
 \end{defi}

\begin{exas}\label{ex:erunning}
\begin{enumerate}
\item\label{ex:erunning:3}
Consider a theory $\TT$ over the $(\LLand,\Ltop,\Leq)$-fragment of First Order Logic. 
Then, the syntactic doctrine $\Prop_\TT$, defined as in \refToExItem{running}{3}, is elementary. 
Indeed, for every context $\vec{\srt} = \ple{x_1:\srt_1,\ldots,x_n:\srt_n}$ one can take as $\delta_{\vec{\srt}}$ the formula
$x_1\Leq y_1 \LLand\ldots \LLand x_n\Leq y_n$. 
over the context $\vec{\srt}\times \vec{\srt}= \ple{x_1:\srt_1,\ldots,x_n:\srt_n,y_1:\srt_1,\ldots,y_n:\srt_n}$; 
conditions (1) and (2) in \refToDef{pelementary} require the provability of entailments of the form
\[
\Ltop \vdash x_1\Leq x_1\LLand\ldots\LLand x_n\Leq x_n 
\qquad 
\form(\vec{z},x_1,\ldots, x_n) \LLand x_1\Leq y_1 \LLand \ldots \LLand x_n\Leq y_n\vdash \form(\vec{z},y_1,\ldots, y_n)
\]
which easly follows from reflexivity and substitutivity of $x\Leq y$. 

\item\label{ex:erunning:1} The primary doctrine $\PP_H$ of \refToExItem{running}{1} is elementary when $H$ has a bottom element. For a set $X$,  the equality predicate $\delta_X\in H^{X\times X}$ is the Kronecker delta mapping $(x,x')$ to $\top$ if $x=x'$ and to $\bot$ otherwise.

\item\label{ex:erunning:2} The primary doctrine $\mathcal{R}_A$ of \refToExItem{running}{2} is elementary where for a set $X$ the function $\delta_X$ maps $(x,x')$ to $|A|$ if $x=x'$ and to $\emptyset$ otherwise \cite{OostenJ:reaait}.
\end{enumerate}
\end{exas}

Elementary doctrines are the objects of the 2-category  \ED 
where 
a 1-arrow 
\oneAr{\ple{F,f}}{\Doc}{\aDoc} is a 1-arrow in \PD such that $f$ preserves equality predicates. The
 2-arrows are those of  \PD. 

As studied in \cite{EPR}, equality is coalgebraic in the sense that the obvious inclusion of \ED into \PD has a right 2-adjoint 
\begin{equation}\label{dia:grafico}
\xymatrix{
\ED\ar@{^(->}[rr]<-1ex> \ar@{}[rr]|-{\top}&&\ar[ll]<-1ex> \PD
}
\end{equation}
and $\ED$ is isomorphic to the 2-category of coalgebras  of the induced 2-comonad on $\PD$.

\begin{rem}\label{rem:modelli1}Similarly to what we have already observed in \refToRem{modelli}, given a theory $\TT$ in the $(\LLand,\Ltop,\Leq)$-fragment of First Order Logic, a categorical semantics of $\TT$ in an elementary doctrine $\Doc$ is an object of $\ED(\Prop_{\TT},\Doc)$. Therefore the adjoint situation in \refToDia{grafico} allows one to build semantics for theories with equality, from semantics of theories without equality. 
\end{rem}

A primary doctrine \fun{\Doc}{\CC\op}{\Pos} is a \dfn{first order doctrine} if $\Doc$ factors through the category of Heyting algebras and their homomorphisms and 
for every projection $\pi:A\times B\to B$ in $\ct{C}$ the map $\Doc(\pi)$ has a left adjoint $\Ex_\pi$ and a right adjoint $\Al_\pi$ and these adjoints are natural in $B$, that is for every square
\[
\xymatrix{
A\times X\ar[r]^{\pi}\ar[d]_-{\id{A}\times f}&X\ar[d]^-{f}\\
A\times B\ar[r]_-{\pi}&B
}
\]
it holds that $\DReIdx{\Doc}{f}\,\Ex_{\pi}=\Ex_\pi\,\DReIdx{\Doc}{\id{A}\times f}$ and $\DReIdx{\Doc}{f}\,\Al_{\pi}=\Al_\pi\,\DReIdx{\Doc}{\id{A}\times f}$.
The naturality of left and right adjoints is often called Beck-Chevalley condition. A primary doctrine \fun{\Doc}{\CC\op}{\Pos} is a \dfn{hyperdoctrine} if it is first order and elementary.  

\begin{exas}\label{ex:erunningfod}
\begin{enumerate}
\item\label{ex:erunningfod:3}The syntactic doctrine $\Prop_\TT$ over a theory $\TT$ in full First Order Logic is first order (recall also \refToExItem{running}{3}).  For every context $\vec{\srt}$ the logical connectives between formulas with free variables in $\vec{\srt}$ provide the operations of the Heyting algebra.
For a formula $\alpha$ over $\vec{\srt}\times \vec{\asrt}=\ple{x_1:\srt_1,\ldots,x_n:\srt_n\,,\, y_1:\asrt_1,\ldots,y_k:\asrt_k}$ the left adjoint along $\pi:\vec{\srt}\times \vec{\asrt}\to \vec{\asrt}$ is the  formula $\Ex_{\pi}\alpha$ with free variables over $\vec{\asrt}$ defined as
\[
\exists_{x_1:\srt_1}\ldots \exists_{x_n:\srt_n}\alpha
\]
The right adjoint $\Al_{\pi}\alpha$ is similar where the universal quantification is used in place of the existential one. The syntactic doctrine $\Prop_\TT$ over a theory $\TT$ (see also \refToExItem{erunning}{3}) in full First Order Logic with equality is a hyperdoctrine.

\item\label{ex:erunningfod:1} The primary doctrine $\PP_H$ of \refToExItem{running}{1} is an hyperdoctirne when $H$ is a complete Heyting algebra, \ie it has arbitrary suprema that distribute over meets \cite{PittsCL, OostenJ:reaait}.

\item\label{ex:erunningfod:2} The primary doctrine $\mathcal{R}_A$ of \refToExItem{running}{2} is an hyperdoctrine \cite{OostenJ:reaait}.
\end{enumerate}
\end{exas}

Denote by \FOD the 2-full 2-subcategory of \PD on first order doctrines and 
those 1-arrows of \PD that additionally commute with all connectives and quantifiers.
Similarly, 
\HD denotes the 2-full 2-subcategory of \ED on hyperdoctrines\footnote{One can see \HD as the pullback of $\ED\to \PD$ along $\FOD\to \PD$.}. The adjoint situation depicted in \refToDia{grafico} restricts to the inclusion of \FOD into \PD and co-restricts to the inclusion of \HD into \ED, leading to the following diagram

\begin{equation}\label{dia:grafico1}
\xymatrix{
\ED\ar@{^(->}[rr]<-1ex> \ar@{}[rr]|-{\top}&&\ar[ll]<-1ex> \PD 
\\
\HD\ar@{^(->}[u]\ar@{^(->}[rr]<-1ex>\ar@{}[rr]|-{\top}&&\ar[ll]<-1ex>\FOD\ar@{^(->}[u]
}
\end{equation}
\vspace{.5em}

Combining the observations made in \refToRem{modelli} and in  \refToRem{modelli1} it not surprising that, given a theory $\TT$ in  full First Order Logic over a signature $L$, a categorical semantics of $\TT$ in a first order doctrine $\Doc$ is an object of $\FOD(\Prop_{\TT},\Doc)$; and if  $\TT$ is with equality and $P$ is a hyperdoctrine (and not only a first order doctrine) models of $\TT$ in $\Doc$ are objects $\HD(\Prop_{\TT},\Doc)$.

%% file: monoidal.tex
\subsection{Primary and elementary linear  doctrines}\label{sect:monoidal} 
Primary doctrines, that models the $(\&,\top)$-fragment of First Order Logic, are the minimal setting where one can formulate rules of equality (\cf\refToDef{pelementary}).  
We shall call their linear counterparts \pl doctrines. They model the $(\Ltensor,\Lone)$-fragment of Linear Logic and provide the minimal structure needed to describe equality in a substructural setting. 
The following definition is carved out from Seely's definition of linear hyperdoctrines \cite{seely1989linear}. 

\begin{defi}\label{def:monoid-doc}
A \dfn{ \pl  doctrine} is a triple $\ple{\Doc,\fmul,\funit}$ where $\fun{\Doc}{\CC\op}{\Pos}$ 
is a doctrine and \TNat{\fmul}{\Doc\times\Doc}{\Doc} and \TNat{\funit}{\terDoc_\CC}{\Doc} are natural transformations such that, for all objects $X$ in \CC, 
\ple{\Doc(X),\fmul_X,\funit_X} is a commutative monoid. 
\end{defi}

In other words, a \pl doctrine is a doctrine \fun{\Doc}{\CC\op}{\Pos} that factors through the category of ordered commutative monoids. 
More explicitly, 
for each object $X$ in \CC, we have a monotone function \fun{\fmul_X}{\Doc (X) \times \Doc (X)}{\Doc (X)} and an element $\funit_X\in\Doc (X)$
such that $\fmul_X$ is associative and commutative and $\funit_X$ is the neutral element 
 for $\fmul_X$ and, 
for \fun{f}{X}{Y} in \CC and $\alpha,\beta$ in $\Doc(Y)$, we have  
$\DReIdx{\Doc}{f}(\alpha)\fmul_X\DReIdx{\Doc}{f}(\beta) = \DReIdx{\Doc}{f}(\alpha\fmul_Y\beta)$ and $\funit_X = \DReIdx{\Doc}{f}(\funit_Y)$\footnote{Concisely, a primary linear doctrine is a monoid in $[\ct{C}\op,\ct{Pos}]$.}. 
We shall often omit subscripts from $\fmul$ and $\funit$ when these are obvious.

\Pl doctrines  are the objects of the 2-category \MD where a 1-arrow 
from \ple{\Doc,\fmul^\Doc,\funit^\Doc}to \ple{\aDoc,\fmul^\aDoc,\funit^\aDoc} is a 1-arrow \oneAr{\ple{F,f}}{\Doc}{\aDoc} in \DC 
such that $f$ is a natural monoid homomorphism, and 
2-arrows are those of \DC.

\begin{exas}\label{ex:runningmon1}
\begin{enumerate}
\item\label{ex:runningmon1:22} 
For a theory $\TT$ in the $(\Ltensor,\Lone)$-fragment of Linear Logic 
 \ie in the fragment described by the rules in \refToFig{frammento} (see \cite{Girard87}) 
\begin{figure}[t] 
\begin{mathpar}
\Infer[Ax]{ }{ \lcons{\form}{\form} } 
\and 
\Infer[Cut]{
  \lcons{\hps}{\form} 
\quad  \lcons{\ahps,\form}{\aform} 
}{ \lcons{\hps,\ahps}{\aform} }\\
\Infer[$\Lone$L]{
  \lcons{\hps}{\form} 
}{ \lcons{\hps,\Lone}{\form} }
\and
\Infer[$\Lone$R]{ }{ \lcons{}{\Lone} } 
\and
\Infer[$\Ltensor$L]{
  \lcons{\hps,\aform,\bform}{\form} 
}{ \lcons{\hps,\aform\Ltensor\bform}{\form} } 
\and 
\Infer[$\Ltensor$R]{
  \lcons{\hps}{\form} 
  \quad 
  \lcons{\ahps}{\aform} 
}{ \lcons{\hps,\ahps}{\form\Ltensor\aform }} 
\end{mathpar} 
\caption{Rules for $(\Ltensor,\Lone)$-fragment of Linear Logic}\label{fig:frammento} 
\end{figure} the syntactic doctrine $\Prop_\TT$ as in \refToExItem{running}{3}  is  \pl. 

\item\label{ex:runningmon1:1}
If $\ple{H,\fmul,\funit}$ is an ordered commutative monoid, then  $\ple{\PP_H,\fmul,\funit}$ is \pl, 
where $\PP_H$ is as in \refToExItem{running}{1} and $\fmul$ and $\funit$ are defined pointwise. 
\item\label{ex:runningmon1:12} 
A BCI algebra $A=\ple{|A|,\app, B, C,I}$ can be seen as a linear version of a combinatory algebra where linear combinators $B, C,I$ are used in place of combinators $K, S$ (a precise definition is in \cite{AbramskyHS02}). Analogously to \refToExItem{running}{2} we define the realisability doctrine \fun{\mathcal{R}_A}{\Set\op}{\Pos} as the functor that maps each set $X$ to the poset reflection of the preorder $\PP(|A|)^X$ where 
for $\alpha,\beta\in\PP(|A|)^X$, it is $\alpha\le \beta$ whenever there is $a\in|A|$ such that, for every $x\in X$ and  $b\in\alpha(x)$, the application $a\app b$ is defined and belongs to $\beta(x)$. 
The action of $\mathcal{R}_A$ on functions is by pre-composition. 
Each $\mathcal{R}_A(X)$ is a poset as $I$ and $B$ realise reflexivity and transitivity of $\le$, respectively. 
Given $\alpha,\beta\in\PP(|A|)^X$, define $\alpha\fmul \beta = x \mapsto  \{\textbf{P}ab\in |A| \mid  a\in \alpha(x)\  \text{and}\ b \in \beta(x)\}$, where $\textbf{P}$ is the BCI pairing \cite{Hoshino}, and $\funit_X = x \mapsto \{I\}$. 
Then, \ple{\mathcal{R}_A,\fmul,\funit} is a \pl doctrine.

\item\label{ex:runningmon1:44}
Let $\WW=\ple{W,\le,\krpmul,\krpunit}$ be an ordered commutative monoid (a.k.a. monoidal Kripke frame). 
Adapting from \cite{DalLagoG22} the notion of $\WW$-relation, we define a doctrine 
$\fun{K_\WW}{\Set\op}{\Pos}$ as follows: 
for a set $X$, $K_\WW(X)$ is the collection of those subsets $U$ of $X\times W$ such that if $\ple{x,w}\in U$ and $w\le w'$ then $\ple{x,w'}\in U$ ordered by set inclusion, and 
for a function $\fun{f}{X}{Y}$, $K_\WW(f)$ maps $U\in K_\WW(Y)$ to $(f\times\id{W})^{-1}(U)$. 
This is actually a \pl doctrine, where, for every set $X$, 
$\funit_X = \{ \ple{x,w}\mid \krpunit\le w\}$ and 
$U\fmul_X V = \{\ple{x, w}\mid \exists_{w_1,w_2}.\, w_1\krpmul w_2\le w, \ple{x,w_1}\in U, \ple{x,w_2}\in V\}$. 
\end{enumerate}
\end{exas}

Recall that elementary doctrines can be understood as those primary doctrine endowed with equality predicates. The following definition introduces those \pl doctrines that are elementary as a direct linearisation of \refToDef{pelementary}   (\ie replacing of $\wedge, \top$ by $\fmul, \funit$).

\begin{defi}\label{def:lelementary}
A \pl doctrine \fun{\Doc}{\CC\op}{\Pos} is \dfn{elementary} if 
for every $A$ in \CC, 
there is  an element $\delta_A$ in $\Doc(A\times A)$ such that 
\begin{enumerate}
\item $\funit_A \order \DReIdx{\Doc}{\Delta_A}(\delta_A)$ 
\item for all $X$ in \CC and $\alpha$ in $\Doc(X\times A)$ it holds that 
$\DReIdx{\Doc}{\ple{\pi_1,\pi_2}}(\alpha) \fmul \DReIdx{\Doc}{\ple{\pi_2,\pi_3}}(\delta_A) \order \DReIdx{\Doc}{\ple{\pi_1,\pi_3}}(\alpha)$
\end{enumerate}
 \end{defi}

Lawvere first noticed that equality is a left adjoint \cite{LawvereF:equhcs}. This turns out to be true also in the linear setting.

\begin{prop}\label{prop:equivalenza} 
Let  \fun{\Doc}{\CC\op}{\Pos} be a \pl doctrine, the following are equivalent
\begin{enumerate}
\item $\Doc$ is elementary
\item
For every $A$ there is $\delta_A$ in $\Doc(A\times A)$ such that for every $X$ and every $\alpha$ in $\Doc(X\times A)$ the assignment 
\[
\xymatrix@R=0.2em{
\Doc(X\times A)\ar[r]^-{\Ex_{\id{X}\times \Delta_A}}&\Doc(X\times A\times A)\\
\alpha\ar@{|->}[r]&\DReIdx{\Doc}{\ple{\pi_1,\pi_2}}(\alpha) \fmul \DReIdx{\Doc}{\ple{\pi_2,\pi_3}}(\delta_A)
}\]
defines a left adjoint to $\DReIdx{\Doc}{\id{X}\times \Delta_A}$.
\end{enumerate}
\end{prop}
\begin{proof} 
(1)$\Rightarrow$ (2). Suppose $\Ex_{\id{X}\times \Delta_A}(\alpha)\order \beta$ and apply $\DReIdx{\Doc}{\id{X}\times \Delta_A}$ to both side of the inequality to get $\alpha\order \DReIdx{\Doc}{\id{X}\times \Delta_A}(\beta)$. Conversely suppose $\alpha\order \DReIdx{\Doc}{\id{X}\times \Delta_A}(\beta)$. Consider $\beta$ as an element over  $(X\times A)\times A$ and let $p_1,p_2$ and $p_3$ denote projections from $(X\times A)\times A\times A$ and $\pi_1, \pi_2$ and $\pi_3$ denote projections from $X\times A\times A$. It holds 
$ 
\DReIdx{\Doc}{\ple{p_1,p_2}}(\beta)\fmul \DReIdx{\Doc}{\ple{p_2,p_3}}(\delta_A)\order \DReIdx{\Doc}{\ple{p_1,p_3}}(\beta)
$.
Evaluate both side of the inequality under \DReIdx{\Doc}{\ple{\pi_1,\pi_2,\pi_2,\pi_3}} to get
\[
\DReIdx{\Doc}{\ple{\pi_1,\pi_2}}\DReIdx{\Doc}{\id{X}\times \Delta_A}(\beta)\fmul \DReIdx{\Doc}{\ple{\pi_2,\pi_3}}(\delta_A)\order \beta\] and hence the claim.

(2)$\Rightarrow$ (1). Note that $\delta_A=\Ex_{\Delta_A}(\funit_A)$, so $\funit_A \order \DReIdx{\Doc}{\Delta_A}(\delta_A)$. Given $\alpha$ in $\Doc(X\times A)$, from $\alpha= \DReIdx{\Doc}{\ple{\id{X}\times \Delta_A}}\DReIdx{\Doc}{\ple{\pi_1,\pi_3}}(\alpha)$ follows $\Ex_{\id{X}\times \Delta_A}(\alpha)=\DReIdx{\Doc}{\ple{\pi_1,\pi_2}}(\alpha) \fmul \DReIdx{\Doc}{\ple{\pi_2,\pi_3}}(\delta_A) \order \DReIdx{\Doc}{\ple{\pi_1,\pi_3}}(\alpha)$.
\end{proof}

The characterisation  \ref{prop:equivalenza} will be helpful to prove some basic properties of equality predicates. Equality predicates in elementary doctrines (\ie in the non-linear setting) are proved to be equivalence relations. 
Their linear counterpart are distances\footnote{Binary products in base $\ct{C}$ of $\Doc$ play a crucial role in defining what a distance is, as one need diagonals and, more generally, tuples of arrows to formulate reflexivity and transitivity. This is clear from the syntactic point of view, as one can see from \refToExItem{distance}{1}, where to write conditions for reflexivity and transitivity one uses formulas where the same variables occurs multiple times.}, which can be defined in an abstract purely algebraic way in any \pl doctrine 
\ple{\Doc,\fmul,\funit} as follows:
a \dfn{\Doc-distance} on an object $A$ is an element  $\dist  \in \Doc(A\times A)$ such that
\begin{center}
\begin{tabular}{ll} 
$\funit \order \DReIdx{\Doc}{\Delta_A}(\dist )$& (reflexivity)\\
$\dist \order\DReIdx{\Doc}{\ple{\pi_2,\pi_1}}(\dist )$& (symmetry)\\
$\DReIdx{\Doc}{\ple{\pi_1,\pi_2}}(\dist ) \fmul \DReIdx{\Doc}{\ple{\pi_2,\pi_3}}(\dist ) \order \DReIdx{\Doc}{\ple{\pi_1,\pi_3}}(\dist )$& (transitivity)
\end{tabular} 
\end{center}

\begin{exas}\label{ex:distance}
\begin{enumerate}

\item\label{ex:distance:1} 
For the syntactic doctrine $\Prop_{\TT}$ as in \refToExItem{runningmon1}{22}, 
a $\Prop_{\TT}$-distance over a context $\vec{\srt}$ is a formula $\dist(\vec{x},\vec{y})$
in the context $\vec{\srt}\times \vec{\srt}$ 
such that 
the following entailments are derivable: 
$\Lone\vdash \dist(\vec{x},\vec{x})$ and 
$\dist(\vec{x},\vec{y})\vdash \dist(\vec{y},\vec{x})$ and 
$\dist(\vec{x},\vec{y})\Ltensor \dist(\vec{y},\vec{z})\vdash \dist(\vec{x},\vec{z})$.
 
\item\label{ex:distance:0}
For the \pl doctrine $\ple{\PP_{[0,\infty]},+,0}$ as in \refToExItem{runningmon1}{1} (where $[0,\infty]$ is the Lawvere quantale \cite{LawvereF:metsgl})
a $\PP_{[0,\infty]}$-distance $\dist$ over a set $X$ is a metric  $\dist:X\times X\to [0,\infty]$, \ie  
$0\ge \dist(x,x)$ and 
$\dist(x,y)\geq \dist(y,x)$ and 
$\dist(x,y)+\dist(y,z)\geq \dist(x,z)$. 

\item\label{ex:distance:13} 
For the realisability doctrine $\mathcal{R}_A$ as in  \refToExItem{runningmon1}{12}, 
a $\mathcal{R}_A$-distance $\dist$ over a set $X$ is function $\fun{\dist}{X\times X}{\PP(|A|)}$ such that there are three elements in $|A|$ realising reflexivity, symmetry, transitivity of $\dist$. 

\item\label{ex:distance:444} 
For the doctrine $K_\WW$  as in \refToExItem{runningmon1}{44}, 
a $K_\WW$-distance over a set $X$ is a ternary relation $\dist\subseteq X\times X\times W$ such that,  
for every $x,y,z$ in $X$  and every $w,w_1,w_2$ in $W$, we have 
$\krpunit\le w$ implies $\ple{x,x,w}\in \dist$ and 
$\ple{x,y,w}\in\dist$ implies $\ple{y,x,w}\in \dist$ and 
$\ple{x,y,w_1}\in \dist$ and $\ple{y,z,w_2}\in \dist$ implies $\ple{x,z,w_1\krpmul w_2}\in \dist$.
\end{enumerate}
\end{exas}

\begin{prop}\label{prop:elem-dist}
Let \ple{\Doc,\fmul,\funit} be an elementary \pl doctrine. 
For every obect $A$ in the base of $\Doc$, 
the equality predicate $\delta_A$ is a $\Doc$-distance. 
\end{prop} 
\begin{proof} 
The first point of \refToDef{lelementary} establishes reflexivity of $\delta_A$. 
The second point gives transitivity when $X$ is $A$ and $\alpha$ is $\delta_A$. 
To show symmetry, 
note that $\ple{\pi_2,\pi_1}\circ \Delta_A=\Delta_A$, hence, from reflexivity we deduce 
$\funit_A \order \Doc\reidx{\Delta_A}(\delta_A) = \Doc\reidx{\Delta_A}(\Doc\reidx{\ple{\pi_2,\pi_1}}(\delta_A))$ and 
and from the adjunction $\Ex\reidx{\Delta_A}\dashv\DReIdx{\Doc}{\Delta_A}$ and \refToProp{equivalenza} we conclude 
$\delta_A = \Ex\reidx{\Delta_A}(\funit_A) \order  \DReIdx{\Doc}{\ple{\pi_2,\pi_1}}(\delta_A)$.
\end{proof}

As already mentioned in the introduction, 
the direct rephrasing of rules for equality from First Order Logic to Linear Logic leads to a replicable equality predicate. 
This was already observed by Dozen~\cite{Hodges1996LogicFF} 
and it can be quickly seen in the following tree 
\[\begin{prooftree}
\[
  \[ \justifies \lcons{}{\trm\Leq\trm} \]
  \quad 
  \[ \justifies \lcons{}{\trm\Leq\trm} \] 
\justifies \lcons{}{\trm\Leq\trm \Ltensor \trm\Leq\trm} 
\]
\quad 
\[ \justifies \lcons{\trm\Leq\atrm}{\trm\Leq\atrm} \] 
\justifies \lcons{\trm\Leq\atrm}{\trm\Leq\atrm \Ltensor \trm\Leq\atrm} 
\end{prooftree} \]
by reflexivity of $\Leq$, we have $\lcons{}{\trm\Leq\trm}$, then, introducing $\Ltensor$ on the right, 
we get $\lcons{}{\trm\Leq\trm \Ltensor \trm\Leq\trm}$;
by the assumption $\lcons{\trm\Leq\atrm}{\trm\Leq\atrm}$, then, 
since $\trm\Leq\trm \Ltensor \trm\Leq\trm$ is equal to $\subst{(\trm\Leq\var \Ltensor \trm\Leq\var)}{\trm}{\var}$, using substitutivity we conclude 
$\lcons{\trm\Leq\atrm}{\trm\Leq\atrm \Ltensor \trm\Leq\atrm}$. 

Doctrines allows to prove this drawback in more generality. 
Given a \pl doctrine \ple{\Doc,\fmul,\funit}, an element $\alpha$ in $\Doc(X)$ is called \dfn{affine} if  $\alpha \order \funit$ and \dfn{replicable} if $\alpha\order \alpha\fmul\alpha$. 
Note that on affine and replicable elements the monoidal operation $\fmul$ behaves as a meet operation, hence, they can be deleted and duplicated. 

\begin{exas}\label{ex:affine}
\begin{enumerate}
\item\label{ex:affine:22} 
For the syntactic doctrine $\Prop_\TT$ over a theory $\TT$ in the $(\Ltensor,\Lone)$-fragment of Linear Logic, an element $\alpha$ in $\Prop_\TT(\vec{\srt})$ (\ie is a formula with free variables over the context $\vec{\srt}$) is affine when weakening ($\alpha\vdash\Lone$) is derivable in $\TT$, while $\alpha$ is replicable if contraction ($\alpha\vdash \alpha\Ltensor\alpha$) is derivable in $\TT$.

\item\label{ex:affine:1}
For the \pl doctrine $\ple{\PP_{[0,\infty]},+,0}$ as in \refToExItem{distance}{0} 
every element $\alpha:A\to [0,\infty]$ is affine as $0$ is also a top element, while $\alpha$ is replicable if and only if its values are either $0$ or $\infty$.
\end{enumerate}
\end{exas}

\begin{prop}\label{prop:tortellini} 
Let $\ple{\Doc,\fmul,\funit} $ be \pl and $\fun{f}{X}{Y}$ an arrow of the base. If $\DReIdx{\Doc}{f}$ has a left adjoint $\Ex_f$, then
$\Ex_f(\funit_X)$ is affine and replicable. 
\end{prop}
\begin{proof} 
We prove affineness by noting that 
$\Ex_{f}(\funit_X) = \Ex_{f}\DReIdx{\Doc}{f}(\funit_{Y})\le \funit_{Y}$. 
Replicability follows from 
$\funit_X = \funit_X\fmul\funit_X\order \DReIdx{\Doc}{f}(\Ex_{f}\funit_X)\fmul\DReIdx{\Doc}{f}(\Ex_{f}\funit_X) = \DReIdx{\Doc}{f}(\Ex_{f}(\funit_X)\fmul\Ex_{f}(\funit_X))$ and the adjunction between $\Ex_f$ and $\DReIdx{\Doc}{f}$. 
\end{proof}

 By \refToProp{equivalenza} it is $\delta_A=\Ex_{\Delta_A}(\funit_A)$. Affineness and replicability follows by \refToProp{tortellini}.

This shows that the standard approach based on left adjoints cannot support a quantitative notion of equality, which then needs to be defined in a different way.

%% file: graded-modality.tex

\section{Quantitative equality via graded modalities} 
\label{sect:graded} 

Graded modalities indexed by a semiring \RR are used in many linear calculi and type theories to enable resource sensitive reasoning \cite{ReedP10,BrunelGMZ14,GaboardiKOBU16,OrchardLE19,AbelB20,DalLagoG22,MoonEO21}. 
We introduce graded modalities in the framework of \pl doctrines to formulate a quantitative variant of elementary doctrines that we call
 \lip doctrines. We motivate the terminology after \refToCor{graded-lip}.
Then, we define a core deductive calculus for quantitative equality and its categorical semantics in \lip doctrines, proving it is sound and complete. 

\subsection{Graded and \lip doctrines}
\label{sect:graded-doc} 

A \dfn{resource semiring} is a semiring in the category \Pos, 
that is, a tuple $\RR=\ple{\RSet,\rord,\rplus,\rmul,\rzero,\rone}$, where 
\ple{\RSet,\rord} is a partially ordered set and $\rplus$ and $\rmul$ are monotone binary operations on it such that 
\begin{itemize}
\item \ple{\RSet,\rplus,\rzero} is a commutative monoid and \ple{\RSet,\rmul,\rone} is a monoid, 
\item 
for all $\res,\ares,\bres \in\RSet$ we have 
$\res\rmul(\ares\rplus\bres) = \res\rmul\ares \rplus \res\rmul\bres$, 
$(\ares\rplus\bres)\rmul\res = \ares\rmul\res \rplus \bres\rmul\res$, 
$\res\rmul\rzero = \rzero$ and 
$\rzero\rmul\res = \rzero$, 
\item $\rplus$ and $\rmul$ are monotone w.r.t.\ $\rord$. 
\end{itemize}

To define \RR-graded modalities, we adapt the approach in \cite{DagninoR21}, 
following the definition of graded linear exponential comonads \cite{BrunelGMZ14,GaboardiKOBU16,Katsumata18}. 

\begin{defi}\label{def:graded}
An \dfn{\RR-graded (linear exponential) modality} on a \pl doctrine \ple{\Doc,\fmul.\funit}  is an $\RSet$-indexed family of  natural transformations 
\TNat{\rmod{\res}}{\Doc}{\Doc} satisfying the axioms below for all objects $X$, $\alpha,\beta$ in $\Doc (X)$ and $\res,\ares$ in $\RSet$: 
\begin{enumerate}
\item\label{def:graded:1} $\funit  \order \rmod{\res}\funit $ \hfill (lax-monoidality-1) 
\item\label{def:graded:2} $\rmod{\res}\alpha \fmul  \rmod{\res}\beta \order \rmod{\res}(\alpha\fmul \beta)$ \hfill  (lax-monoidality-2) 
\item\label{def:graded:3} $\rmod{\rzero}\alpha \order \funit $  \hfill (weakening) 
\item\label{def:graded:4} $\rmod{\res\rplus\ares}\alpha \order \rmod{\res}\alpha \fmul \rmod{\ares}\alpha$ \hfill (contraction) 
\item\label{def:graded:5} $\rmod{\rone}\alpha\order\alpha$   \hfill (counit) 
\item\label{def:graded:6} $\rmod{\res\rmul\ares}\alpha \order \rmod{\res}(\rmod{\ares}\alpha)$ \hfill  (comultiplication) 
\item\label{def:graded:7} $\rmod{\res} \alpha \order \rmod{\ares} \beta$ if $\ares\rord\res$  \hfill  (contravariance) 
\end{enumerate}
An \dfn{\RR-graded (\pl) doctrine} is a \pl doctrine together with an \RR-graded modality on it. 
\end{defi}
\noindent 

Naturality ensures that 
$\rmod{\res}$ is monotone and stable under reindexing, for all $\res$.  
making the graded modality a graded linear exponential comonad on each fibre. 

We collect below some examples adapted from 
\cite{BreuvartP15,OrchardLE19,DalLagoG22}. 

\begin{exas}\label{ex:graded}
\begin{enumerate}

\item\label{ex:graded:12} 
The singleton set $\{\infty\}$ with the trivial order and operations is a resource semiring. 
An $\{\infty\}$-graded modality on 
$\Doc$ satisfies the following
 inequalities fibre-wise: 
$\rmod{\infty}\alpha \order \funit$, 
$\rmod{\infty} \alpha \order \rmod{\infty}\alpha \fmul \rmod{\infty}\alpha$, 
$\rmod{\infty}\alpha \order \alpha$ and 
$\rmod{\infty}\alpha \order \rmod{\infty}\rmod{\infty}\alpha$. 
Hence, $\rmod{\infty}$ models the usual bang modality of Linear Logic \cite{Girard87}. 
Syntactic doctrines on the $(\Ltensor, \Lone, \Lbang)$-fragment of Linear Logic are $\{\infty\}$-graded.
 
\item\label{ex:graded:2} 
Let $\NN$ abbreviate the semiring $\ple{\N,+,\cdot,0,1}$ of natural numbers and $\NN_\infty$ its extension by a new point $\infty$ with 
$\infty+x = x+\infty = \infty\cdot x = x \cdot \infty  =\infty$.
Denote by $\NN^=$ and $\NN^=_\infty$ the semirings $\NN$ and $\NN_\infty$ ordered by the equality relation.  
If $\rmod{}$ is an $\NN^=$-graded modality on $\Doc$, since $n = 1 + \ldots + 1$ ($n$ times),  contraction and counit imply
$\rmod{n} \alpha \order \alpha \fmul \ldots\fmul\alpha$ ($n$ times). 
Hence, $\rmod{n}\alpha$ provides exactly $n$ copies of $\alpha$.
Syntactically, this corresponds to the \emph{exact usage} modality $\Lbang_n$ of Bounded Linear Logic \cite{GirardSS92}. Indeed, 
syntactic doctrines built out of it are $\NN^=$-graded. 
If we consider $\NN^=_\infty$, we get the additional modality $\rmod{\infty}$ which behaves as in Item \ref{ex:graded:12}, modelling unrestricted usage. 

\item\label{ex:graded:0}
Let $\RRPos = \ple{\RPos,\le,+,\cdot,0,1}$ be the semiring of non-negative real numbers, $[0,\infty]$ be the Lawvere quantale \cite{LawvereF:metsgl} and $\ple{\PP_{[0,\infty]},+,0}$ be a doctrine as in \refToExItem{runningmon1}{1}. 
For  $\alpha \in \PP_{[0,\infty]}(X)$ and $\res\in\RPos$, the assignment $\alpha\mapsto \textbf{m}_{\res}\alpha$, with  $\textbf{m}_{\res}\alpha(x) = \res \cdot \alpha(x)$ is an $\RRPos$-graded modality on  $\PP_{[0,\infty]}$.
Here  $\res\cdot \infty=\infty$ when $\res\ne 0$ and $0\cdot\infty = 0$.

\item\label{ex:graded:1} 
Consider an ordered $\RR$-Linear Combinatory Algebra (\RR-LCA) \cite{Atkey18}, that is, a BCI algebra $A$ (\cf \refToExItem{runningmon1}{12}) with  
a function \fun{!_\res}{|A|}{|A|} for $\res\in \RSet$ and elements $K, W_{\res,\ares}, D, d_{\res,\ares}, F_\res$ and $O_{\res,\ares}$, with $\res\rord\ares$,  of $|A|$ satisfying
\[\begin{array}{ll} 
K\app x\app !_\rzero y =x 
& \quad 
W_{\res,\ares}\app x\app !_{\res\rplus\ares} y = x\app !_{\res}y\app !_{\ares}y 
\\
D\app !_{\rone}x=x   
& \quad 
d_{\res,\ares}\app !_{\res\rmul\ares} x = !_{\res}!_{\ares}x 
\\ 
F_{\res}\app !_{\res}x\app !_{\res}y = !_{\res}(x\app y)  
& \quad 
O_{\res,\ares}\app !_{\ares}x = !_{\res}x
\end{array}\] 
The realisability doctrine $\mathcal{R}_A$ (\cf \refToExItem{runningmon1}{12}) is \RR-graded: 
for $\alpha\in \PP(|A|)^X$  and $\res\in \RSet$, define $\rmod{\res}\alpha = x \mapsto \{!_\res a\mid a \in \alpha(x)\}$. 

\item\label{ex:graded:444} 
Let $\RR$ be a resource semiring and $\WW = \ple{W,\le,\krpmul,\krpunit}$ a monoidal Kripke frame as in \refToExItem{runningmon1}{44}.  
Recall from \cite{DalLagoG22} that a lax action of $\RR$ on $\WW$ is a monotone map $\fun{a}{\RSet \times W}{W}$ such that 
\[\begin{array}{ll} 
a_X(\res,\krpunit)\le \krpunit
& \quad
a_X(\res, w_1\krpmul w_2)\le a_X(\res,w_1) \krpmul a_X(\res,w_2) \\
\krpunit\le a(\rzero,w)
& \quad 
a(\res,w)\krpmul a(\ares,w)\le a(\res\rplus\ares,w)\\
w\le a(\rone,w)
& \quad  
a(\res,a(\ares, w)) \le a(\res\rmul\ares,w) 
\end{array}\] 
This gives rise to the $\RR$-graded modality $\overline{a}$ on the \pl doctrine $K_\WW$ of \refToExItem{runningmon1}{44} 
 whose component at $X$ maps $U\in K_\WW(X)$ to $\overline{a}_\res U = \{\ple{x,w}\mid \exists_v.\, a(\res,v)\le w, \ple{x,v}\in U\}$.
\end{enumerate}
\end{exas}

\RR-graded doctrines  are the objects of the 2-category \GMD{\RR}  where a 1-arrow from \ple{\Doc,\rmod{}} to \ple{\aDoc,\armod{}} is a 1-arrow \oneAr{\ple{F,f}}{\Doc}{\aDoc} in \MD preserving modalities, that is, for all $\res\in\RSet$ and object $X$ in the base of \Doc, it holds
$ {\armod{\res}}_{FX} \circ f_X = f_X \circ {\rmod{\res}}_X $; while 
2-arrows, compositions and identities are those of \MD. 

\refToProp{graded-primary} establishes
the inclusion  $\PD \hookrightarrow \GMD{\RR}$.  

\begin{prop}\label{prop:graded-primary}
Let \ple{\Doc,\fmul,\funit} be a \pl doctrine and \TNat{\rmod{\res}}{\Doc}{\Doc} be the identity for all $\res\in\RSet$. 
Then, \Doc is a primary doctrine iff $\rmod{}$ is an \RR-graded modality on it. 
\end{prop}
\begin{proof}
Let $X$ be an object in the base category of \Doc and $\alpha$ an element in $\Doc (X)$. 
The left-to-right implication is easy to check: 
the non-trivial axiom to verify is weakening, indeed, 
we have $\rmod{\rzero}\alpha = \alpha \order \funit$, as $\funit$ is the top element of $\Doc (X)$. 
For the converse we have to show that $\funit$ is the top element of $\Doc (X)$ and $\alpha\order\alpha\fmul\alpha$ for all $\alpha\in\Doc (X)$. 
By weakening $\alpha = \rmod{\rzero}\alpha \order \funit$ and 
contraction gives $\alpha = \rmod{\res\rplus\res} \alpha \order \rmod{\res}\alpha \fmul \rmod{\res}\alpha = \alpha\fmul\alpha$. 
\end{proof} 

A quantitative equality for an \RR-graded doctrine \ple{\Doc,\rmod{}} is then defined to be a $\Doc$-distance satisfying an \emph{\RR-graded substitutivite property} as detailed below: 

\begin{defi}\label{def:affine-lip} 
An \dfn{\RR-\lip doctrine} is a triple \ple{\Doc,\rmod{},\delem} where 
\ple{\Doc,\rmod{}} is an \RR-graded doctrine and, 
for each object $A$ in the base, $\delem_A$ is a \Doc-distance on $A$ 
such that, for all objects $A$ and $X$ and $\alpha$ in $\Doc(X\times A)$
\begin{enumerate}[label=(\alph*)] 
\item\label{def:affine-lip:1} there is an $\res$ in $\RSet$ such that 
\[\DReIdx{\Doc}{\ple{\pi_1,\pi_2}}(\alpha)\fmul \rmod{\res}\DReIdx{\Doc}{\ple{\pi_2,\pi_3}}(\delem_A)\order \DReIdx{\Doc}{\ple{\pi_1,\pi_3}}(\alpha);\]
\item\label{def:affine-lip:2}$\delem_{A\times X} = \DReIdx{\Doc}{\ple{\pi_1,\pi_3}}(\delem_A)\fmul\DReIdx{\Doc}{\ple{\pi_2,\pi_4}}(\delem_X)$;
\item\label{def:affine-lip:3}$\delem_{A\times X} \order \DReIdx{\Doc}{\ple{\pi_1,\pi_3}}(\delem_A)$ and $\delem_{A\times X} \order \DReIdx{\Doc}{\ple{\pi_2,\pi_4}}(\delem_X)$; 
\item\label{def:affine-lip:4}$\delem_1 = \funit$. 
\end{enumerate}
\end{defi}

The key difference between elementary and \RR-\lip doctrines is the substitutive property, which, taking advantage of graded modalities, in the latter is  stated in a resource sensitive way. 
Indeed, to prove a substitution, we need to have \emph{enough equality resources}. 

Since symmetry and transitivity are no longer derivable from substitutivity, for \lip doctrines we assumed $\delem_A$ to be a $\Doc$-distance (and not only a reflexive relation), as these are natural properties for equality.  
 
The product $A\times B$ mimics the concatenation of the context $A$ with the context $B$, hence axioms \ref{def:affine-lip:2}, \ref{def:affine-lip:3} and \ref{def:affine-lip:4}  encode the \emph{independence of contexts}. 
It allows the independent use of equalities in a product.

In \refToSect{comonad} we will describe a construction producing a \lip doctrine out of any graded doctrine, 
which will provide us with several examples. 
For the moment we consider only the following one: 

\begin{exas}\label{ex:liplip}
Denote by $\LIP$ the category of metric spaces (whose metrics take values in $[0,\infty]$) and \lip continuous functions. 
For a metric space $A$, denote by $|A|$ the underlying set and by $\delem_A$ the metric of $A$. 
Note that $\LIP$ has finite products where 
$|A\times B| = |A|\times |B|$ and $\delem_{A\times B}(x,y,x',y') = \delem_A(x,x')+\delem_B(y,y')$. 
Endow the Lawvere quantale $[0,\infty]$ with the Euclidean metric and let $\LIPDoc(A)$ be the set of \lip continuous functions from a metric space $A$ to $[0,\infty]$. 
Each $\LIPDoc(A)$ is an ordered commutative monoid where the order and the operations are defined pointwise. 
If $\fun{f}{A}{B}$ and $\fun{\alpha}{ B}{[0,\infty]}$ are \lip continuous, denote by $\LIPDoc(f)(\alpha)$ the composition $\fun{\alpha f}{A}{[0,\infty]}$. 
Then $\fun{\LIPDoc}{\LIP\op}{\Pos}$ is a \pl doctrine. 
Note that $\delem_A\in\LIPDoc(A\times A)$ as the inequality 
$\delem_A(x,x')+\delem_A(y,y') \ge |\delem_A(x',y')-\delem_A(x,y)|$ holds. 
Endow $\LIPDoc$ with the $\RRPos$-graded modality \textbf{m} which acts as in \refToExItem{graded}{0}, \ie $\textbf{m}_{\res}\alpha(x) = \res \cdot \alpha(x)$, 
then $\ple{\LIPDoc, \textbf{m}, \delem}$ is a \lip doctrine: indeed every $\alpha\in\LIPDoc(X\times A)$ is \lip continuous, hence there is $\res\in\RPos$ such that 
$\res\cdot\delem_{X\times A}(x,a,x',a') \ge |\alpha(x',a')-\alpha(x,a)|$ that, by reflexivity of $\delem_X$, implies 
$\res\cdot\delem_A(a,a') \ge |\alpha(x,a')-\alpha(x,a)|$ and this is equivalent to 
$\alpha(x,a)+\res\cdot \delem_A(a,a') \ge \alpha(x,a')$. 

\end{exas}

Axiom \ref{def:affine-lip:1}, describing substitution, can be equivalently rephrased as follows.

\begin{prop}\label{prop:cara} 
Let \ple{\Doc,\rmod{}}  be an \RR-graded doctrine and $\delem$ a family of \Doc-distances. 
The following are equivalent:
\begin{enumerate}
\item  \ple{\Doc,\rmod{},\delem}   is \RR-\lip;
\item \ple{\Doc,\rmod{},\delem} satisfies axioms of \refToDef{affine-lip} where \ref{def:affine-lip:1}  is replaced by the following:
\begin{enumerate}[label=(\alph*')]
\item\label{def:affine-lip:1p} 
for all objects $A$ and $\alpha$ in $\Doc(A)$, there is $\res\in\RSet$ such that 
\[\DReIdx{\Doc}{\pi_1}(\alpha) \fmul \rmod{\res}\delem_A \order \DReIdx{\Doc}{\pi_2}(\alpha).\] 
\end{enumerate} 
\end{enumerate}
\end{prop}

\begin{proof} 1)$\Rightarrow$2) take as $X$ the terminal object of $\CC$ to get (a'). 2)$\Rightarrow$1)  Take $\alpha$ in $\Doc(X\times A)$, then $\DReIdx{\Doc}{\ple{\pi_1,\pi_2}}(\alpha) \fmul \rmod{\res}\delem_{X\times A} \order \DReIdx{\Doc}{\ple{\pi_3,\pi_4}}(\alpha)$. By (b) one has $\DReIdx{\Doc}{\ple{\pi_1,\pi_2}}(\alpha) \fmul \rmod{\res}\DReIdx{\Doc}{\ple{\pi_2,\pi_4}} \delem_{A} \order \DReIdx{\Doc}{\ple{\pi_3,\pi_4}}(\alpha)$. Reindexing along $\ple{\pi_1,\pi_2,\pi_1,\pi_3}:X\times A\times A\to X\times A\times X\times A$ completes the proof.
\end{proof}

Axioms \ref{def:affine-lip:3} and \ref{def:affine-lip:4} of \refToDef{affine-lip} are equivalent to affineness of equality.
\begin{prop}\label{prop:cara2} 
Let \ple{\Doc,\rmod{}}  be an \RR-graded doctrine and $\delem$ be a family of \Doc-distances. 
The following are equivalent:
\begin{enumerate}
\item  \ple{\Doc,\rmod{},\delem}   is \RR-\lip;
\item $\delem$ satisfies the following axioms: 
\begin{enumerate}[label=(\alph*'')]
\item axioms \ref{def:affine-lip:1} and \ref{def:affine-lip:2} of \refToDef{affine-lip} hold,
\item $\delem_A$ is affine.
\end{enumerate} 
\end{enumerate}
\end{prop}

\begin{proof}
1)$\Rightarrow$2) follows from 
$\delem_{A\times 1} = \DReIdx{\Doc}{\ple{\pi_1,\pi_3}}(\delem_A)$ and
$\delem_{A\times 1}\order \delem_1 = \funit$. 2)$\Rightarrow$1) is immediate.
\end{proof}

The quantitative nature of this notion of equality can be noticed also from its relationship with reindexing, described by the following proposition. 

\begin{prop}\label{prop:graded-lip0}
Let \ple{\Doc,\rmod{},\delem} be an \RR-\lip doctrine. 
For any \fun{f}{A_1\times\ldots\times A_n}{B} in the base category of \Doc, there are $\res_1,...,\res_n\in\RSet$ such that 
\[\rmod{\res_{1}}\DReIdx{\Doc}{\ple{\pi_1,\pi_{n+1}}}(\delem_{A_1})\fmul...\fmul\rmod{\res_n}\DReIdx{\Doc}{\ple{\pi_n,\pi_{2n}}}(\delem_{A_n})\order \DReIdx{\Doc}{f\times f}(\delem_B)\] 
\end{prop}
\begin{proof}
For all $i \in 1..n$, let $\vec{\pi}_i = \ple{\pi_i,\ldots,\pi_n}$. 
By  \refToDefItem{affine-lip}{1}
we know that, for all $i \in 1..n$, there is $\res_i \in \RSet$ such that 
\[
\DReIdx{\Doc}{\ple{\vec\pi_1^{i-1},\vec\pi_i^n,\vec\pi_{n+1}^{n+i-1},\vec\pi_i^n}}(\DReIdx{\Doc}{f\times f}(\delem_B)) \fmul  \rmod{\res_i} (\DReIdx{\Doc}{\ple{\pi_i,\pi_{n+i}}}(\delem_{A_i})) 
\order 
\DReIdx{\Doc}{\ple{\vec\pi_1^i,\vec\pi_{i+1}^n,\vec\pi_{n+1}^{n+i},\vec\pi_{i+1}^n}}(\DReIdx{\Doc}{f\times f}(\delem_B)) 
\]
holds in $A_1\times\ldots\times A_n\times A_1\times\ldots\times A_n$, where $\vec\pi_k^h$, with $1\le k\le h+1\le 2n+1$, denotes the sequence of projections $\pi_k,\pi_{k+1},\ldots,\pi_{h-1},\pi_h$. 
Combining these inequalities, using monotonicity of $\fmul$ and transitivity of $\order$, we get 
\[
\DReIdx{\Doc}{\ple{\pi_1,\ldots,\pi_n,\pi_1,\ldots,\pi_n}}(\DReIdx{\Doc}{f\times f}(\delem_B)) \fmul 
\rmod{\res_1}(\DReIdx{\Doc}{\ple{\pi_1,\pi_{n+1}}}(\delem_{A_1})) \fmul \ldots \fmul \rmod{\res_n}(\DReIdx{\Doc}{\ple{\pi_n,\pi_{2n}}}(\delem_{A_n})) \order 
\DReIdx{\Doc}{f\times f}(\delem_B) 
\]
Note that $\DReIdx{\Doc}{\ple{\pi_1,\ldots,\pi_n, \pi_1,\ldots,, \pi_n}}(\DReIdx{\Doc}{f\times f}(\delem_B)) =  \DReIdx{\Doc}{f\circ\ple{\pi1,\ldots,\pi_n}}(\DReIdx{\Doc}{\Delta_B}(\delem_B))$. 
Then, the thesis follows because $\funit \order \DReIdx{\Doc}{f\circ\ple{\pi_1,\ldots,\pi_n}}(\DReIdx{\Doc}{\Delta_B}(\delem_B))$ holds by reflexivity of $\delem_B$. 
\end{proof}

\begin{cor}\label{cor:graded-lip}
Let \ple{\Doc,\rmod{},\delem} be an \RR-\lip doctrine. 
For any \fun{f}{A}{B} in the base category of \Doc, there exists $\res\in\RSet$ such that 
$\rmod{\res}\delem_A\order \DReIdx{\Doc}{f\times f}(\delem_B)$. 
\end{cor}

\refToCor{graded-lip} describes a condition similar to the \lip condition on maps between metric spaces and motivates the name in \refToDef{affine-lip}.

Finally, also 
\RR-\lip doctrines can be organised into a 2-category \EGMD{\RR}, whose objects are \RR-\lip doctrines, 
a 1-arrow from \ple{\Doc,\rmod{},\delem} to \ple{\aDoc,\rmod{}',\delem'} is 1-arrow \oneAr{\ple{F,f}}{\ple{\Doc,\rmod{}}}{\ple{\aDoc,\rmod{}'}}   preserving the distance, that is, such that, for each object $A$ in the base category of \Doc ,  it holds $\delem'_{FA} = f_{A\times A}(\delem_A)$, and 
a 2-arrow from \ple{F,f} to \ple{G,g} is a 2-arrow \twoAr{\theta}{\ple{F,f}}{\ple{G,g}} in \GMD{\RR}.

%% file: new-graded-logic.tex

\subsection{A deductive calculus for quantitative equality} 
\label{sect:grade-syntax} 

So far we have seen that theories $\TT$ over a signature $L$ generate syntactic doctrines $\Prop_\TT$. 
This section  describes a  calculus whose associated syntactic doctrine is \RR-\lip. 

Assume $\RR = \ple{\RSet,\rord,\rplus,\rmul,\rzero,\rone}$ is a resource semiring and $L$ a first order signature. 
A judgment has the form $\lcons[\vec{\srt}]{\Gamma}{\form}$ where $\Gamma=\form_1,\ldots,\form_k$ is a finite multiset and $\form_1,\ldots,\form_k,\form$ are formulas over $L$ in the context $\vec{\srt}=\ple{x_1:\srt_1,\ldots,x_n:\srt_n}$. 
We will often omit the context $\vec{\srt}$ to improve readability. 

We first introduce \RPLL, a calculus extending the $(\Ltensor,\Lone)$-fragment of Linear Logic  by a family of modalities $\Lbang_\res$, for $\res\in\RSet$. 
The front `P' is for primary as its associated syntactic doctrine is a \pl doctrine (endowed with an \RR-graded modality).

Rules of \RPLL are a subset of the calculus presented in \cite{BreuvartP15}, in \refToFig{graded} we report only rules for modalities as structural rules and those for $\Ltensor$ and $\Lone$ can be found in \refToFig{frammento}. 
The first four rules are the graded variant of standard rules for the bang modality of Linear Logic. 
Rules \rn{w} and \rn{c} encode graded structural rules: by weakening we can add formulas marked as not used (labelled by $\rzero$) and contraction  tracks the usage by addition. 
Rule \rn{der} tells that hypotheses with grade $\rone$ can be treated as linear hypotheses and 
rule \rn{pro} introduces $\Lbang_\res$, scaling by $\res$ the grades of the hypotheses. 
Rule \rn{decr} allows to approximate the usage of a formula following (contravariantly) the order of the semiring. 

\begin{figure}[t] 
\begin{mathpar}
\Infer[w]{
  \lcons{\hps}{\form} 
}{ \lcons{\hps,\Lbang_\rzero \aform}{\form} } 
\and 
\Infer[c]{
  \lcons{\hps,\Lbang_\res\aform,\Lbang_\ares\aform}{\form} 
}{ \lcons{\hps,\Lbang_{\res\rplus\ares} \aform}{\form} } 
\and  
\Infer[der]{
  \lcons{\hps,\aform}{\form} 
}{ \lcons{\hps,\Lbang_\rone \aform}{\form} } 
\and 
\Infer[pro]{
  \lcons{\Lbang_{\ares_1}\aform_1,\ldots,\Lbang_{\ares_n}\aform_n}{\form} 
}{ \lcons{\Lbang_{\res\rmul\ares_1}\aform_1,\ldots,\Lbang_{\res\rmul\ares_n}\aform_n}{\Lbang_\res\form} } 
\and 
\Infer[decr]{
  \lcons{\hps}{\Lbang_\ares\form} 
}{ \lcons{\hps}{\Lbang_\res\form} }\ \res\rord\ares 
\end{mathpar} 
\caption{Rules for graded bang modality}\label{fig:graded} 
\end{figure} 

To extend \RPLL by a quantitative equality, we need to add a resource sensitive substitution rule. 
Hence, we need a way to compute for each formula $\form$ and variable $x$ the cost (represented as an element of the semiring) of substituting $x$  in $\form$. 
To this end, we first enrich the notion of signature. 

An \emph{\RR-graded signature} is a first order signature $L$ where symbols have an \emph{\RR-graded arity}, that is, 
an assignment $\gar{\blank}$ as the following: 
\begin{align*} 
\gar{f}= \ple{\res_1,\srt_1},\ldots,\ple{\res_n,\srt_n}\to \asrt && 
\gar{p}= \ple{\res_1,\srt_1},\ldots,\ple{\res_n,\srt_n}
\end{align*} 
where  $f$ and $p$ are a function and a predicate symbol, respectively, and, 
$\srt_i$ is the sort of the $i$-th argument, while  $\res_i\in\RSet$ says how much the cost of substituting a variable in the $i$-th position is amplified. 
We write $\gar{f}_i$ (resp. $\gar{p}_i$) in place of $\res_i$ in the assignments above. 

Denote by $\Var$ and $\Trm$ the sets of variables and terms inductively constructed from symbols in $L$ in the usual way (using the standard arity obtained from the graded one by erasing resources). 
Resources in the graded arity determines a function $\fun{\gr}{\Trm\times\Var}{\RSet}$ as follows: 
\begin{itemize}
\item $\gr(z,x)=\rzero$ if $x\ne z$
\item $\gr(x,x)=\rone$
\item $\gr(f(\trm_1,\ldots,\trm_n),x) = \gar{f}_1\rmul\gr(\trm_1,x) \rplus \ldots \rplus\gar{f}_n\rmul\gr(\trm_n,x)$
\end{itemize}
Intuitively, $\gr(\trm,\var)$ represents the cost of substituting the variable $\var$ inside the term $\trm$.
Note that $\gr(\trm,x)$  depends on the number of occurrences of $x$ in $\trm$. 
For instance, 
if $f_0$, $f_1(x)$,\ldots, $f_n(x,x,\ldots,x)$ are terms where 
$\gar{f}_i = \rone$ for all $i \in 0..n$, then 
$\gr(f_i(x,\ldots,x),x) = \rone\rplus\ldots\rplus\rone$ ($i$ times). 
In particular, for constants, \ie function symbols $f$ with no arguments, we always have $\gr(f,x)= \rzero$. 

Denote by $\Wff$ the set of well-formed formulas constructed from symbols in $L$, the equality symbol $\Leq$ and using $\Ltensor$, $\Lone$ and $\Lbang_\res$ as connectives. 
We extend the function $\gr$ to a function $\fun{\gr}{\Wff\times\Var}{\RSet}$ mapping a formula $\form$ and a variable $x$ to the amount of resources needed to substitute $x$ in $\form$: 
\begin{itemize}
\item $\gr(p(\trm_1,\ldots,\trm_n),x) = \gar{p}_1\rmul\gr(\trm_1,x)\rplus\ldots\rplus \gar{p}_n\rmul\gr(\trm_n,x))$, 
\item $\gr(\trm\Leq_{\asrt}\atrm,x) = \gr(\trm,x) \rplus \gr(\atrm,x)$, 
\item $\gr(\form\Ltensor\aform,x)=\gr(\form,x)\rplus\gr(\aform,x)$, 
\item $\gr(\Lone,x) = \rzero$, 
\item $\gr(\Lbang_\res\form,x) = \res\rmul\gr(\form,x)$. 
\end{itemize}
Similarly to terms, the cost for substituting the variable $x$ in a formula $\form$ depends on the number of occurrences of $x$ in $\form$. 
For instance, we have 
$\gr(p(x)\Ltensor p(y),x) = \gar{p}_1$, while 
$\gr(p(x)\Ltensor p(x),x) = \gar{p}_1 \rplus \gar{p}_1$. 
Again, for a predicate symbol $p$ with no arguments, we have 
$\gr(p,x) =\rzero$. 

The calculus \RPLL extended by rules in 
\refToFig{graded-eq} will be called \ERPLL, where the first `L' is for \lip.  

\begin{figure}[t]
\begin{mathpar}
\Infer[r]{ }{ \lcons{}{\trm\Leq_\asrt\trm} }
\and 
\Infer[s]{
  \lcons{\hps}{\trm\Leq_\asrt\atrm} 
}{ \lcons{\hps}{\atrm\Leq_\asrt\trm} }
\and 
\Infer[t]{
  \lcons{\hps}{\trm\Leq_\asrt\atrm} 
  \quad 
  \lcons{\ahps}{\atrm\Leq_\asrt\btrm} 
}{ \lcons{\hps,\ahps}{\trm\Leq_\asrt\btrm} }
\and 
\Infer[w-eq]{
  \lcons{\hps}{\form} 
}{ \lcons{\hps,\trm\Leq_\asrt\atrm}{\form} } 
\and 
\Infer[subst]{
  \lcons{\hps}{\subst{\form}{\trm}{\var}}
  \quad 
  \lcons{\ahps}{\Lbang_{\gr(\form,x)} \trm\Leq_\asrt\atrm} 
}{ \lcons{\hps,\ahps}{\subst{\form}{\atrm}{\var}} } 
\end{mathpar}
\caption{Rules for graded equality.}\label{fig:graded-eq} 
\end{figure}

We have  rules for reflexivity, symmetry and transitivity of $\Leq_\srt$ and a rule for weakening. 
Note the substitution rule: a substitution is derivable only if \emph{enough equality resources} are available. These are determined by the function $\gr$. 

\begin{rem}\label{rem:lip-syntax} 
If $\gar{f}=\ple{\gar{f}_1,\srt_1}\ldots\ple{\gar{f}_n,\srt_n}\to \asrt$ is the arity of the function symbol $f$, then 
one can prove the entailment 
\[
\lcons{
  \Lbang_{\gar{f}_1} (x_1\Leq_{\srt_1} y_1),\ldots, \Lbang_{\gar{f}_n} (x_n\Leq_{\srt_n} y_n)
}{ f(\vec{x}) \Leq_\asrt f(\vec{y}) }
\] 
where $\vec{x}=x_1,\ldots,x_n$ and $\vec{y}=y_1,\ldots,y_n$. 
Interpreting the equality as a distance means that the application of the function $f$ amplifies the distance between $x_i$ and $y_i$ by a factor $\gar{f}_i$.  
Similarly, if $p$ is a predicate symbol of arity $\gar{p} = \ple{\gar{p}_1,\srt_1}\ldots\ple{\gar{p}_n,\srt_n}$, 
one can derive the entailment 
\[
\lcons{ 
  p(\vec{x}) , \Lbang_{\gar{p}_1} (x_1\Leq_{\srt_1} y_1),\ldots, \Lbang_{\gar{p}_n} (x_n\Leq_{\srt_n} y_n) 
}{ p(\vec{y}) }
\]
meaning that one has to amplify equality between $x_i$ and $y_i$ by $\gar{p}_i$ to derive 
$p(\vec{y})$ from $p(\vec{x})$. 
\end{rem} 
\begin{prop}\label{prop:dottrinasintattica}
Let $\TT$ be a theory in \ERPLL over the \RR-graded signature $L$. 
The syntactic doctrine $\fun{\Prop_\TT}{\ct{Cxt}_L\op}{\Pos}$ 
is an \RR-\lip doctrine with \RR-graded modality given by $\Lbang$ and a family of distances $\delem^\Leq$ inductively defined by
\[\delem^\Leq_{\ple{}} = \Lone\quad\quad\quad\quad\delem^\Leq_{\ple{\vec{\sigma},x_{n+1}:\srt_{n+1}}} = \delem^\Leq_{\vec{\sigma}} \Ltensor (x_{n+1} \Leq_{\srt_{n+1}} x'_{n+1})\]
\end{prop}
\begin{proof}
To check that $\Prop_\TT$ is \RR-\lip we first need to prove that it is \RR-graded where 
$\Ltensor$ and $\Lone$ give the \pl structure and $\Lbang$ the \RR-graded modality. 
The former is known (\cf\refToExItem{runningmon1}{22}), while to prove the latter 
we need to show that rules in \refToFig{graded} suffice to derive that $\Lbang$ satisfies axioms listed in \refToDef{graded}. 
Naturality of $\Lbang_\res$ holdes as it commutes with subsitution, while its monotonicity easily follows by rules \rn{der} and \rn{pro}. 
Moreover, weakening, contraction, counit, comultiplication and controvariance follows immediately from rules $\rn{w}$, $\rn{c}$, $\rn{der}$, $\rn{pro}$ and $\rn{decr}$. 
To show lax-monoidality-2, 
we first derive $\Lbang_\rone\alpha,\Lbang_\rone\beta\vdash\alpha\Ltensor\beta$ from $\alpha,\beta\vdash\alpha\Ltensor\beta$, using rule \rn{der} twice. 
Then,  by rule $\rn{pro}$, we get $\Lbang_{\res}\alpha, \Lbang_{\res}\beta\vdash  \Lbang_{\res} (\alpha \Ltensor \beta)$.  
Finally, then introducing $\Ltensor$ on the left, we get $\Lbang_\res\alpha\Ltensor\Lbang_\res\beta\vdash \Lbang_\res(\alpha\Ltensor\beta)$, as needed. 
The proof of lax-monoidality-1 is similar: 
from $\vdash\Lone$ we derive $\vdash \Lbang_\res\Lone$ using \rn{pro} and we conclude introducing $\Lone$ on the left. 

To prove  that $\Prop_\TT$ is \RR-\lip, we have to show that rules in \refToFig{graded-eq} imply conditions in \refToDef{affine-lip}. 
Rule \rn{w-eq} ensures that $\delem^\Leq_{\vec\srt}$ is affine for every context $\vec\srt$, hence, by \refToProp{cara2}, it suffices to check only axioms (a) and (b) of \refToDef{affine-lip}. 
Since products in $\ct{Cxt}_L$ are given by context concatenation, axiom (b) is straightforward, by definition of $\delem^\Leq$. 
Finally, by \refToProp{cara}, it remains to show that, for every context $\vec{\srt} = \ple{x_1:\srt_1,\ldots,x_n:\srt_n}$ and every formula $\alpha$ in $\Prop_\TT(\vec\srt)$, there is $\res$ in $\RSet$ such that 
the entailment 
$$\alpha(x_1,\ldots,x_n)\Ltensor \Lbang_\res\delem^\Leq_{\vec\srt} \vdash \alpha(y_1,\ldots,y_n)$$
in the context $\ple{x_1:\srt_1,\ldots,x_n:\srt_n,y_1:\srt_1,\ldots,y_n:\srt_n}$
is provable in $\TT$. 
To this end, 
first of all, note that, 
using rule \rn{subst}, for every $\Gamma$ and $i \in 1..n$, 
from $\Gamma\vdash\alpha(y_1,\ldots,y_{i-1},x_i,\ldots,x_n)$ and $\Lbang_{\res_i} x_i \Leq_{\srt_i}y_n\vdash\Lbang_{\res_i} x_i\Leq_{\srt_i} y_i$, 
we can derive 
$
\Gamma,\Lbang_{\res_i} x_i \Leq_{\srt_i} y_i  \vdash \alpha(y_1,\ldots,y_i,x_{i+1},\ldots,x_n)
$, 
where $\res_i = \gr(\alpha(y_1,\ldots,y_{i-1},x_i,\ldots,x_n),x_i)$. 
Therefore, starting from $\alpha(x_1,\ldots,x_n)\vdash \alpha(x_1,\ldots,x_n)$ and iteratively using this fact, we can prove
$$
\alpha(x_1,\ldots,x_n), \Lbang_{\res_1}x_1\Leq_{\srt_1}y_1,\ldots, \Lbang_{\res_n} x_n\Leq_{\srt_n} y_n \vdash \alpha(y_1,\ldots,y_n) 
$$
Now, observe that, 
using rule \rn{w-eq}, we have $\delem^\Leq_{\vec\srt} \vdash x_i \Leq_{\srt_i} y_i$, for all $i \in 1..n$, and, by monotonicity of $\Lbang_\ares$, we get $\Lbang_\ares \delem^\Leq_{\vec\srt} \vdash \Lbang_\ares x_i\Leq_{\srt_i}y_i$, for all $\ares \in \RSet$ and $i \in 1..n$. 
Hence, by cut, we derive 
$$
\alpha(x_1,\ldots,x_n), \Lbang_{\res_1}\delem^\Leq_{\vec\srt},\ldots,\Lbang_{\res_n}\delem^\Leq_{\vec\srt} \vdash \alpha(y_1,\ldots,y_n) 
$$
Let us set 
$\res = \res_1\rplus\ldots\rplus\res_n$. 
Then, by iteratively applying rule \rn{c} we get 
$\alpha(x_1,\ldots,x_n),\Lbang_\res \delem^\Leq_{\vec\srt} \vdash\alpha(y_1,\ldots,y_n)$, 
and so the thesis follows by introducing $\Ltensor$on the left. 
\end{proof}

\subsection{Semantics in \RR-\lip doctrines} 
\label{sect:grade-sem}

The interpretation of a theory in a doctrine is standard
\cite{PittsCL,JacobsB:catltt}: 
it maps contexts and terms to objects and arrows of the base and formulas to elements of the fibres,  respecting the entailments. The interpretation of  \ERPLL in an \RR-\lip doctrine has to be defined with the additional requirement that  $\gr$ agrees with the structure of the doctrine. 

Let $L$ be an \RR-graded signature. 
An \dfn{\RR-graded interpretation} of $L$ into a \RR-\lip doctrine \ple{\Doc,\rmod{},\delem} assigns 
to every sort $\srt$ an object $\sem{\srt}$ in the base of \Doc,
to every function symbol $f$ of arity $\ple{\gar{f}_1,\srt_1}\ldots\ple{\gar{f}_n,\srt_n}\to\asrt$ an arrow \fun{\sem{f}}{\sem{\srt_1}\times\ldots\times\sem{\srt_n}}{\sem{\asrt}} in the base of \Doc such that
\begin{equation}\label{agreement}
\prod_{i = 1}^n \rmod{\gar{f}_i}\DReIdx{\Doc}{\ple{\pi_i,\pi_{n+i}}}(\delem_{\sem{\srt_i}}) \order \DReIdx{\Doc}{\sem{f}\times\sem{f}}(\delem_{\sem{\asrt}})
\end{equation}
and to every predicate symbol $p$ of arity $\ple{\gar{p}_1,\srt_1}\ldots\ple{\gar{p}_n,\srt_n}$ an element $\sem{p}$ in $\Doc(\sem{\srt_1}\times\ldots\times\sem{\srt_n})$ such that 
\begin{equation}\label{agreement2}
\DReIdx{\Doc}{\pr_1}(\sem{p}) \fmul \prod_{i=1}^n  \rmod{\gar{p}_i}\DReIdx{\Doc}{\ple{\pi_i,\pi_{n+i}}}(\delem_{\sem{\srt_i}}) \order \DReIdx{\Doc}{\pr_2}(\sem{p})
\end{equation}
where $\pr_1 = \ple{\pi_1,\ldots,\pi_n}$ and $\pr_2 = \ple{\pi_{n+1},\ldots,\pi_{2n}}$. 

Intuitively, we can pick as interpretation of a function symbol $f$ an arrow $\sem{f}$ for which the cost of substituting the $i$-th argument  is $\gar{f}_i$, and similarly for predicate symbols. 
In other words, the graded arity of $f$ determines valid \lip constansts for its interpretation $\sem{f}$. 

\begin{exa}\label{ex:provvisorio}
We can interpret an $\RRPos$-graded signature $L$ in the \lip doctrine $\fun{\LIPDoc}{\LIP\op}{\Pos}$  presented in \refToEx{liplip}. 
The interpretation maps  
sorts to metric spaces, 
a function symbol $f$ of arity $\ple{\gar{f}_1,\srt_1}\ldots\ple{\gar{f}_n,\srt_n}\to \asrt$ to a \lip function \fun{\sem{f}}{\sem{\srt_1}\times\ldots\sem{\srt_n}}{\sem{\asrt}} whose \lip constant is $\gar{f}_1+\ldots +\gar{f}_n$ and 
a predicate symbol $p$ of arity $\ple{\gar{p}_1,\srt_1}\ldots\ple{\gar{p}_n,\srt_n}$ to a \lip function \fun{\sem{p}}{\sem{\srt_1}\times \ldots \times \sem{\srt_n}}{[0,\infty]} whose \lip constant is $\gar{p}_1 + \ldots + \gar{p}_n$. 
\end{exa}

An \RR-graded interpretation of $L$ induces an interpretation of contexts and terms in the base of \Doc and 
an interpretation of formulas in the fibres of \Doc. 
More precisely, a context $\vec{\srt} = \ple{x_1:\srt_1,\ldots,x_n:\srt_n}$  is interpreted by 
the product $\sem{\vec{\srt}}_\Trm = \sem{\srt_1}\times\ldots\times \sem{\srt_n}$, 
a term $\trm$ of sort $\asrt$ in $\vec{\srt}$ by an arrow \fun{\sem{\trm}_\Trm}{\sem{\vec{\srt}}_\Trm}{\sem{\asrt}} and 
a formula $\form$ in $\vec{\srt}$ by an element $\sem{\form}_\Wff \in \Doc(\sem{\vec{\srt}}_\Trm)$, as follows: 
\begin{align*}
\sem{x_i}_\Trm &= \pi_i 
    & 
\sem{\Lone}_\Wff &= \funit 
    \\ 
\sem{f(\trm_1,\ldots,\trm_n)}_\Trm &= \sem{f} \circ \ple{\sem{\trm_1}_\Trm,\ldots,\sem{\trm_n}_\Trm} 
    & 
\sem{\form\Ltensor\aform}_\Wff &= \sem{\form}_\Wff \fmul \sem{\aform}_\Wff 
     \\
\sem{p(\trm_1,\ldots,\trm_n)}_\Wff &= \DReIdx{\Doc}{\ple{\sem{\trm_1}_\Trm,\ldots,\sem{\trm_n}_\Trm}}(\sem{p}) 
    & 
 \sem{\Lbang_\res \form}_\Wff &= \rmod{\res} \sem{\form}_\Wff
    \\
\sem{\trm\Leq_\asrt \atrm}_\Wff &= \DReIdx{\Doc}{\ple{\sem{\trm}_\Trm,\sem{\atrm}_\Trm}}(\delem_{\sem{\asrt}})
\end{align*}

An \RR-graded interpretation of a theory $\TT$ of \ERPLL over $L$, is an \RR-graded interpretation $\sem{\blank}$ of $L$ such that $\sem{\blank}_\Wff$ respects entailments in $\TT$. 

We will often omit subscripts $\Trm$ and $\Wff$ from $\sem{\blank}$ when these are clear. If $\Gamma$ is a list of formulas $\form_1,...,\form_n$ then $\sem{\Gamma}$ abbreviates $\sem{\form_1}\fmul...\fmul\sem{\form_n}$.

The following theorem states soundness for the semantics of \ERPLL. 

\begin{thm}\label{thm:sound}
Let $\sem{\blank}$ be an \RR-graded interpretation of a theory $\TT$ in \ERPLL over $L$ into an \RR-\lip \ple{\Doc,\rmod{},\delem}. Then, 
\oneAr{\ple{\sem{\blank}_\Trm,\sem{\blank}_\Wff}}{\ple{\Prop_\TT,\Lbang,\delem^\Leq}}{\ple{\Doc,\rmod{},\delem}} is a 1-arrow in \EGMD{\RR}. 
\end{thm} 
 
The proof is carried out by the usual induction on rules \cite{PittsCL,JacobsB:catltt}, relying on the following lemma: 

\begin{lem}\label{lem:sound}
For every  \RR-graded interpretation $\sem{\blank}$ of $L$ into \ple{\Doc,\rmod{},\delem} 
we have: 
\begin{enumerate}
\item for any term $\trm:\asrt$ in the context $\vec{\srt}=\ple{x_1:\srt_1,\ldots,x_n:\srt_n}$
\[\prod_{i = 1}^n \rmod{\gr(\trm, x_i)}\DReIdx{\Doc}{\ple{\pi_i,\pi_{n+i}}}(\delem_{\sem{\srt_i}}) \order \DReIdx{\Doc}{\sem{\trm}\times\sem{\trm}}(\delem_{\sem{\asrt}})\]
\item for any formula $\form$ in the context $\vec{\srt}=\ple{x_1:\srt_1,\ldots,x_n:\srt_n}$ 
\[\DReIdx{\Doc}{\pr_1}(\sem{\form}) \fmul \prod_{i = 1}^n \rmod{\gr(\form,x_i)}\DReIdx{\Doc}{\ple{\pi_i,\pi_{n+i}}}(\delem_{\sem{\srt_i}}) \order 
\DReIdx{\Doc}{\pr_2}(\sem{\form})\] 
where $\pr_1 = \ple{\pi_1,\ldots,\pi_n}$ and $\pr_2 = \ple{\pi_{n+1},\ldots,\pi_{2n}}$. 
\end{enumerate}
\end{lem}

\begin{proof}
1) is proved by induction on $\trm:\srt$. If $\trm:\srt$ is a variable in $\vec{\srt}$, \ie $\trm:\srt$ is $x_p:\srt_p$ for $1\le p\le n$, then $\sem{\trm}$ is the projection $\pi_p:\sem{\srt_1}\times....\times\sem{\srt_n}\to\sem{\srt_p}$ and $\gr(\trm,x_j)=\rone$ if $p=j$ and is $\rzero$ if $p\not=j$, so $\rmod{\gr(\trm,x_j)}\delem_{\sem{\srt_j}}\le\funit$ (for $j\not=p$ by weakening of $\rmod{}$) and $\rmod{\gr(\trm,x_p)}\delem_{\sem{\srt_p}}\le \delem_{\sem{\srt_p}}$ (by counit of $\rmod{}$), so 
$$\prod_{i=1}^n \DReIdx{\Doc}{\ple{\pi_i,\pi_{n+i}}}(\rmod{\gr(\trm, x_i)}\delem_{\sem{\srt_i}}) \order \DReIdx{\Doc}{\ple{\pi_p,\pi_{n+p}}}(\rmod{\gr(\trm,x_p)}\delem_{\sem{\srt_i}})$$
whence the claim as $\ple{\pi_p,\pi_{n+p}}=\pi_p\times \pi_p=\sem{\trm}\times\sem{\trm}$.

Suppose $\trm:\asrt$ is of the form $f(\trm_1,...,\trm_m)$ where $f$ is a function symbol of arity $\asrt_1,....,\asrt_m\to \asrt$, and  each $\trm_k$ is a term of type $\asrt_k$ in the context $\srt_1,...,\srt_n$, satisfying the inductive hypothesis, \ie

$$\prod_{i=1}^n \DReIdx{\Doc}{\ple{\pi_i,\pi_{n+i}}}(\rmod{\gr(\trm_k, x_i)}\delem_{\sem{\srt_i}}) 
\order \DReIdx{\Doc}{\sem{\trm_k}\times\sem{\trm_k}}(\delem_{\sem{\asrt_k}})$$

Multiply both sides by $\rmod{\gar{f}_k}$ and use comultiplication of $\rmod{}$ to get

$$
\prod_{i=1}^n
\DReIdx{\Doc}{\ple{\pi_i,\pi_{n+i}}}(\rmod{\gar{f}_k\rmul\gr(\trm_k, x_i)}\delem_{\sem{\srt_i}}) 
\order \DReIdx{\Doc}{\sem{\trm_k}\times\sem{\trm_k}}(\rmod{\gar{f}_k}\delem_{\sem{\asrt_k}})$$
From  $\sem{\trm_k}\times \sem{\trm_k}=\ple{\pi_k,\pi_{m+k}} \circ (\ple{\sem{\trm_1},...,\sem{\trm_m}}\times \ple{\sem{\trm_1},...,\sem{\trm_m}})$ one has 
$$
\DReIdx{\Doc}{\sem{\trm_k}\times\sem{\trm_k}}(\rmod{\gar{f}_k}\delem_{\sem{\asrt_k}})=\DReIdx{\Doc}{\ple{\sem{\trm_1},...,\sem{\trm_m}}\times \ple{\sem{\trm_1},...,\sem{\trm_m}}}
\DReIdx{\Doc}{\ple{\pi_k,\pi_{m+k}}}
(\rmod{\gar{f}_k}\delem_{\sem{\asrt_k}})$$

The way in which the interpretation of an \RR-graded signature is defined ensures the satisfaction of inequality (\ref{agreement}) for $f$.
$$
\prod_{k=1}^m \DReIdx{\Doc}{\ple{\pi_k,\pi_{m+k}}}(\rmod{\gar{f}_k}\delem_{\sem{\asrt_k}}) 
\order \DReIdx{\Doc}{\sem{f}\times\sem{f}}(\delem_{\sem{\asrt}})
$$
Since $\sem{\trm}=\sem{f(\trm_1,...,\trm_m)} = \sem{f}\circ\ple{\sem{\trm_1},...,\sem{\trm_m}}$, evaluating both side of the inequality along $\ple{\sem{\trm_1},...,\sem{\trm_m}}\times \ple{\sem{\trm_1},...,\sem{\trm_m}}$
one has 
$$
\prod_{k=1}^m \DReIdx{\Doc}{\sem{\trm_k}\times \sem{\trm_k}}(\rmod{\gar{f}_k}\delem_{\sem{\asrt_k}}) 
\order \DReIdx{\Doc}{\sem{\trm}\times\sem{\trm}}(\delem_{\sem{\asrt}})
$$
whence 
$$
\prod_{k=1}^m
\prod_{i=1}^n \DReIdx{\Doc}{\ple{\pi_i,\pi_{n+i}}}(\rmod{\gar{f}_k\rmul\gr(\trm_k, x_i)}\delem_{\sem{\srt_i}}) 
\order  \DReIdx{\Doc}{\sem{\trm}\times\sem{\trm}}(\delem_{\sem{\asrt}})$$

finally recall that $\gr(\trm,x_i) = \gar{f}_1\rmul\gr(\trm_1,x_i) \rplus ... \rplus\gar{f}_m\rmul\gr(\trm_m,x_i))$. Contraction of $\rmod{}$ completes the proof.

2) 
If $\form$ is $p(\trm_1,...,\trm_m)$ where $p$ predicate symbol of $L$ of arity $\asrt_1,...,\asrt_m$ and $\trm_i$ is a term of $L$ of arity $\srt_1,...,\srt_n\to \asrt_i$, then $\gr(p,x_i)=\gar{p}_1\rmul\gr(\trm_1,x_i)\rplus...\rplus\gar{p}_n\rmul\gr(\trm_n,x_i)$. Condition displayed in (\ref{agreement2}) together with contraction of $\rmod{}$ gives
$$\DReIdx{\Doc}{\pr_1}(\sem{p})\fmul\prod_{i=1}^n  \rmod{\gr(\form,x_i)}\DReIdx{\Doc}{\ple{\pi_i,\pi_{n+i}}}(\delem_{\sem{\srt_i}}) \order 
\DReIdx{\Doc}{\pr_2}(\sem{p})$$ 
where $\pr_1 = \ple{\pi_1,\ldots,\pi_m}$ and $\pr_2 = \ple{\pi_{m+1},\ldots,\pi_{2m}}$.

The reindexing of both sides along $\ple{\sem{\trm_1},...,\sem{\trm_m}}$ yields that the right hand side, that becomes $\DReIdx{\Doc}{\ple{\pi_{n+1},\ldots,\pi_{2n}}}(\sem{\form})$, is greater or equal to the left hand side, that becomes  

$$\DReIdx{\Doc}{\ple{\pi_1,\ldots,\pi_n}}(\sem{\form})\fmul \prod_{k=1}^m \DReIdx{\Doc}{\ple{\pi_k,\pi_{n+k}}}(\rmod{\gr(p,x_k)} \DReIdx{\Doc}{\sem{\trm_k}\times\sem{\trm_k}}\delem_{\sem{\asrt_k}}) 
$$
By point 1)  of this proposition, the displayed formula is greater than or equal to 
$$\DReIdx{\Doc}{\ple{\pi_1,\ldots,\pi_n}}(\sem{\form}) \fmul\prod_{k=1}^m \prod_{i=1}^n 
\DReIdx{\Doc}{\ple{\pi_1,\pi_{n+1}}}
(\rmod{\gr(p,x_k)\rmul\gr(\trm_k, x_i)}\delem_{\sem{\srt_i}})
$$
recall that $\gr(\form,x_i)=\gar{p}_1\rmul\gr(\trm_1,x_i)\rplus...\rplus\gar{p}_m\rmul\gr(\trm_m,x_i)$, so contraction of $\rmod{}$ leads to the claim.

If $\form$ is $\trm\Leq_\asrt\atrm$ where $\trm\atrm$ are terms of arity $\srt_1,...,\srt_n\to \asrt$, then $\sem{\form}=\DReIdx{\Doc}{\ple{\sem{\trm},\sem{\atrm}}}(\delem_{\sem{\asrt}}) $ and $\gr(\form,x_i)=\gr(\trm,x_i)\rplus \gr(\atrm,x_i)$, so by of $\rmod{}$ and conditions on $\trm$ and $\atrm$ it holds that
\begin{align*} 
&\DReIdx{\Doc}{\ple{\pi_1,\ldots,\pi_n}}
\DReIdx{\Doc}{\ple{\sem{\trm},\sem{\atrm}}}(\delem_{\sem{\asrt}}) \fmul
\prod_{i=1}^n \DReIdx{\Doc}{\ple{\pi_i,\pi_{n+i}}}(\rmod{\gr(\form,x_i)} \delem_{\sem{\srt_i}}) \\
\order\ & 
\DReIdx{\Doc}{\ple{\pi_1,\ldots,\pi_n}}
\DReIdx{\Doc}{\ple{\sem{\trm},\sem{\atrm}}}(\delem_{\sem{\asrt}}) 
\fmul
\DReIdx{\Doc}{\sem{\trm}\times\sem{\trm}}(\delem_{\sem{\asrt}}) 
\fmul
\DReIdx{\Doc}{\sem{\atrm}\times\sem{\atrm}}(\delem_{\sem{\asrt}})  \\ 
\order\ & 
\DReIdx{\Doc}{\ple{\pi_{n+1},\ldots,\pi_{2n}}}
\DReIdx{\Doc}{\ple{\sem{\trm},\sem{\atrm}}}(\delem_{\sem{\asrt}})
\end{align*} 
by transitivity of $\delem_{\sem{\asrt}}$.

If $\form$ is $\aform\Ltensor\bform$ where both $\aform$ and $\bform$ are well formed formulas in the context $\vec{\srt}$ satisfying the inductive hypothesis, then $\sem{\form}=\sem{\aform}\fmul\sem{\bform}$ and $\gr(\form,x_i)=\gr(\aform,x_i)+\gr(\bform,x_i)$. By contraction of $\rmod{}$, it follows 
\begin{align*} 
&\DReIdx{\Doc}{\ple{\pi_1,\ldots,\pi_n}}(\sem{\aform}\fmul\sem{\bform}) \fmul\prod_{i=1}^n \DReIdx{\Doc}{\ple{\pi_i,\pi_{n+i}}}(\rmod{\gr(\aform,x_i)\rplus\gr(\bform,x_i)} \delem_{\sem{\srt_i}})\\ 
\order\ & 
\DReIdx{\Doc}{\ple{\pi_1,\ldots,\pi_n}}(\sem{\aform}\fmul\sem{\bform})\fmul \prod_{i=1}^n \DReIdx{\Doc}{\ple{\pi_i,\pi_{n+i}}}(\rmod{\gr(\aform,x_i)} \delem_{\sem{\srt_i}}\fmul \rmod{\gr(\bform,x_i)}) \delem_{\sem{\srt_i}})\\ 
=\ & 
\DReIdx{\Doc}{\ple{\pi_1,\ldots,\pi_n}}(\sem{\aform})\fmul \prod_{i=1}^n \DReIdx{\Doc}{\ple{\pi_i,\pi_{n+i}}}(\rmod{\gr(\aform,x_i)} \delem_{\sem{\srt_i}})] \fmul \DReIdx{\Doc}{\ple{\pi_1,\ldots,\pi_n}}(\sem{\bform})\fmul \prod_{i=1}^n \DReIdx{\Doc}{\ple{\pi_i,\pi_{n+i}}}(\rmod{\gr(\bform,x_i)} \delem_{\sem{\srt_i}})
\end{align*} 
By hypothesis on $\aform$ the firs formulas in the square brackets is less than or equal to $\DReIdx{\Doc}{\ple{\pi_{n+1},\ldots,\pi_{2n}}}(\sem{\aform})$ and the same for $\bform$, whence the claim.

If $\form$ is $\textbf{1}$ then $\gr(\form,x_i)=0$. Then $\sem{\textbf{1}}=\funit$ and the claim is trivial by weakening fo $\rmod{}$.

If $\form$ is $\Lbang_\res\aform$ for some well formed formula $\aform$ in the context $\vec{\sigma}$, then $\sem{\form}=\rmod{\res}\sem{\aform}$ and $\gr(\form,x_i)=\res\rmul\gr(\aform,x_i)$. By the hypothesis on $\aform$ it is 
$$\DReIdx{\Doc}{\ple{\pi_1,\ldots,\pi_n}}(\sem{\aform})\fmul \prod_{i=1}^n \DReIdx{\Doc}{\ple{\pi_i,\pi_{n+i}}}(\rmod{\gr(\aform,x_i)} \delem_{\sem{\srt_i}}) \order \DReIdx{\Doc}{\ple{\pi_{n+1},\ldots,\pi_{2n}}}(\sem{\aform})$$
It suffices to multiply both sides by $\rmod{\res}$ and use comultiplication of $\rmod{}$ to prove the claim.
\end{proof}

We can now sketch the proof of \refToThm{sound}.

\begin{proof}[Proof of \refToThm{sound}]  
The only non-trivial part of the theorem is to prove that every component of 
$\sem{\blank}_\Wff$
 is monotone. 
As we said before, this is done by induction on rules of \ERPLL. 
We check only cases for \rn{pro} and \rn{subst}, leaving the other to the reader. 
\begin{description}
\item [\rn{pro}] 
Suppose $\sem{\Lbang_{\ares_1} \aform_1}\fmul...\fmul \sem{\Lbang_{\ares_n}\aform_n}
 \le \sem{\form}$. 
By definition of $\sem{\blank}_\Wff$, the left hand side of the inequality is equal to 
$\rmod{\ares_1}\sem{\aform_1}\fmul...\fmul \rmod{\ares_n}\sem{\aform_n}$. 
Then, by monotonicity of $\rmod{\res}$, we get 
$\rmod{r}(\rmod{\ares_1}\sem{\aform_1}\fmul...\fmul \rmod{\ares_n}\sem{\aform_n})
 \le \rmod{r}\sem{\form}$. 
By lax-monoidality-2 and comultiplication we get 
$\rmod{\res\rmul\ares_1}\sem{\aform_1}\fmul\ldots\fmul\rmod{\res\rmul\ares_n}\sem{\aform_n}\le \rmod\res\sem{\form}$. 
Finally, again by definition of $\sem{\blank}_\Wff$, we get the thesis. 
\item [\rn{subst}] 
Consider a context $\vec\srt$, a sort $\asrt$, terms $\trm$, $\atrm$ of sort $\asrt$ in the context $\vec\srt$ and a formula $\form$ in the context $\vec\srt, x:\asrt$. 
Suppose that $\sem\Gamma\le\sem{\form[\trm/x]}$ and $\sem\Delta\le \sem{\Lbang_{\gr(\form,x)} \trm\Leq_\asrt\atrm}$. 
By definition we have 
$\sem{\form[\trm/x]} = \DReIdx{\Doc}{\ple{\id{\sem{\vec\srt}},\sem\trm}}(\sem\form)$ and 
$\sem{\Lbang_{\gr(\form,x)} \trm\Leq_\asrt\atrm} = \DReIdx{\Doc}{\ple{\sem\trm,\sem\atrm}}(\rmod{\gr(\form,x)} \delem_{\sem\asrt})$. 
Combining these inequalities, we get 
$$ 
\sem{\Gamma,\Delta} = \sem\Gamma\fmul\sem\Delta 
  \le \DReIdx{\Doc}{\ple{\id{\sem{\vec\srt}},\sem\trm}}(\sem\form) \fmul \DReIdx{\Doc}{\ple{\sem\trm,\sem\atrm}}(\delem_{\sem\asrt})
$$ 
Using Item (2) of \refToLem{sound} and reflexivity of equality predicates, the following inequality holds in the poset $\Doc(\sem{\vec\srt}\times\sem\asrt\times\sem\asrt)$: 
$ 
\DReIdx{\Doc}{\ple{\pi_1,\pi_2}}(\sem\form) \fmul \rmod{\gr(\form,x)} \DReIdx{\Doc}{\ple{\pi_2,\pi_3}}(\delem_{\sem\asrt}) \le \DReIdx{\Doc}{\ple{\pi_1,\pi_3}}(\sem\form) 
$, and, 
reindexing along \fun{\ple{\id{\sem{\vec\srt}},\sem\trm,\sem\atrm}}{\sem{\vec\srt}}{\sem{\vec\srt}\times\sem\asrt\times\sem\asrt}, we obtain 
$$
\DReIdx{\Doc}{\ple{\id{\sem{\vec\srt}},\sem\trm}}(\sem\form) \fmul \rmod{\gr(\form,x)} \DReIdx{\Doc}{\ple{\sem\trm,\sem\atrm}}(\delem_{\sem\asrt}) \le \DReIdx{\Doc}{\ple{\id{\sem{\vec\srt}},\sem\atrm}}(\sem\form) 
$$
Putting together these inequalities and by definition of $\sem\blank_\Wff$ we get the thesis.
\qedhere
\end{description}
\end{proof}

\begin{rem}
Recall from \refToEx{provvisorio} 
that an $\RRPos$-graded signature can be interpreted into the \lip doctrine $\fun{\LIPDoc}{\LIP\op}{\Pos}$. 
The equality predicate  $\Leq_\asrt$ is interpreted as the metric $\delem_{\sem{\asrt}}$ of the space $\sem{\asrt}$. 
This implies that replicability of equality is not derivable in \ERPLL, as it would imply $\delem_{\sem{\asrt}}(x,y)\geq \delem_{\sem{\asrt}}(x,y)+\delem_{\sem{\asrt}}(x,y)$ which needs not to be true for a generic metric space $\sem{\asrt}$.  
\end{rem}

The following proposition states completeness of the semantics of \ERPLL in the sense of \cite{PittsCL}. 

\begin{prop}\label{prop:complete}
Let $\TT$ be a theory in \ERPLL over a signature $L$ and $\form,\aform$ formulas over $L$. 
If, for all \RR-\lip doctrines \ple{\Doc,\rmod{},\delem} and for all interpretations $\sem{\blank}$ of $\TT$ in \ple{\Doc,\rmod{},\delem}, we have $\sem{\form}_\Wff \order \sem{\aform}_\Wff$, then $\lcons{\form}{\aform}$ is provable in $\TT$. 
\end{prop}

The proof of this result is straightforward, as we can take the trivial interpretation in the syntactic doctrine $\Prop_\TT$ as in \refToProp{dottrinasintattica}, which  is \RR-\lip 

\begin{rem}
One could prove a stronger completeness result, showing an \emph{internal language theorem}. 
This consists in proving the equivalence between the 2-category of \RR-\lip doctrines and a suitable 2-category of \ERPLL theories. 
As it is well-known, the key lemma of this kind of results shows that 
from any doctrine one can extract a theory, where, loosely speaking,  sorts and function symbols are objects and arrows in the base, while predicate symbols and axioms are elements and entailments in the fibres. 
One could do this for \RR-\lip doctrines,  where the difficult part is to determine grades in the \RR-graded arities.
As noticed in \refToRem{lip-syntax}, graded arities provide a choice of \lip constants for the corresponding symbols, 
but from \refToProp{graded-lip0} we only know that in a \RR-\lip doctrine such \lip constants exist.
Hence, in order to obtain graded arities, we need the axiom of choice. 
A rigorous proof is left for future work.
\end{rem}

%% file: qet.tex

\section{Examples from quantitative equational theories and beyond}
\label{sect:qet}

Mardare et al. \cite{MardarePP16,MardarePP17} introduced the notion of \emph{quantitative equational theory} (QET) as a formal tool to describe and reason about quantitative algebras, that is, algebras whose carrier is a metric space and whose operations are non-expansive maps. 
In this section,  
we show how the calculus \ERPLL and its semantics can be used to reason about quantitative algebras, comparing it with QETs and obtaining at the same time a range of examples of \ERPLL theories. 
In the literature \cite{MardarePP16,MardarePP17,MardarePP21,Adamek22,BacciMPP18}, QETs are studied in great detail using the language of monads and proving many completeness results. 
Here, we just focus on examples, illustrating how \ERPLL works, 
leaving for future work a formal and systematic comparison between \ERPLL and QETs. 

We first  recall the key features of a QET. 
Syntactically, terms are built out from a single-sorted signature $\Omega$ of possibly infinitary function symbols, 
we write $f:I \in \Omega$, where $I$ is a set,  when $f$ is a function symbol in $\Omega$ whose arity is $I$.
In order to deal with distances, 
the key idea is to explicitly handle quantities, by working with labelled equations  $\trm \Leq_\ee \atrm$, where $\ee$ is a non-negative rational number. 
Such an equation states that the distance between $\trm$ and $\atrm$ is bounded by $\ee$. 
A sequent has the shape $\Gamma \vdash \trm \Leq_\ee \atrm$, where $\Gamma$ is a possibly infinite set of equations, and the calculus is defined by the following rules, where $f:I\in \Omega$ and $\sigma$ is a substitution. 
\begin{align*}  
\qrn{Refl}\ &\vdash \trm \Leq_0\trm 
   & 
\qrn{NExp}\ &\{\trm_i \Leq_{\ee} \atrm_i \mid i\in I\} \vdash f((\trm_i)_{i\in I}) \Leq_{\ee} f((\atrm_i)_{i\in I}) \\
\qrn{Sym}\ &\trm \Leq_\ee \atrm \vdash \atrm \Leq_\ee \trm 
   & 
\qrn{Subst}\ &\Gamma \vdash \trm\Leq_\ee\atrm\ \text{implies}\ \sigma(\Gamma) \vdash \sigma(\trm)\Leq_\ee\sigma(\atrm)\\
\qrn{Triang}\ &\{\trm \Leq_\ee \atrm, \atrm\Leq_{\ee'}\btrm\} \vdash \trm \Leq_{\ee+\ee'} \btrm  
   &
\qrn{Assum}\ &\text{if}\ \trm\Leq_\ee\atrm\in \Gamma,\ \text{then}\ \Gamma \vdash \trm\Leq_\ee\atrm\\ 
\qrn{Arch}\ &\{\trm \Leq_{\ee'} \atrm \mid \ee'>\ee\} \vdash \trm \Leq_{\ee} \atrm  
   & 
\qrn{Cut}\ &\text{if}\ \Gamma \vdash \psi\ \text{for all $\psi \in \Theta$ and}\ \Theta \vdash \trm \Leq_\ee \atrm,\\
         &&&\text{then}\ \Gamma\vdash \trm\Leq_\ee\atrm 
\end{align*} 
A first crucial difference between a QET and \ERPLL is that 
the underling logic of the former is not linear, as contraction and weakening are derivable: 
the hypotheses $\Gamma$ form a set, the rule \qrn{Cut} is in additive form and the rule \qrn{Assum} disregards some hypotheses. 
As a consequence, each formula in a QET has to be interpreted in a standard non-quantitative way, \eg as a subset, while 
in \ERPLL every formula admits a direct quantitative interpretation \eg, as $[0,\infty]$-valued function. 
Therefore, to reason quantitatively, 
in a QET one has to deal with an infinite family of predicates (\ie labelled equations) representing all possible approximations of a metric, 
which requires an explicit handling of quantities.
On the other hand, in \ERPLL 
one has a single equality predicate which is directly interpreted as a $[0,\infty]$-valued metric. 
Working  with explicit quantities makes it necessary to add specific rules to manage them (such as \qrn{Arch}) and to allow sequents with infinitely many hypotheses. 
Note also that infinitary sequents are needed to deal with infinitary function symbols, which are allowed in QETs, while \ERPLL only deals with finitary ones to keep sequents finitary as well. 

Another important  difference is that in QETs function symbols are forced to be non-expansive by the rule \qrn{NExp}, while in \ERPLL they can be arbitrary \lip functions. 
Since non-expansive maps are \lip with constant equal to $1$, \ERPLL can treat them, but, as we will see, it can go beyond this limitation, dealing with quantitative algebras where operations are arbitrary \lip functions. 

Remaining rules are instead similar: 
rules \qrn{Refl}, \qrn{Sym} and \qrn{Triang} above formalise the axioms of a distance and correspond to rules \rn{r}, \rn{s} and \rn{t} of \ERPLL, while 
rules \qrn{Subst} and \qrn{Cut} correspond to monotonicity of substitution and the cut rule of \ERPLL.

We now focus on examples. 
Examples in  \cite{MardarePP16} can be rephrased as \ERPLL theories where \RR is the semiring $\RRPos$ of non-negative real numbers, 
in such  a way that their intended semantics induces a sound interpretation of the corresponding \ERPLL theory in the doctrine of metric spaces and \lip maps (\cf \refToEx{liplip}). 
We will shorten notation for functions symbols in a single-sorted graded signature: 
if $f$ is a function symbol with $\gar{f} = \ple{\res_1,\star}\ldots\ple{\res_n,\star}\to\star$, we write $f : \ple{\res_1,\ldots,\res_n}$. 

We start by considering the theory \textbf{QS$_0$} of quantitative semilattices with zero. 
To write \textbf{QS$_0$} as a QET, one takes a signature whose function symbols are $\qsplus:2$ and  $\qszero:0$, 
while, to write \textbf{QS$_0$} as an \ERPLL  theory, one takes a $\RRPos$-graded signature with one sort, no predicate symbols and whose function symbols are $\qsplus:\ple{1,1}$ and $\qszero:\ple{}$, meaning that $\qsplus$ is non-expansive in both arguments. 
In \refToFig{qs-th}, we list on the left the axioms of \textbf{QS$_0$}  given in \cite{MardarePP16} and on the right the corresponding ones in \ERPLL.
\begin{figure}
\begin{math} 
\begin{array}{lll}
&\text{QET}&\text{\ERPLL}\\
\qrn{S0}&\vdash x\qsplus\qszero\Leq_0 x&\vdash x\qsplus\qszero\Leq x\\
\qrn{S1}&\vdash x\qsplus x\Leq_0 x& \vdash x\qsplus x\Leq x\\
\qrn{S2}&\vdash x\qsplus y\Leq_0 y\qsplus x& \vdash x\qsplus y\Leq y\qsplus x\\
\qrn{S3}&\vdash (x\qsplus y)\qsplus z\Leq_0 x\qsplus (y\qsplus z)& \vdash (x\qsplus y)\qsplus z\Leq x\qsplus (y\qsplus z)\\
\qrn{S4}& \{x\Leq_{\ee} x', y\Leq_{\ee'} y'\}\vdash x\qsplus y\Leq_{\ee\lor\ee'} x'\qsplus y'& 
\end{array} 
\end{math}
\caption{Axioms for quantitative semigroups as a QET and as an \ERPLL theory.}
\label{fig:qs-th}
\end{figure} 
The axiom \qrn{S0} says that the distance between $x\qsplus\qszero$ and $x$ is less than $0$. 
This information in \ERPLL is given forcing $\Lone \vdash x\qsplus\qszero\Leq x$, thanks to the fact that $\Lone$ is interpreted as the function constantly equal to $0$. 
Axioms \qrn{S1}, \qrn{S2} and \qrn{S3} are translated similarly. 
Axiom \qrn{S4} in the QET refines the non-expansiveness of $\qsplus$  and it is needed to properly handle labels of equations. 
Since in \ERPLL there are no labels, this axiom is not needed, as non-expansiveness of $\qsplus$ follows from rules of \ERPLL and the graded arity in the considered signature. 
Note that the axioms in \ERPLL are the usual ones of the theory of semilattices: $\qsplus$ is associative, commutative and idempotent and $\qszero$ is the neutral element. 

A semantics for \textbf{QS$_0$} as QET is given in \cite{MardarePP16} as follows. 
Let $\ple{M,d}$ be a metric space and consider the metric space \ple{C_d,H_d} where $C_d$ the set of compact subsets of $M$ with respect to $d$ and 
$H_d(A,B)=\max\{\sup_{a\in A}d(a,B),\sup_{b\in B}d(b,A)\}$, with $d(m,X) =\inf_{x\in X}d(m,x)$ for $m \in M$ and $X\subseteq M$, is the Hausdorff distance. 
Then, $\qszero$ is interpreted as the empty set and $\qsplus$ as the function of type $\ple{C_d,H_d}\times\ple{C_d,H_d}\rightarrow\ple{C_d,H_d}$ mapping two compact sets to their union. 
This assignment gives rise to a semantics also of \textbf{QS$_0$} seen as an \ERPLL theory 
in the  $\RRPos$-\lip doctrine \fun{\LIPDoc}{\LIP\op}{\Pos} (\cf \refToEx{liplip}).
We interpret a context with $n$ variables  as the product of $n$ copies of $\ple{C_d,H_d}$ and 
a term with $n$ free variables as a \lip function from $\ple{C_d,H_d}^n$ to \ple{C_d,H_d}. 

Other examples discussed in \cite{MardarePP16} deal with metrics over probability distributions. 
The considered theories are defined over a signature  whose 
function symbols are $+_e:2$ for all $e \in [0,1]$, which can be thought of as probabilistic choice operators with probability $e$, 
that is, intuitively, $x+_e y$ is $x$ with probability $e$ and $y$ with probability $1-e$. 
The QET \textbf{IBA$_p$} of $p$-Interpolative Barycentric Algebras, with $p\ge 1$, 
can be rephrased in \ERPLL using a single-sorted $\RRPos$-graded signature with no predicate symbols and whose function symbols are $+_e : \ple{e^{1/p}, (1-e)^{1/p}}$ for all $e \in [0,1]$, 
meaning that $+_e$ scales the distance of the first argument by a factor $e^{1/p}$ and that of the second argument by  a factor $(1-e)^{1/p}$. 
We report the axioms in \refToFig{iba}. 
\begin{figure}
\begin{math} 
\begin{array}{lll}
&\text{QET} &\text{\ERPLL} \\ 
\qrn{B1}&\vdash x +_1 y \Leq_0 x  &\vdash x+_1 y\Leq x\\
\qrn{B2}&\vdash x +_e x \Leq_0 x  &\vdash x+_e x\Leq x \\
\qrn{SC}&\vdash x +_e y \Leq_0 y +_{1-e} x &\vdash x +_e y\Leq y+_{1-e}x \\
\qrn{SA}&\vdash (x +_e y) +{e'} z \Leq_0 x +_{ee'} (y +_{\frac{e'-ee'}{1-ee'}} z) &\vdash (x +_e y) +_{e'} z\Leq x +_{ee'} (y +_{\frac{e'-ee'}{1-ee'}} z) \\
\qrn{IB$_p$}&\{x\Leq_{\ee_1} y, x' \Leq_{\ee_2} y'\}\vdash x+_e y \Leq_{\delta} x'+_e y' \quad & \\ 
& \hfill (\delta \ge (e\ee_1^p + (1-e)\ee_2^p)^{1/p}) & 
\end{array}
\end{math}
\caption{Axioms for interpolative barycentric algebras as a QET and as an \ERPLL theory.}
\label{fig:iba}
\end{figure}
The axioms are translated as for the theory \textbf{QS$_0$}. 
The last one has no counterpart in \ERPLL as it essentially states that $+_e$ is \lip with coeficients $e^{1/p}$ and $(1-e)^{1/p}$,  
\footnote{To check this we rely on the inequality $(a+b)^{1/p} \le a^{1/p} + b^{1/p}$, which holds for non-negative real numbers $a,b$.}
and this in \ERPLL follows from the the fact that $+_e$ has gaded arity $\ple{e^{1/p}, (1-e)^{1/p}}$ and so we can derive the entailment 
$\Lbang_{e^{1/p}} x\Leq y, \Lbang_{(1-e)^{1/p}} x' \Leq y' \vdash x+_e y \Leq x'+_e y'$. 

Mardare et al. \cite{MardarePP16} define a semantics of this theory taking as domain of the interpretation the set $\Delta[M]$ of Borel probability measures over a metric space \ple{M,d} with the $p$-Wasserstein distance $W_d^p$ defined by 
\[ W_d^p(\mu,\nu) = \sup \left\{ \left| \int f^p \mathsf{d}\mu - \int f^p \mathsf{d}\nu \right|^{1/p} \mid \fun{f}{M}{\RPos}\ \text{is non-expansive}\right\} \] 
and interpreting $+_e$ as the convex combination of two distributions, that is, 
the function mapping a pair of distributions \ple{\mu,\nu}  to $e\mu + (1-e)\nu$. 
This choice, induces an interpretation of the corresponding \ERPLL theory in 
\fun{\LIPDoc}{\LIP\op}{\Pos}, 
as it is easy to see that the convex combination operation has \lip constants $e^{1/p}$ and $(1-e)^{1/p}$ with respect to $W_d^p$. 

A specialised version of the theory \textbf{IBA$_p$} is the theory \textbf{BA} of Left-Invariant Barycnetric Algebras, which is obtained by replacing the axiom \qrn{IB$_p$} by 
the axiom \qrn{LI} stating 
$\vdash x +_e z \Leq_\ee y +_e z$ where $\ee\ge e$. 
This axiom, together with the others,  basically requires that \lip constants of $+_e$ are $e$ and $1-e$, but it also forces the distance to be smaller than $1$. 
To rephrase this theory in \ERPLL, we need to consider a single-sorted $\RRPos$-graded signature where function symbols are 
$+_e: \ple{e,1-e}$, for all $e\in [0,1]$, 
and having a constant predicate symbol $\onepred:\ple{}$ modelling the real number $1$. 
This matches the intuition that in our setting real numbers are ``truth values'', hence they can be represented syntactically by constant predicate symbols. 
Finally, we have to add to the theory above the axiom $\onepred\vdash x \Leq y$, stating that the distance is smaller than $1$. 
Note that, using the rules of \ERPLL, we can prove $\Lbang_\ee \onepred \vdash x \Leq x+_ez\Leq y+_ez$, for all $\ee\ge e$, which corresponds exactly to the axiom \qrn{LI}. 

Following again \cite{MardarePP16}, 
a semantics of the theory \textbf{BA} in the  \lip doctrine $\DMet{\PP_{[0,\infty]}}{\textbf{m}}$ 
can be constructed by taking  
as interpretation of the sort the set 
$\Delta[S,\Sigma]$ of probability measures over a measurable space $\ple{S,\Sigma}$, with the total variation distance 
$T(\mu,\nu)=\sup_{E\in \Sigma}|\mu(E)-\nu(E)|$, 
 and interpreting $+_e$ again as the convex combination of distributions.

\subsection{Beyond QETs}
As already noticed, in QETs function symbols must be interpreted by non-expansive maps due to the rule \qrn{NExp}. 
The calculus \ERPLL, instead, can express theories where function symbols have to be interpreted by arbitrary \lip maps. 
As examples of such situation, 
we briefly describe the theories 
\textbf{$\mathbb{R}$-MVS}, of metric vector spaces over the field $\mathbb{R}$ of real numbers, and 
\textbf{\RR-GC}, of \RR-graded combinators over a given resource semiring \RR, 
which, therefore, cannot be expressed as QETs. 

The single-sorted $\RRPos$-graded signature of \textbf{$\mathbb{R}$-MVS} has no predicate symbols and consists of function symbols
$\mvplus : \ple{1,1}$, $\mvmin : \ple{1}$ and  $\mvzero : \ple{}$ for the group structure and 
$\mvscalar_a : \ple{ |a| }$ for scalar multiplication by $a$, for all $a \in \mathbb{R}$. 
Note that the scalar multiplication $\mvscalar_a$ is required to scale the distance by a factor $|a|$, hence in general it is not non-expansive; 
indeed, the entailment $\Lbang_{|a|} x\Leq y \vdash \mvscalar_a(x) \Leq \mvscalar_a(y)$ is provable. 
Below we report the axioms:
\begin{align*}
\qrn{G1}\  & \vdash (x\mvplus y)\mvplus z \Leq x\mvplus (y\mvplus z)  
   &
\qrn{L1}\  & \vdash \mvscalar_a(x + y) \Leq \mvscalar_a(x) \mvplus \mvscalar_a(y) \\ 
\qrn{G2}\  & \vdash x \mvplus y \Leq y \mvplus x 
    &
\qrn{L2}\  & \vdash \mvscalar_{a+b}(x) \Leq \mvscalar_a(x) \mvplus \mvscalar_b(x) \\ 
\qrn{G3}\  & \vdash x \mvplus\mvzero \Leq x 
   &
\qrn{L3}\  & \vdash \mvscalar_{ab}(x) \Leq \mvscalar_a(\mvscalar_b(x)) \\
\qrn{G4}\  & \vdash x \mvplus (\mvmin x) \Leq \mvzero  
   &
\qrn{L4}\  & \vdash \mvscalar_1(x) \Leq x 
\end{align*} 
These axioms are actually the standard ones for vector spaces, the only difference is that the underlying calculus is substructural, hence equality can be interpreted as a distance. 
Any real vector space $V$ with a norm $\|\blank\|$ gives a semantics for this theory: the distance $d(u,v) = \|u-v\|$ makes the group operations non-expansive and the scalar multiplication by $a$ a \lip function with constant $|a|$ as 
$\|au-av\| =
 |a|\cdot \|u-v\|$.

Another example where function symbols are not necessarily non-expansive is the theory \textbf{\RR-GC} of \RR-graded combinators 
an \RR-Linear Combinatory Algebra \cite{Atkey18} (see also \refToExItem{graded}{1}) 
over a given resource semiring \RR. 
The single-sorted \RR-graded signature of \textbf{\RR-GC} has no predicate symbols, 
function symbols $\cdot : \ple{\rone,\rone}$ and $\lcab_\res : \ple{\res}$,\footnote{We use $\lcab_\res$ instead of the more standard $!_\res$ to avoid confusion with the modality of the logical calculus.}  for all $\res\in\RSet$, 
constants 
$B,C,I,K, W_{\res,\ares}, D, d_{\res,\ares}, F_\res$, for all $\res,\ares\in\RSet$, and 
$O_{\res,\ares}$, for all $\res,\ares\in\RSet$ with $\res\rord\ares$. 
The symbol $\lcab_\res$ has grade $\res$ as, intuitively, the expression $\lcab_\res \trm$ provides ``$\res$ copies of $\trm$'', hence, the cost of a substitution inside $\lcab_\res$ should be amplified by $\res$. 
Indeed, as a consequence of the rules of \ERPLL the  entailment 
$\lcons{\Lbang_\res \trm\Leq\atrm}{ \lcab_\res \trm \Leq \lcab_\res\atrm}$ is derivable. 
The axioms of the theory are the following: 
\begin{align*} 
\qrn{GC1}\  & \vdash ((B\cdot x)\cdot y)\cdot z \Leq x\cdot (y\cdot z) 
  & 
\qrn{GC2}\  & \vdash ((C\cdot x)\cdot y)\cdot z \Leq (x\cdot z)\cdot y\\ 
\qrn{GC3}\  & \vdash I \cdot x \Leq x  
  & 
\qrn{GC4}\  & \vdash (K\cdot x)\cdot \lcab_\rzero y \Leq x \\ 
\qrn{GC5}\  & \vdash (W_{\res,\ares}\cdot x)\cdot \lcab_{\res\rplus\ares} \Leq (x \cdot \lcab_\res y)\cdot \lcab_\ares y 
  & 
\qrn{GC6}\  & \vdash D\cdot \lcab_\rone x \Leq x \\
\qrn{GC7}\  & \vdash d_{\res,\ares} \cdot \lcab_{\res\rmul\ares} x \Leq \lcab_\res (\lcab_\ares x) 
  & 
\qrn{GC8}\  & \vdash (F_\res\cdot \lcab_\res x)\cdot \lcab_\res y \Leq \lcab_\res (x\cdot y) \\
\qrn{GC9}\  & \vdash O_{\res,\ares} \cdot \lcab_\ares x \Leq \lcab_\res x 
\end{align*} 
When \RR is the semiring $\RRPos$ of non-negative real numbers, 
models of \textbf{$\RRPos$-GC} in $\LIPDoc$ 
provide a metric variant of (total) $\RRPos$-Linear Combinatory Algebras \cite{Atkey18} (see also \refToExItem{graded}{1}). 
Hence, equality is interpreted as a metric on combinators, which could be regarded as a program distance, 
and, for every $\res\in\RPos$, $\lcab_\res$ is interpreted as a \lip map with constant $\res$.

%% file: categorical.tex

\section{The 2-category of \lip doctrines} 
\label{sect:categorical} 

As we have seen in the previous section, \RR-\lip doctrines are the objects of the 2-category $\EGMD{\RR}$. 
In this section we study its properties, relating it with with other 2-categories of doctrines. 
More in detail, we first show that \RR-\lip doctrines arise as  coalgebras for a 2-comonad on \RR-graded ones. 
This result smoothly generalises what happens in the non-linear setting and provides us with a universal construction of \RR-\lip doctrines from \RR-graded ones. 
Then, we study how our quantitative equality relates to the standard one given by left adjoints, introducing elementary \RR-graded doctrines and comparing them with \RR-\lip ones.

\input{comonad}
\input{relating-eq}

%% file: comonad.tex

\subsection{\RR-\lip doctrines as coalgebras}
\label{sect:comonad}

A key property of standard equality in the non-linear setting, recognise by \cite{PasqualiF:cofced,EPR}, is that it arises as a coalgebra structure over a primary doctrine, that is, a doctrine modelling conjunctions. 
More precisely, consider \refToDia{grafico} from \refToSect{background}
\[
\xymatrix{
\ED\ar@{^(->}[rr]<-1ex> \ar@{}[rr]|-{\top}&&\ar[ll]<-1ex> \PD  
}
\]
It shows an adjoint situation between  \PD and \ED, \ie the  2-categories of primary doctrines  and that of elementary ones that is, primary doctrines with equality. That adjoint situation is comonadic.
This fact not only reveals the coalgebraic nature of equality, but provides a universal construction yielding elementary doctrines from primary ones. 

The goal of this subsection is 
to generalise this fact to our linear setting. That is, we present a universal construction producing a \lip doctrine out of any \RR-graded one.
To this end, 
we define a 2-functor \fun{\GMEFun}{\GMD{\RR}}{\EGMD{\RR}} that is the right 2-adjoint of the obvious forgetful 2-functor \fun{\GEMFun}{\EGMD{\RR}}{\GMD{\RR}} and prove that 
\RR-\lip doctrines are the coalgebras of the 2-comonad induced by such a 2-adjunction.  
On one hand, this result provides us with a universal construction of \RR-\lip doctrines, which gives us a tool to produce semantics for the calculus \ERPLL. 
On the other hand,  it shows that our quantitative notion of equality, although not falling within Lawvere's definition by left adjoints (\refToProp{tortellini}), 
generalises it and preserves its coalgebraic nature. 

To construct $\GMEFun$, for any \RR-graded doctrine \ple{\Doc,\rmod{}}, 
we build an \RR-\lip doctrine \fun{\DMet{\Doc}{\rmod{}}}{\Met{\Doc}{\rmod{}}\op}{\Pos}, called \emph{\lip completion} of \ple{\Doc,\rmod{}}. 
Let $\CC$ be the base of \Doc. 
The category \Met{\Doc}{\rmod{}} is defined as follows: 
\begin{itemize}
\item objects are pairs \ple{A,\dist} where $A$ is an object in \CC and $\dist$ is an affine \Doc-distance on $A$; 
\item an arrow from \ple{A,\dist} to \ple{B,\adist} is an arrow \fun{f}{A}{B} in \CC such that $\rmod{\res} \dist \order \DReIdx{\Doc}{f\times f}(\adist)$, for some $\res\in\RSet$; 
\item composition and identities are those  of \CC. 
\end{itemize} 
Given  $\fun{f}{\ple{A,\dist}}{\ple{B,\adist}}$ we say that $\res\tracks{}f$ if $\res$ is such that  $\rmod{\res} \dist \order \DReIdx{\Doc}{f\times f}(\adist)$. So $\fun{f}{A}{B}$ underlies an arrow $\fun{f}{\ple{A,\dist}}{\ple{B,\adist}}$ if there is $\res$ that tracks $f$.

Intuitively, objects of  \Met{\Doc}{\rmod{}} can be regarded as metric spaces inside $\Doc$ and its arrows as abstract \lip functions between such spaces. 
It is easy to see that \Met{\Doc}{\rmod{}} is a category: 
identities are well-defined as, for every $\ple{A,\dist}$, the identity $\id{A}$ is tracked by $\rone$ thanks to the counit axiom of $\rmod{}$; 
also composition is well defined as, 
if $\res\tracks{}f$ and $\ares\tracks{}g$, then $\ares\rmul\res\tracks{} gf$ by
comultiplication and naturality of $\rmod{\ares}$ 
\[
\rmod{\ares\rmul\res} \dist 
\order \rmod{\ares} \rmod{\res} \dist 
\order \rmod{\ares} \DReIdx{\Doc}{f\times f}(\adist) 
= \DReIdx{\Doc}{f\times f}(\rmod{\ares} \adist) 
\order \DReIdx{\Doc}{f\times f}(\DReIdx{\Doc}{g\times g}(\bdist)) 
\]
Given elements  $\dist\in\Doc(A\times A)$ and $\adist\in\Doc(B\times B)$, 
we denote by $\dist\dmul\adist$ the element  
$\DReIdx{\Doc}{\ple{\pi_1,\pi_3}}(\dist)\fmul\DReIdx{\Doc}{\ple{\pi_2,\pi_4}}(\adist)$ 
in $\Doc(A\times B\times A\times B)$.

The proof of the following proposition is straightforward.
\begin{prop}\label{prop:dist-mul}
Let \ple{\Doc,\fmul,\funit} be a \pl doctrine.  
If $\dist$ and $\adist$ are \Doc-distances, then  $\dist\dmul\adist$ is a \Doc-distance. 
If $\dist$ and $\adist$ are affine, then $\dist\dmul\adist$ is affine. 
\end{prop}
Relying on this property, we can prove the following result. 
\begin{prop}\label{prop:met-cat}
The category \Met{\Doc}{\rmod{}} has finite products. 
\end{prop}

\begin{proof}
A terminal object is \ple{1,\funit_{1\times 1}}, where $1$ is a terminal object in \CC. 
Indeed, for any object \ple{A,\dist} in \Met{\Doc}{\rmod{}}, the unique map \fun{t_A}{A}{1} in \CC, induces a map from \ple{A,\dist} to \ple{1,\funit_{1\times 1}} since by weakening it holds 
$\rmod{\rzero} \dist \order \funit_{A\times A} = \DReIdx{\Doc}{t_A\times t_A}(\funit_{1\times 1})$.
Given \ple{A,\dist} and \ple{B,\adist} in \Met{\Doc}{\rmod{}},
the pair \ple{A\times B,\dist\dmul\adist} is an object in \Met{\Doc}{\rmod{}} by \refToProp{dist-mul}, as both $\dist$ and $\adist$ are affine \Doc-distances. 

Since both $\dist$ and $\adist$ are affine, projections \fun{\pi_A}{A\times B}{A} and \fun{\pi_B}{A\times B}{B} induce maps from \ple{A\times B,\dist\dmul\adist} to \ple{A,\dist} and \ple{B,\adist}, respectively; 
indeed, using affineness and counit, we have $\rmod{\rone}(\dist\dmul\adist) \order \rmod{\rone} \DReIdx{\Doc}{\ple{\pi_1,\pi_3}}(\dist) \order \DReIdx{\Doc}{\ple{\pi_1,\pi_3}}(\dist) = \DReIdx{\Doc}{\pi_A\times\pi_A}(\dist)$  and similarly for $\pi_B$. 
Finally, to check the universal property, let \ple{C,\bdist} be an object in \Met{\Doc}{\rmod{}} and \fun{f}{\ple{C,\bdist}}{\ple{A,\dist}} and \fun{g}{\ple{C,\bdist}}{\ple{B,\adist}} arrows in \Met{\Doc}{\rmod{}}. 
It suffices to show that the arrow \fun{\ple{f,b}}{C}{A\times B} in \CC induced by $f$ and $g$ underlines an arrow from \ple{C,\bdist} to \ple{A\times B,\dist\dmul\adist} in \Met{\Doc}{\rmod{}}. 
By hypothesis, we know that there are $\res,\ares\in\RSet$ such that 
$\rmod{\res}\bdist\order \DReIdx{\Doc}{f\times f}(\dist)$ and $\rmod{\ares}\bdist\order \DReIdx{\Doc}{g\times g}(\adist)$, hence, using contraction, we get 
\[ \rmod{\res\rplus\ares} \bdist \order \rmod{\res}\bdist \fmul  \rmod{\ares}\bdist  \order \DReIdx{\Doc}{f\times f}(\dist) \fmul \DReIdx{\Doc}{g\times g}(\adist) \] 
whence the claim.
\end{proof} 

We now focus on the fibres. 
Given an affine \Doc-distance $\dist$ on an object $A$ in \CC, an element $\alpha$ in $\Doc(A)$ and $\res$ in $\RSet$, we write $\res\tracks{\dist}\alpha$ when 
$\DReIdx{\Doc}{\pi_1}(\alpha)\fmul \rmod{\res} \dist \order \DReIdx{\Doc}{\pi_2}(\alpha) $. 
The suborder of $\Doc (A)$ of \dfn{\RR-graded descent data} is: 
\[ 
\Des{\dist}(A) = \{ \alpha \in \Doc(A) \mid \res\tracks{\dist}\alpha \mbox{ for some } \res\in\RSet \} 
\] 
Consider \fun{f}{\ple{A,\dist}}{\ple{B,\adist}} in \Met{\Doc}{\rmod{}} and $\beta \in \Des{\adist}(B)$, 
then,  $\res\tracks{} f$ and $\ares \tracks{\adist}\beta$ for some $\res,\ares \in \RSet$. 
By comultiplication and naturality of $\rmod{\ares}$, we get
$$\rmod{\ares\rmul\res}\dist \order \rmod{\ares}\rmod{\res}\dist \order \rmod{\ares}\DReIdx{\Doc}{f\times f}(\adist) = \DReIdx{\Doc}{f\times f}(\rmod{\ares}\adist)$$ 
Thus, 
since $\pi_i \circ (f\times f) = f\circ\pi_1$, for $i = 1,2$, we get 
$$
\DReIdx{\Doc}{\pi_1}(\DReIdx{\Doc}{f}(\beta)) \fmul \rmod{\ares\rmul\res}\dist \order 
\DReIdx{\Doc}{f\times f}( \DReIdx{\Doc}{\pi_1}(\beta) \fmul \rmod{\ares}\adist ) \order 
\DReIdx{\Doc}{\pi_2}(\DReIdx{\Doc}{f}(\beta))
$$ 
showing that  
$\ares\rmul\res\tracks{\dist} \DReIdx{\Doc}{f}(\beta)$. 
In other words $\DReIdx{\Doc}{f}$ applies $\Des{\adist}(B)$ to $\Des{\dist}(A)$ and  the assignments
\[ \DMet{\Doc}{\rmod{}}\ple{A,\dist} := \Des{\dist}(A) \qquad \DMet{\Doc}{\rmod{}}(f) := \DReIdx{\Doc}{f}  \]
determine a functor \fun{\DMet{\Doc}{\rmod{}}}{\Met{\Doc}{\rmod{}}\op}{\Pos}.

\begin{prop}\label{prop:des-doc}
$\DMet{\Doc}{\rmod{}}$ is a \RR-\lip  doctrine.
\end{prop}

\begin{proof}
We have $\funit \in\Des{\dist}$, as $\funit \fmul \rmod{\rzero}\dist \order \funit$, by weakening. 
For $\alpha,\beta\in\Des{\dist}$ it holds 
$\DReIdx{\Doc}{\pi_1}(\alpha) \fmul \rmod{\res}\dist \order \DReIdx{\Doc}{\pi_2}(\alpha)$ and 
$\DReIdx{\Doc}{\pi_1}(\beta) \fmul \rmod{\ares}\dist \order \DReIdx{\Doc}{\pi_2}(\beta)$, for some $\res,\ares \in \RSet$. 
Then, by contraction and commutativity of $\fmul$
\begin{align*} 
\DReIdx{\Doc}{\pi_1}(\alpha\fmul\beta) \fmul \rmod{\res\rplus\ares}\dist &\order 
\DReIdx{\Doc}{\pi_1}(\alpha) \fmul \rmod{\res}\dist \fmul \DReIdx{\Doc}{\pi_1}(\beta) \fmul \rmod{\ares}\dist \\ 
&\order 
\DReIdx{\Doc}{\pi_2}(\alpha) \fmul \DReIdx{\Doc}{\pi_2}(\beta) = 
\DReIdx{\Doc}{\pi_2}(\alpha \fmul \beta) 
\end{align*} 
For $\alpha \in \Des{\dist}$ and $\res \in \RSet$ it holds $\DReIdx{\Doc}{\pi_1}(\alpha) \fmul \rmod{\ares}\dist \order \DReIdx{\Doc}{\pi_2}(\alpha)$, for some $\ares \in \RSet$. 
Comultiplication, lax-monoidality  and naturality of $\rmod{\res}$ lead to 
\begin{align*} 
\DReIdx{\Doc}{\pi_1}(\rmod{\res}\alpha) \fmul \rmod{\res\rmul\ares}\dist &\order  
\rmod{\res}\DReIdx{\Doc}{\pi_1}(\alpha) \fmul \rmod{\res}\rmod{\ares}\dist\\
&\order 
\rmod{\res}(\DReIdx{\Doc}{\pi_1}(\alpha)\fmul\rmod{\ares}\dist) \\ 
&\order 
\rmod{\res}\DReIdx{\Doc}{\pi_2}(\alpha) = 
\DReIdx{\Doc}{\pi_2}(\rmod{\res}\alpha) 
\end{align*} 

In the end, to show that $\DMet{\Doc}{\rmod{}}$ is \RR-\lip, we set $\delem_{\ple{A,\dist}} := \dist$, for each object \ple{A,\dist} in \Met{\Doc}{\rmod{}}. 
First of all, note that $\delem_\ple{A,\dist} \in \DMet{\Doc}{\rmod{}}(\ple{A,\dist}\times\ple{A,\dist}) = \DMet{\Doc}{\rmod{}}\ple{A\times A,\dist\dmul\dist}$ since, by symmetry and transitivity of $\dist$ and counit, we have 
$\DReIdx{\Doc}{\ple{\pi_1,\pi_2}}(\dist) \fmul \rmod{\rone} (\dist\dmul\dist) \order \DReIdx{\Doc}{\ple{\pi_3,\pi_4}}(\dist)$. 
Then, the thesis follows by \refToProp{cara2}. 
\end{proof} 

The construction 
we have just described 
extends to a 2-functor \fun{\GMEFun}{\GMD{\RR}}{\EGMD{\RR}} as discussed below. 
Let \oneAr{\ple{F,f}}{\ple{\Doc,\rmod{}}}{\ple{\aDoc,\rmod{}'}} be a 1-arrow in \GMD{\RR} and let \fun{h}{\ple{A,\dist}}{\ple{B,\adist}} be an arrow in \Met{\Doc}{\rmod{}}. 
The following assignments 
define a functor 
\fun{\overline{F}}{\Met{\Doc}{\rmod{}}}{\Met{\aDoc}{\rmod{}'}}: 
\[
\overline{F}\ple{A,\dist} := \ple{FA,f_{A\times A}(\dist)}\qquad \overline{F}h := Fh 
\] 
Then, $\GMEFun$ is defined on arrows as follows: 
\[ 
\GMEFun\ple{F,f} := \ple{\overline{F},f} \qquad \GMEFun(\theta) := \theta 
\]
where \twoAr{\theta}{\ple{F,f}}{\ple{G,g}} is a 2-arrow in \GMD{\RR}. 

\begin{prop}\label{prop:gmefun}
The assignments above define a 2-functor 
$$\fun{\GMEFun}{\GMD{\RR}}{\EGMD{\RR}}$$ 
\end{prop} 
\begin{proof}
To prove the thesis, we just have to show that \GMEFun is well-define, then algebraic identities follow immediately. 
Let \oneAr{\ple{F,f}}{\ple{\Doc,\rmod{}}}{\ple{\aDoc,\rmod{}'}} be a 1-arrow in \GMD{\RR}. 
The fact that \fun{\overline{F}}{\Met{\Doc}{\rmod{}}}{\Met{\aDoc}{\rmod{}'}} is a well-defined product-preserving functor follows from the following two facts. 
\begin{itemize}
\item Let $\dist$ be an affine \Doc-distance on $A$, then $f_{A\times A}(\dist)$ is an affine \aDoc-ditance on $FA$. 
This is immediate as $f$ is a natural transformation preserving the monoidal structure. 
\item Let \fun{h}{\ple{A,\dist}}{\ple{B,\adist}} be an arrow in \Met{\Doc}{\rmod{}}, then  \fun{Fh}{\ple{FA,f_{A\times A}(\dist)}}{\ple{FB,f_{B\times B}(\adist)}} is an arrow in \Met{\aDoc}{\rmod{}'}. 
By definition of arrows in \Met{\Doc}{\rmod{}}, we know that 
$\rmod{\res}\dist \order \DReIdx{\Doc}{h}(\adist)$ holds for some $\res\in\RSet$. 
Since \ple{F,f} is a 1-arrow in \GMD{\RR}, we know that $ {\rmod{\res}'}_{FA\times FA}  \circ f_{A\times A} \order f_{A\times A}\circ {\rmod{\res}}_{A\times A}$, hence, applying $f_{A\times A}$ we get  
\begin{small}
\[ \rmod{\res}' f_{A\times A} (\dist) \order f_{A\times A}(\rmod{\res}\dist) \order 
f_{A\times A} (\DReIdx{\Doc}{h} (\adist)) = 
\DReIdx{\aDoc}{Fh}(f_{B\times B}(\adist))  \] 
\end{small}
\end{itemize}
To check that $f_A$ applies $\DMet{\Doc}{\rmod{}} \ple{A,\dist} = \Des{\dist}$ into $\DMet{\aDoc}{\rmod{}'}\ple{FA,f_{A\times A}(\dist)} = \Des{f_{A\times A}(\dist)}$, it is enough to note that, since, for all $\alpha \in \Des{\dist}$, we have  $\DReIdx{\Doc}{\pi_1}(\alpha) \fmul^\Doc \rmod{\res}\dist \order \DReIdx{\Doc}{\pi_2}(\alpha)$, for some $\res\in\RSet$, we 
\begin{align*} 
\DReIdx{\aDoc}{F\pi_1}(f_A(\alpha)) \fmul^\aDoc \rmod{\res}'f_{A\times A}(\dist) &\order 
\DReIdx{\aDoc}{f\pi_1}(f_A(\alpha)) \fmul^\aDoc f_{A\times A}(\rmod{\res}\dist) \\
&= f_{A\times A} (\DReIdx{\Doc}{\pi_1}(\alpha) \fmul^\Doc \rmod{\res}\dist) \\
&\order 
f_{A\times A} \DReIdx{\Doc}{\pi_2}(\alpha)) \\
&= \DReIdx{\aDoc}{F\pi_2}(f_A(\alpha)) 
\end{align*} 
Since the \RR-graded structure of $\DMet{\Doc}{\rmod{}}$ and $\DMet{\aDoc}{\rmod{}'}$ are, respectively, that of \ple{\Doc,\rmod{}} and \ple{\aDoc,\rmod{}'}, $f$ preserves the structure by hypothesis, hence \ple{\overline{F},f} is a well-defined 1-arrow in \EGMD{\RR}. 

Let \twoAr{\theta}{\ple{F,f}}{\ple{G,g}} be a 2-arrow in \GMD{\RR}. The fact that $\GMEFun(\theta) = \theta$ is a well-defined 2-arrow in \EGMD{\RR} follows from the following observation: 
let \ple{A,\dist} be an object in \Met{\Doc}{\rmod{}}, then \fun{\theta_A}{\ple{FA,f_{A\times A}(\dist)}}{\ple{GA,g_{A\times A}(\dist)}} is an arrow in \Met{\aDoc}{\rmod{}'}. 
This is immediate as, by definition of 2-arrow and counit, we get 
$\rmod{\rone}'f_{A\times A}(\dist) \order f_{A\times A}(\dist) \order \DReIdx{\aDoc}{\theta_{A\times A}}(g_{A\times A}(\dist)) = \DReIdx{\aDoc}{\theta_A\times\theta_A}(g_{A\times A}(\dist))$. 
\end{proof}

For each \RR-graded doctrine \ple{\Doc,\rmod{}} and each \RR-\lip doctrine  $\ple{\aDoc,\armod{},\delem}$ there are a 1-arrow \ple{U^\Doc,u^\Doc} in \GMD\RR and a 1-arrow \ple{E^\aDoc,e\aDoc} in \EGMD\RR 
as depicted below: 
\[
\xymatrix@C=7.5em@R=1em{
\Met{\Doc}{\rmod{}}\op \ar[rd]^(.4){\DMet{\Doc}{\rmod{}} }_(.4){}="P"
\ar[dd]_{(U^\Doc)\op}^{}="F"
&\\
 & \Pos \\
\CC\o \ar[ru]_(.4){\Doc}^(.4){}="R"&
\ar"P";"R"_{u^\Doc \kern.5ex\cdot\kern-.5ex}="b"
} 
\qquad\qquad 
\xymatrix@C=7.5em@R=1em{
\D\op \ar[rd]^(.4){\aDoc}_(.4){}="P"
\ar[dd]_{(E^\aDoc)\op}^{}="F"
&\\
 & \Pos \\
\Met{\aDoc}{\armod{}}\op \ar[ru]_(.4){\DMet{\aDoc}{\armod{}}}^(.4){}="R"&
\ar"P";"R"_{e^\aDoc\kern.5ex\cdot\kern-.5ex}="b"
}
\]
On the left, 
$U^\Doc$ forgets distances, mapping \ple{A,\dist} to its underlying obect $A$ and being the identity on arrows, $u^\Doc$ is the obvious natural inclusion of $\DMet{\Doc}{\rmod{}}\ple{A,\dist}$ into $\Doc(A)$.  
On the right, 
$E^\aDoc$ maps an obect $A$ to \ple{A,\delem_A} and is the identity on arrows 
and $e^\aDoc$ is the identity on the fibres. 
Note that $E^\Doc$ is well-defined on objects as \ple{\aDoc,\armod{},\delem} is \RR-\lip, hence $\delem_A$ is an affine $\aDoc$-distance on $A$, and on arrows by \refToCor{graded-lip}; 
moreover, it preserves finite products by Items~\ref{def:affine-lip:2},\ref{def:affine-lip:4} of \refToDef{affine-lip}. 

These families of 1-arrows will give rise, respectively, to the counit and the unit of the 2-adunction $\GEMFun\dashv\GMEFun$ which is the core of the next theorem and also the main result of this subsection. 

\begin{thm}\label{thm:graded-adj}
The 2-functor \fun{\GEMFun}{\EGMD{\RR}}{\GMD{\RR}} is 2-comonadic: 
\begin{enumerate}
\item $\GMEFun$ is the right 2-adoint of $\GEMFun$ and 
\item \EGMD\RR is isomorphic to the 2-category of coalgebras for the 2-comonad $\GEMFun\circ\GMEFun$. 
\end{enumerate}
\end{thm} 
\begin{proof}
First we prove that $\GEMFun\dashv\GMEFun$. 
Consider the following 2-natural transformations: 
\[
\counit_\ple{\Doc,\rmod{}} := \ple{U^\Doc,u^\Doc} 
\qquad 
\unit_\ple{\aDoc,\armod{}} := \ple{E^\aDoc,e^\aDoc} 
\]
for each object \ple{\Doc,\rmod{}} in \GMD{\RR} and \ple{\aDoc,\armod{},\adelem} in \EGMD{\RR}. 
We have that \TNat{\counit}{\GEMFun\circ\GMEFun}{\Id{\GMD{\RR}}} is a 2-natural transformation because, 
for any 1-arrow \oneAr{\ple{F,f}}{\ple{\Doc,\rmod{}}}{\ple{\Doc',\rmod{}'}} in \GMD{\RR}, it holds $U^{\Doc'}\circ \overline{F}  = F \circ U^\Doc$, by definition of \GMEFun. 
On the other hand, \TNat{\unit}{\Id{\EGMD{\RR}}}{\GMEFun\circ\GEMFun} is a 2-natural transformation because 
each 1-arrow \oneAr{\ple{F,f}}{\ple{\aDoc,\armod{},\delem}}{\ple{\aDoc',\armod{}',\delem'}} in \EGMD{\RR} preserves distances, hence 
$E^{\aDoc'}FA = \ple{FA,\delem'_{FA}} = \ple{FA,f_{A\times A}(\delem_A)} = \overline{F} E^\aDoc A$, for any object $A$ in the base category of \aDoc.
The two triangular identities hold as the following diagrams (in \GMD{\RR}) commute for any \RR-graded doctrine \ple{\Doc,\rmod{}} and \RR-\lip doctrine \ple{\aDoc,\armod{},\adelem}: 
\[
\xymatrix@C=10ex{
\ple{\DMet{\Doc}{\rmod{}},\rmod{}} \ar[rd]_{\ple{\Id{},\id{}}} \ar[r]^-{\ple{E^{\DMet{\Doc}{\rmod{}}},e^{\DMet{\Doc}{\rmod{}}}}} 
  & \ple{\DMet{\DMet{\Doc}{\rmod{}}}{\rmod{}},\rmod{}} \ar[d]^{\ple{\overline{U^{\DMet{\Doc}{\rmod{}}}},u^{\DMet{\Doc}{\rmod{}}}}} \\ 
  & \ple{\DMet{\Doc}{\rmod{}},\rmod{}} 
}\quad
\xymatrix@C=10ex{
\ple{\aDoc,\armod{}} \ar[dr]_{\ple{\Id{},\id{}}} \ar[r]^-{\ple{E^\aDoc,e^\aDoc}}  
  & \ple{\DMet{\aDoc}{\armod{}},\armod{}} \ar[d]^{U^\aDoc,u^\aDoc} \\ 
  & \ple{\aDoc,\armod{}} 
}
\]

Let \shopr{T} be the 2-comonad on \GMD{\RR} induced by the adjunction. 
In order to prove that \EGMD{\RR} is isomorphic to the 2-category $\EGMD{\RR}^{\shopr{T}}$ of coalgebras for \shopr{T}, note that there is a comparison 2-functor \fun{\shopr{K}}{\EGMD{\RR}}{\GMD{\RR}^{\shopr{T}}} mapping 
an \RR-\lip doctrine \ple{\aDoc,\armod{},\adelem} to the coalgebra \oneAr{\ple{E^\aDoc,e^\aDoc}}{\ple{\aDoc,\armod{}}}{\ple{\DMet{\aDoc}{\armod{}},\armod{}}}and being the identity on arrows. 
On the other hand, given a coalgebra \oneAr{\ple{F,f}}{\ple{\Doc,\rmod{}}}{\ple{\DMet{\Doc}{\rmod{}},\rmod{}}} on an \RR-graded  doctrine \ple{\Doc,\rmod{}}, since $\ple{U^\Doc,u^\Doc}\circ\ple{F,f}=\ple{\Id{},\id{}}$, we have $f = \id{}$ and $FA = \ple{A,\dist}$ for any object $A$ in the base category of \Doc and $\Doc (A) = \Des{\dist}(A)$. Hence \ple{\Doc,\rmod{},\delem}  is \RR-\lip with $\delem_A = \dist$. 
On arrows the 2-functor is the identity, becuase being a coalgebra morphism is exactly the same as preserving distances. Such a 2-functor is the inverse of \shopr{K}. 
\end{proof}

This result shows that \RR-\lip doctrines can be seen as coalgebras for the 2-comonad $\GEMFun\circ\GMEFun$ on $\GMD\RR$, that is, 
pairs consisting of an \RR-graded doctrine \ple{\Doc,\rmod{}} and a 1-arrow \oneAr{\ple{F,f}}{\ple{\Doc,\rmod{}}}{\ple{\DMet{\Doc}{\rmod{}},\rmod{}}} in \GMD\RR. 
This means that \RR-\lip doctrines are \emph{structures} over \RR-graded one. 
The next theorem such structures in a very precise way. 

\begin{prop}\label{prop:kz}
Let \oneAr{\ple{F,f}}{\ple{\Doc,\rmod{}}}{\ple{\DMet{\Doc}{\rmod{}},\rmod{}}} be a $\GEMFun\circ\GMEFun$-coalgebra structure over \ple{\Doc,\rmod{}}. 
Then, \ple{F,f} is a left adoint of \ple{U^\Doc,u^\Doc} in the 2-category $\GMD\RR$. 
\end{prop}
\begin{proof}
Since \ple{F,f} is a coalgebra, we have that $\ple{U\Doc,u\Doc}\circ\ple{F,f} = \ple{\Id,\id{}}$, hence we take the identity to be the unit of the adunction. 
Now, given an object \ple{A,\dist} in $\Met{\Doc}{\rmod{}}$, we have that $FU^\Doc\ple{A,\dist} = \ple{A,\delem_A}$, where $\delem_A \in \Doc(A\times A)$ is a $\Doc$-distance satisfying the \RR-graded substitutivity condition because \ple{\Doc,\rmod{},\delta} is \RR-\lip by \refToThm{graded-adj}. 
Therefore, since $\dist \in \Doc(A\times A)$, we get that 
$\DReIdx{\Doc}{\ple{\pi_1,\pi_2}}(\dist) \fmul \rmod{\res} \DReIdx{\Doc}{\ple{\pi_2,\pi_3}}(\delem_A) \order \DReIdx{\Doc}{\ple{\pi_1,\pi_3}}(\dist)$, for some $\res\in\RSet$. 
Then, reindexing along \fun{\ple{\pi_1,\pi_1,\pi_2}}{A\times A}{A\times A \times A} and using reflexivity of $\dist$, we get 
$\rmod\res \delem_A \order \dist$. 
This proves that the identity on $A$ gives riseto an arrow \fun{\id{A}}{\ple{A,\delem_A}}{\ple{A,\dist}} in $\Met\Doc{\rmod{}}$, which will be the components of the counit of the adunction. 
Finally, the triangular laws of adunctions trivially holds as both the unit and the counit have the identity as underlying arrow, 
therefore we get the thesis. 
\end{proof}

In other words, \refToProp{kz} proves that the 2-comonad $\GEMFun\circ\GMEFun$ is a \emph{KZ comonad} \cite{Kock95}. 
This express in categorical terms the fact that being \RR-\lip is a \emph{property} rather than a structure. 
Indeed, \refToProp{kz} shows that every $\GEMFun\circ\GMEFun$-coalgebra structure on a \RR-graded doctrine \ple{\Doc,\rmod{}} must be left adjoint of the counit \ple{U^\Doc,u^\Doc} of the comonad, hence they are all isomorphic thanks to general properties of adunctions in 2-categories. 
Equivalently, there is at most one $\GEMFun\circ\GMEFun$-coalgebra structure on \ple{\Doc,\rmod{}} up to isomorphism. 

Comparing with the non-linear case, 
by \refToProp{graded-primary}, we know that \PD is a 2-subcategory of $\GMD{\RR}$ and also that $\ED$ is 2-subcategory of $\EGMD{\RR}$, since axioms of \RR-\lip doctrines (\cf \refToDef{affine-lip}) trivally holds for elementary ones when the modality is the identity as in \refToProp{graded-primary}. 
Then, it is easy to check that, by restricting $\GEMFun\dashv\GMEFun$ to \PD and \ED, we retrieve the 2-adjuction presented in \cite{EPR}. 
This is essentially because in a primary doctrine, viewed as a \RR-graded one, distances are actually equivalence relations and \lip arrows just preserves such relations. 
This leads us to the following commutative diagram: 
\begin{equation}\label{dia:diagram} 
\xymatrix{
\GMD{\RR} \ar@<1ex>[r]^{\GMEFun}\ar@{}[r]|-{\top} & \EGMD{\RR}\ar@<1ex>[l] \\ 
\PD \ar[u]\ar@<1ex>[r]\ar@{}[r]|-{\top} & \ED\ar[u] \ar@<1ex>[l]
}
\end{equation} 
\vspace{.5em}

The 2-functor \fun{\GMEFun}{\GMD{\RR}}{\EGMD{\RR}} can be used to construct examples of \RR-\lip doctrines,  
providing semantics for the calculus \ERPLL, 
obtained in the same way as in \refToEx{provvisorio}.

\begin{exas}\label{ex:provvisorio2}
\begin{enumerate}
\item\label{ex:provvisorio2:1} 
The \lip doctrine $\ple{\LIPDoc, \textbf{m}, \delem}$  of \refToEx{liplip} is the \lip completion of the $\RRPos$-graded doctrine \ple{\PP_{[0,\infty]},\textbf{m}} of \refToExItem{graded}{0}. 

\item\label{ex:provvisorio2:2} 
The \lip completion of the \RR-graded doctrine $\ple{K_\WW,\overline{a}}$ of \refToExItem{graded}{444} (where  $\WW$ is a monoidal Kripke frame and $a$ a lax axions of $\RR$ on $\WW$) is the \lip doctrine $\fun{\DMet{K_\WW}{\overline{a}}}{\Met{K_\WW}{\overline{a}}\op}{\Pos}$. 
Objects of the category $\Met{K_\WW}{\overline{a}}$ are pairs $\ple{X,\dist}$ where the ternary relation $\dist\subseteq X\times X\times W$ is an affine $K_\WW$-distance in the sense of \refToExItem{distance}{444}. 
An arrow $f:\ple{X,\dist}\to\ple{Y,\adist}$ is a function $f:X\to Y$ for which there is $\res\in \RSet$ such that, 
for all $x,x'\in X$ and $w\in W$, 
$\ple{x,x',w}\in\dist$ implies $\ple{f(x),f(x'),a(\res,w)}\in\adist$. 
An element in $\DMet{K_\WW}{\overline{a}}\ple{X,\dist}$ is $U\subseteq X\times W$ for which there is $\res\in \RSet$ such that, 
for all $x,x'\in X$ and $w_1,w_2\in W$, 
$\ple{x,w_1}\in U$ and $\ple{x,x',w_2}\in \dist$ implies $\ple{x', w_1\krpmul a(\res,w_2)}\in U$.

\item\label{ex:provvisorio2:3} 
The \lip completion of the \RR-graded realizability doctrine 
$\ple{\mathcal{R}_A,\rmod{}}$ of \refToExItem{graded}{1} (where $A$ be an ordered \RR-LCA) 
is the \RR-\lip doctrine \fun{\DMet{\mathcal{R}_A}{\rmod{}}}{\Met{\mathcal{R}_A}{\rmod{}}\op}{\Pos}. 
Objects of $\Met{\mathcal{R}_A}{\rmod{}}$ are pairs $\ple{X,\dist}$ such that $\fun{\dist}{X\times X}{\PP(|A|)}$ is an affine $\mathcal{R}_A$-distance as in \refToExItem{distance}{1}.
An arrow $\fun{f}{\ple{X,\dist}}{\ple{Y,\adist}}$ of $\Met{\ct{R}_A}{\rmod{}}$ is a function such that there $\res\in \RSet$ and there is a realizer in $|A|$ mapping, for every $x,x'\in X$, elements of $!_\res\dist(x,x')$ into elements of $\adist(f(x),f(x'))$. 
An element $\alpha$ in $\DMet{\mathcal{R}_A}{\rmod{}}\ple{X,\dist}$ is a function $\fun{\alpha}{X}{\PP(|A|)}$ such that there $\res\in \RSet$ and a realizer  in $|A|$ mapping, for every $x,x'\in X$, elements of 
$\{\textbf{P}ab\in |A| \mid  a\in \alpha(x), b \in !_\res\dist(x,x')\}$ into elements of 
$\alpha(x')$. 
\end{enumerate}
\end{exas}

The fact that $\GMEFun$ is a right 2-adoint allows us to characterise in a precise way model of \ERPLL theories in \RR-\lip doctrines obtained as \lip completions of an \RR-graded one. 
Indeed, if $\TT$ is a theory in \ERPLL and \ple{\Doc,\rmod{}} is an \RR-graded doctrine, the 2-adunction $\GEMFun\dashv\GMEFun$ gives us the following isomorphism between hom-categories: 
\[ \GMD\RR(\ple{\Prop_\TT,\Lbang},\ple{\Doc,\rmod{}}) \simeq \EGMD\RR(\ple{\Prop_\TT,\Lbang,\delem^\Leq},\ple{\DMet{\Doc}{\rmod{}}, \rmod{}, \delem}) \]
which says that models (and their homomorphisms) into an \RR-graded doctrine \ple{\Doc,\rmod{}}, that is, ignoring equality, are the same as models into the \lip completion of \ple{\Doc,\rmod{}} that preserve equality.

%% file: relating-eq.tex

\subsection{Relating notions of equality}
\label{sect:rel-eq} 

In this section we establish a relation between the quantitative equality introduced in \refToSect{graded} and the traditional Lawvere's notion of equality defined by left adjoints.

First of all, we cast Lawvere's equality to \RR-graded doctrines. 
This is quite easy since \RR-graded doctrines are in particular \pl doctrines, hence one can always consider those that are elementary according to \refToDef{lelementary}. 
More explicitly, 
this means that, a \RR-graded doctrine \ple{\Doc,\rmod{}} is \emph{elementary} if 
for every object $A$ in the base of \Doc there is an element $\delta_A \in \Doc(A\times A)$, which is reflexive and substitutive. 

We have already seen in \refToProp{tortellini} that such an equality predicate is replicable. 
In the \RR-graded setting, this non-linearity is even more evident as detailed by the next proposition. 

Given an \RR-graded doctrine \ple{\Doc,\rmod{}}, an element 
$\alpha$ in $\Doc(A)$ is said to be \emph{$\rmod{}$-intuitionistic} if 
$\alpha \order \rmod{\res}\alpha$ for all $\res\in\RSet$. 
Intuitively, this means that $\alpha$ can provide an arbitrary amount of copies of itself, hence it can be used in an unrestricted way. 
Note that a $\rmod{}$-intuitionistic element is, in particular, affine and replicable. 
Indeed, if $\alpha$ is $\rmod{}$-intuitionistic, we have 
$\alpha \order \rmod{\rzero}\alpha \order \funit$ (by weakening) and 
$\alpha \order \rmod{\rone\rplus\rone}\alpha \order \rmod{\rone}\alpha \fmul \rmod{\rone}\alpha \order \alpha \fmul\alpha$ (by contraction and counit). 
Notice also that, by lax monoidality of $\rmod{}$, we have that $\funit$ is $\rmod{}$-intuitionistic and, 
if $\alpha$ and $\beta$ are $\rmod{}$-intuitionistic, so is $\alpha\fmul\beta$.
Then, \refToProp{tortellini} is generalised as follows.

\begin{prop}\label{prop:graded-left-adj}
Let \ple{\Doc,\rmod{}}  be an \RR-graded doctrine and $\fun{f}{X}{Y}$ an arrow of the base. 
If $\DReIdx{\Doc}{f}$ has a left adjoint $\Ex_f$, then 
$\Ex_f(\funit_X)$ is $\rmod{}$-intuitionistic. 
\end{prop}
\begin{proof}
We have 
$\funit_X\order\rmod{\res}\funit_X$ by lax monoidality of $\rmod{\res}$ and 
$\funit_X\order \DReIdx{\Doc}{f}(\Ex_f(\funit_X))$ by the adjunction $\Ex_f\dashv\DReIdx{\Doc}{f}$. 
This implies
$\funit_X \order \rmod{\res}\DReIdx{\Doc}{f}(\Ex_f(\funit_X)) = \DReIdx{\Doc}{f}(\rmod{\res}\Ex_f(\funit_X))$, by naturality and monotonicity of $\rmod{\res}$. 
Then we get $\Ex_f(\funit_X) \order \rmod{\res}\Ex_f(\funit_X)$, again by the adjunction $\Ex_f\dashv\DReIdx{\Doc}{f}$, as needed. 
\end{proof}
Relying on 
 \refToProp{graded-left-adj} implies that equality predicates in an elementary \RR-graded doctrines are $\rmod{}$-intuitionistic. 

At this point, the natural question is the following: 
how does this (standard) notion of equality relate to our quantitative equality? 
To answer this question in a precise way, first of all we observe that also elementary \RR-graded doctrines can be organised in a 2-category, and then we compare it with the 2-category of \RR-\lip doctrines. 

Elementary \RR-graded doctrines are the objects of the 2-category \ElGMD{\RR} where 
a 1-arrow from \ple{\Doc,\rmod{}} to \ple{\aDoc,\armod{}} is a 1-arrow \oneAr{\ple{F,f}}{\ple{\Doc,\rmod{}}}{\ple{\aDoc,\armod{}}} in \GMD{\RR} such that  
$\Ex_{\id{FX}\times\Delta_{FA}} \circ f_{X\times A} = f_{X\times A\times A}\circ \Ex_{\id{X}\times \Delta_A}$, for all objects $X$ and $A$ in the base of \Doc. 
A 2-arrow from \ple{F,f} to \ple{G,g} is a 2-arrow \twoAr{\theta}{\ple{F,f}}{\ple{G,g}} in \GMD{\RR}. Compositions and identities are those of \GMD{\RR}. 

First of all, we observe that an elementary \RR-graded doctrine \ple{\Doc,\rmod{}} is also \RR-\lip. 
Indeed, given an object $A$ in the base of $\Doc$, the equality predicate $\delta_A$ is an affine $\Doc$-distance, by \refToProp{elem-dist} and \refToProp{graded-left-adj}, and it satisfies the graded substitutive property since we have 
\[
\DReIdx{\Doc}{\ple{\pi_1,\pi_2}}(\alpha) \fmul \rmod{\rone} \DReIdx{\Doc}{\ple{\pi_2,\pi_3}}(\delta_A) 
\order  \DReIdx{\Doc}{\ple{\pi_1,\pi_2}}(\alpha) \fmul \DReIdx{\Doc}{\ple{\pi_2,\pi_3}}(\delta_A) 
\order \DReIdx{\Doc}{\ple{\pi_1,\pi_3}}(\alpha)
\]
by counit and substitutivity of $\delta_A$ (\cf \refToDef{lelementary}). 
This observation immediately provides us with a 2-functor (actually an inclusion) 
$\ELFun:\ElGMD{\RR}\hookrightarrow\EGMD{\RR}$.

In the rest of the section we will show that the 2-functor $\ELFun$ is 2-comonadic. 
We start by constructing a candidate to be its right 2-adoint. 
Consider an \RR-\lip doctrine \ple{\Doc,\rmod{},\delem} where $\CC$ is the base of $\Doc$. 
We denote by 
\In{\CC}{\delem} the full subcategory  of \CC on those objects $A$ such that $\delem_A$ is $\rmod{}$-intuitionistic. 
This category inherits finite products from \CC. 
Indeed, the terminal object of \CC is in \In\CC\delem, as $\delem_1 = \funit$ is $\rmod{}$-intuitionistic, and 
given objects $A$ and $B$ in \In\CC\delem, the product $A\times B$ is in \In\CC\delem as well, since 
$\delem_{A\times B} = \delem_A\dmul\delem_B$ is $\rmod{}$-intuitionistic as both $\delem_A$ and $\delem_B$ are. 
Denote by \DIn{\Doc} the restriction of \Doc to \In{\CC}{\delem}; then \ple{\DIn\Doc,\rmod{}} is a \RR-graded doctrine. 

\begin{prop}\label{prop:lip-elem}
\ple{\DIn{\Doc},\rmod{}} is an elementary \RR-graded doctrine. 
\end{prop}
\begin{proof}
We just have to prove that \ple{\DIn\Doc,\rmod{}} is elementary according to \refToDef{lelementary}. 
Since \ple{\Doc,\rmod{},\delem} is \RR-\lip, we know that, for every object $A$ in $\In\CC\delem$, 
$\delem_A \in \Doc(A\times A) = \DIn\Doc(A\times A)$ and $\delem_A$ is reflexive as it is a $\Doc$-distance. 
The substitutivity property follows since $\delem_A$ is $\rmod{}$-intuitionistic by definition of \In\CC\delem and by the \RR-sustitutivity property of \RR-\lip doctrines as detailed below. 
Indeed, we have that, 
for every $\alpha \in \DIn\Doc(X\times A) = \Doc(X\times A)$, 
there is $\res\in \RSet$, such that 
\[
\DReIdx{\Doc}{\ple{\pi_1,\pi_2}}(\alpha)\fmul\DReIdx{\Doc}{\ple{\pi_2,\pi_3}}(\delem_A) \order 
\DReIdx{\Doc}{\ple{\pi_1,\pi_2}}(\alpha)\fmul\rmod\res\DReIdx{\Doc}{\ple{\pi_2,\pi_3}}(\delem_A) \order 
\DReIdx{\Doc}{\ple{\pi_1,\pi_3}}(\alpha) \qedhere
\]
\end{proof}

Indeed, if \oneAr{\ple{F,f}}{\ple{\Doc,\rmod{},\delem}}{\ple{\aDoc,\armod{},\adelem}} is a 1-arrow in \EGMD\RR and $\delem_A$ is $\rmod{}$-intuitionistic, then 
$\adelem_{FA}$ is $\armod{}$-intuitionistic, as 
$\adelem_{FA} =  f_{A\times A}(\delem_A) \order f_{A\times A}(\rmod{\res}\delem_A) = \armod{\res}f_{A\times A}(\delem_A) = \armod{\res} \adelem_{FA}$. 

Thanks to this observation, 
we can easily extend the construction above to a 2-funtor \fun{\LEFun}{\EGMD\RR}{\ElGMD\RR}, which just restricts 1-arrows and 2-arrows to the subcategories of shape \In\CC\delem introduced above. 
Finally, we get the following theorem. 

\begin{thm}\label{thm:elem-comonad}
The 2-functor \fun{\ELFun}{\ElGMD\RR}{\EGMD{\RR}} is 2-comonadic: 
\begin{enumerate}
\item $\LEFun$ is the right 2-adoint of $\ELFun$ and 
\item $\ElGMD\RR$ is isomorphic to the 2-category of coalgebras for the KZ 2-comonad $\ELFun\circ\LEFun$. 
\end{enumerate}
\end{thm}
\begin{proof} (Sketched) 
Since $\LEFun\circ\ELFun$ is the identity, the unit is the identity as well, while 
the counit is $\ple{I^\Doc,i^\Doc}:\ple{\DIn{\Doc},\rmod{},\delem} \to \ple{\Doc,\rmod{},\delem}$,  
where $I^\Doc$ is the  inclusion of $\In{\CC}{\delem}$ into $\CC$ and $i^\Doc$ is the identity. 
Then, the two triangular identities trivially hold, proving the first item. 
Towards a proof of the second one, 
let $\ple{F,f}$ be a coalgebra on $\ple{\Doc,\rmod{},\delem}$, then the condition $\ple{I^\Doc,i^\Doc}\circ\ple{F,f}=\ple{\Id{},\id{}}$ 
leads to $\ple{F,f} = \ple{I^\Doc,i^\Doc} = \ple{\Id{},\id{}}$, since $I^\Doc$ is an inclusion and $i^\Doc$ is the identity. 
This in particular proves that the 2-comonad $\ELFun\circ\LEFun$ is KZ. 
Moreover, this implies that, 
for every object $A$ in the base of $\Doc$, 
$\delem_A$ is $\rmod{}$-intuitionistic, hence, $\ple{\Doc,\rmod{},\delem}$ already is in $\ElGMD\RR$.
The converse inclusion is straightforward. 
\end{proof}

Composing the 2-adjuntions in   \refToThm{graded-adj} and \refToThm{elem-comonad} we get a 2-adjunction between \GMD\RR and \ElGMD\RR.  Although in general 2-comonadic 2-adjunctions do not compose, this 2-adjunction turns out to be 2-comonadic and as it can be observed  following arguments similar to those already used in  \refToThm{graded-adj} and \refToThm{elem-comonad}
\[\xymatrix{
\GMD\RR \ar[r]<+.8ex>^{\GMEFun} \ar@{}[r]|{\top} & \ar[l]<+.8ex> \EGMD\RR  \ar[r]<+.8ex>^{\LEFun} \ar@{}[r]|{\top}& \ar@{_(->}[l]<+.8ex> \ElGMD\RR \\
\PD\ar@{_(->}[u] \ar[rr]<+.8ex> \ar@{}[rr]|{\top} && \ar@{_(->}[ll]<+.8ex> \ED \ar@{_(->}[u] \ar@{_(->}[ul]
}\]

The square on the left is \refToDia{diagram}. The right triangle is commutative. This is due to the fact that the composition of $\LEFun$ with the inclusion of $\ElGMD\RR$ into $\EGMD\RR$ is the identity on $\ElGMD\RR$.

%% file: altredoc.tex
\section{Richer fragments of Linear Logic}
\label{sect:altredoc}

The calculus \ERPLL, introduced 
in \refToSect{grade-syntax}, 
is the minimal core calculus for quantitative equality, as  
it only considers connectives necessary to deal with it, namely,  $\Ltensor,\Lone$ and $\Lbang_\res$.
This section extends quantitative equality to larger fragments of  Linear Logic, introducing other connective in \ERPLL. 

Each new connective comes with usual rules, hence we do not report them. 
On the other hand, we have to extend the definition of the function $\gr$ to new connectives, in order to 
determine resources they require to derive substitutions. 

Since constants do not depend on variables, there is nothing to substitute, 
hence the cost to derive a substitution is zero. 
Thus we set 
$$\gr(\Lzero,x)=\gr(\Ltop,x)=\rzero$$
A substitution in a formula built  out of binary multiplicative connectives costs the sum of the resources that the two subformulas need to derive the substitution, because both of them can be subsequently used.
Thus we set 
\[ 
\gr(\form\multimap\aform,x)=\gr(\form,x)\rplus\gr(\aform,x)
\]
A substitution in a formula built out of quantifiers costs the same resources as the subformula, only if the substituted and the quantified variable differs, otherwise the cost is zero. 
Thus we set 
\[
\gr(\forall z. \form,x)=\gr(\exists z. \form,x)= \begin{cases} 
\gr(\form,x)&\mbox{ if } x\ne z\\ 
\rzero&\mbox{ if } x=z
\end{cases}
\] 
This is because, 
since $z$ is bounded in $\forall z.\form$ and $\exists z.\form$, substitutions propagate to subformulas only if the substituted variable is different from $z$, otherwise $\forall z.\form$ and $\exists z.\form$ behave as constants. 

The extension of \ERPLL by additive connectives requires additional structure on \RR. 
A straightforward sufficient condition is that \RR has binary suprema denoted by $\vee$. 
Then, the cost of a substitution in a formula built out of binary additive connectives is the supremum of the resources that the two subformulas need to derive a substitution, 
because only one of them  can be subsequently used. 
Thus we set 
\[
\gr(\form\,\&\,\aform,x)=\gr(\form\oplus\aform,x)= \gr(\form,x)\vee \gr(\aform,x)
\] 
When all the connectives above are considered we abbreviate the resulting calculus by \IERLL. 

A way to obtain classical calculi is by adding a constant $\bot$ which is dualising, namely, it satisfies the entailment $\lcons{(\form\multimap\bot)\multimap\bot}{\form}$. 
Being a constant, $\bot$ is such that  $\gr(\bot,x) = \rzero$.
The classical version of \IERLL is denoted by \ERLL.

It is straightforward to define $\RR$-graded doctrines that correspond to the various fragments of (classical or intuitionistic) Linear Logic.  
Here we consider only the following. 

\begin{defi}\label{def:MLL-docs}
A \dfn{multiplicative (linear) doctrine} is a \pl doctrine $\ple{\Doc,\fmul,\funit}$ such that
\begin{itemize}
\item for every $A$ and every $\alpha,\beta,\gamma$ in $P(A)$ there is $\alpha\multimap \beta$ in $\Doc(A)$ such that $\gamma\le \alpha\multimap \beta$ if and only if $\gamma\fmul\alpha\le \beta$
\item for every \fun{f}{X}{A}  and every $\alpha,\beta$ in $\Doc(A)$, it holds 
 $\DReIdx{\Doc}{f}(\alpha\multimap\beta)=\DReIdx{\Doc}{f}(\alpha)\multimap\DReIdx{\Doc}{f}(\beta)$;
\end{itemize}
A \dfn{multiplicative-additive  (linear) doctrine} is a multiplicative linear  doctrine $\ple{\Doc,\fmul,\funit}$ such that
$\fun{\Doc}{\CC\op}{\Pos}$ factors through the category of lattices.

A multiplicative linear  doctrine $\Doc$ \dfn{has quantifiers} if for every projection $\pi:A\times B\to B$ the map $\Doc_\pi$ has a left adjoint $\Ex_\pi$ and a right adjoint $\Al_\pi$ and these are natural in $A$ (this naturality condition is sometimes called Beck-Chevalley condition). 

A multiplicative linear  doctrine is \dfn{classical} if for each object $A$ 
there is $\bot_A$ in $\Doc(A)$ such that $(\alpha\multimap \bot_A)\multimap \bot_A=\alpha$, which is preserved by reindexing. 

A \dfn{first order linear  doctrine} is a multiplicative-additive linear  doctrine with quantifiers. 
\end{defi}

The syntactic doctrine associated with the $(\Ltensor,\Lone,\multimap)$-fragment of ILL is a multiplicative linear doctrine;  
the one associated with the $(\Ltensor,\Lone,\multimap, \oplus,\Lzero,\LLand,\Ltop)$-fragment of ILL is a multiplicative-additive linear doctrine; 
the one associated with full first order ILL is a first order linear doctrine.
Similarly, syntactic doctrines associated with classical variants of the above fragments provides classical versions of the corresponding doctrines. 

The following propositions show that the additional structure of \RR-graded doctrines modelling larger fragments of LL is inherited by their \lip completion. 

\begin{prop}\label{prop:fragments} 
Suppose $\ple{\Doc, \rmod{}}$ is an $\RR$-graded doctrine and consider its \lip completion \fun{\DMet{\Doc}{\rmod{}}}{\Met{\Doc}{\rmod{}}\op}{\Pos}.
\begin{itemize}
\item If $\Doc$ is multiplicative, so is $\DMet{\Doc}{\rmod{}}$;
\item if $\Doc$ is multiplicative and classical, so is $\DMet{\Doc}{\rmod{}}$;
\item if $\Doc$ is multiplicative and has quantifier, so is $\DMet{\Doc}{\rmod{}}$.
\end{itemize}
\end{prop}

\begin{proof}
 Suppose \Doc is multiplicative. 
Take $\ple{A,\dist}$ in $\Met{\Doc}{\rmod{}}$ and $\alpha,\beta$ in $\Des{\dist}(A)$, hence 
there are $\res,\ares$ in $\RSet$ such that $\DReIdx{\Doc}{\pi_1}(\alpha)\fmul \rmod{\res}\dist\le \DReIdx{\Doc}{\pi_2}(\alpha)$ and $\DReIdx{\Doc}{\pi_1}(\beta)\fmul \rmod{\ares}\dist\le \DReIdx{\Doc}{\pi_2}(\beta)$. 
Then, using contraction of $\rmod{}$ and symmetry of $\dist$, we have 
$\DReIdx{\Doc}{\pi_2}(\alpha)\fmul\DReIdx{\Doc}{\pi_1}(\alpha\multimap\beta)\fmul \rmod{\res\rplus\ares}\dist\le 
\DReIdx{\Doc}{\pi_2}(\alpha)\fmul\DReIdx{\Doc}{\pi_1}(\alpha\multimap\beta)\fmul \rmod{\res}\dist\fmul \rmod{\ares}\dist\le 
\DReIdx{\Doc}{\pi_1}(\alpha)\fmul\DReIdx{\Doc}{\pi_1}(\alpha\multimap\beta)\fmul \rmod{\ares}\dist\le 
\DReIdx{\Doc}{\pi_1}(\beta)\fmul \rmod{\ares}\dist\le \DReIdx{\Doc}{\pi_2}(\beta)$, showing that $\DReIdx{\Doc}{\pi_1}(\alpha\multimap\beta)\fmul \rmod{\res\rplus\ares}\dist\le \DReIdx{\Doc}{\pi_2}(\alpha\multimap\beta)$. This makes $\alpha\multimap\beta$ an element of $\Des{\dist}$.

 Suppose \Doc is multiplicative and classical and 
take $\ple{A,\dist}$ in $\Met{\Doc}{\rmod{}}$.  
By weakening of $\rmod{}$, it holds 
$\DReIdx{\Doc}{\pi_1}(\bot_A)\fmul \rmod{\rzero}\dist\le \DReIdx{\Doc}{\pi_1}(\bot_A)\fmul\funit =\DReIdx{\Doc}{\pi_1}(\bot_A) = \bot_{A\times A} = \DReIdx{\Doc}{\pi_2}(\bot_A)$ so $\bot_A$ is in $\Des{\dist}(A)$.

 Suppose $\Doc$ is multiplicative with quantifiers and take $\gamma$ in $\Des{\dist\dmul\adist}(A\times B)$. 
By naturality $\DReIdx{\Doc}{\pi_1}(\Ex_{\pi_1}\alpha)\fmul \rmod{\res}\adist=
\Ex_{\ple{\pi_1,\pi_2}}(\DReIdx{\Doc}{\ple{\pi_1,\pi_3}}(\alpha))\fmul \rmod{\res}\adist \le 
\Ex_{\ple{\pi_1,\pi_2}}(\DReIdx{\Doc}{\ple{\pi_1,\pi_3}}(\alpha)\fmul \rmod{\res}\DReIdx{\Doc}{\ple{\pi_1,\pi_2}}(\adist))\le 
\Ex_{\ple{\pi_1,\pi_2}}\DReIdx{\Doc}{\ple{\pi_2,\pi_3}}\alpha=
\DReIdx{\Doc}{\pi_2}\Ex_{\pi_1}\alpha$ for some $\res$ in $\RR$. 
So descent data are closed under left adjoints along projections. 
Similarly, for right adjoints we have 
$\DReIdx{\Doc}{\pi_1}(\Al_{\pi_1}\alpha)\fmul \rmod{\res}\adist=
\Al_{\ple{\pi_1,\pi_2}}(\DReIdx{\Doc}{\ple{\pi_1,\pi_3}}(\alpha))\fmul \rmod{\res}\adist \le 
\Al_{\ple{\pi_1,\pi_2}}(\DReIdx{\Doc}{\ple{\pi_1,\pi_3}}(\alpha)\fmul \rmod{\res}\DReIdx{\Doc}{\ple{\pi_1,\pi_2}}(\adist))\le 
\Al_{\ple{\pi_1,\pi_2}}\DReIdx{\Doc}{\ple{\pi_2,\pi_3}}\alpha=
\DReIdx{\Doc}{\pi_2}\Al_{\pi_1}\alpha$ for some $\res$ in $\RR$.  
\qedhere
\end{proof}

\begin{prop}\label{prop:fragments2} 
Let $\ple{\Doc, \rmod{}}$ be an $\RR$-graded doctrine where $\RR$ has finite suprema and consider its \lip completion \fun{\DMet{\Doc}{\rmod{}}}{\Met{\Doc}{\rmod{}}\op}{\Pos}.
\begin{itemize}
\item if $\Doc$ is multiplicative-additive, so is $\DMet{\Doc}{\rmod{}}$;
\item if  $\Doc$ is first order linear, so is $\DMet{\Doc}{\rmod{}}$.
\end{itemize}
\end{prop}

\begin{proof}
Suppose $\Doc$ is multiplicative-additive. 
Take $\ple{A,\dist}$ in $\Met{\Doc}{\rmod{}}$ and $\alpha,\beta$ in $\Des{\dist}(A)$, hence 
there are $\res,\ares$ in $\RSet$ such that $\DReIdx{\Doc}{\pi_1}(\alpha)\fmul \rmod{\res}\dist\le \DReIdx{\Doc}{\pi_2}(\alpha)$ and $\DReIdx{\Doc}{\pi_1}(\beta)\fmul \rmod{\ares}\dist\le \DReIdx{\Doc}{\pi_2}(\beta)$. 
Then, $\DReIdx{\Doc}{\pi_1}(\alpha\Land \beta)\fmul \rmod{\res\vee\ares}\dist\le(\DReIdx{\Doc}{\pi_1}(\alpha)\fmul \rmod{\res\vee\ares}\dist)\Land (\DReIdx{\Doc}{\pi_1}(\beta)\fmul \rmod{\res\vee\ares}\dist)$. 
Since $\rmod{\res\vee\ares}\dist\le \rmod{\res}\dist$ and  $\rmod{\res\vee\ares}\dist\le \rmod{\ares}\dist$, it follows 
$\DReIdx{\Doc}{\pi_1}(\alpha\Land \beta)\fmul \rmod{\res\vee\ares}\dist\le
(\DReIdx{\Doc}{\pi_1}(\alpha)\fmul \rmod{\res}\dist)\Land (\DReIdx{\Doc}{\pi_1}(\beta)\fmul \rmod{\ares}\dist)\le \DReIdx{\Doc}{\pi_2}(\alpha)\Land \DReIdx{\Doc}{\pi_2}(\beta)$, which proved $\Des{\dist}(A)$ is closed under binary meets. 
By a similar argument proves the closure under binary joins. 
Top and bottom elements belong to $\Des{\dist}(A)$ as one can prove using the same arguments as for $\bot$. The second item follows by previous item and \refToProp{fragments}. \qedhere
\end{proof}

\begin{cor}\label{cor:classic} 
Let $\ple{\Doc, \rmod{}}$ be an $\RR$-graded doctrine where $\RR$ has finite suprema. 
If  \Doc is classical and first order linear, so is  
the \lip completion \fun{\DMet{\Doc}{\rmod{}}}{\Met{\Doc}{\rmod{}}\op}{\Pos}. 
\end{cor} 

\refToThm{graded-adj}  shows that the notion of quantitative equality given in this paper is coalgebraic, in the sense that \lip doctrines are the coalgebras of a comonad over the category of graded doctrines. This generalizes a known situation that holds in the non-linear case, where elementary doctrines are the coalgebras of a comonad over the category of primary doctrines. This is summarised by the \refToDia{diagram} recalled below on the left.
\begin{equation}\label{eq:diagrampp} 
\xymatrix{
\GMD{\RR} \ar@<1ex>[r]^{\GMEFun}_-{\top} & \EGMD{\RR}\ar@<1ex>[l] \\ 
\PD \ar[u]\ar@<1ex>[r]_-{\top} & \ED\ar[u] \ar@<1ex>[l]
}
\ \ \ \ 
\xymatrix{
\ED\ar@{^(->}[rr]<-1.2ex>&&\ar[ll]<-1ex>^-{\top}\PD
\\
\HD\ar@{^(->}[u]\ar@{^(->}[rr]<-1.2ex>&&\ar[ll]<-1ex>^-{\top}\FOD\ar@{^(->}[u]
}
\end{equation}
In the non-linear case the coalgebraic nature of equality  behaves well with respect to other connectives or quantifiers that the doctrine may have. In particular if one consider the 2-category \FOD of first order doctrines (\ie those doctrines modelling full first order logic) and its subcategory \HD on hyperdoctrines (\ie those first order doctrines that are also elementary) one shows that \HD is the category of coalgebras of a comonad on \FOD. \refToDia{grafico1} (rewritten above on the right) shows the adjoint situation that determines the mentioned comonad.

Gluing the two diagram together suggests the existence of a cube of 2-categories of doctrines. We devote the rest of the section to complete such a cube.

Denote by \LLL{\RR} the 2-full 2-subcategory of \GMD{\RR} on first order \RR-graded  doctrines 
and 1-arrows that preserve the first order linear  structure. 
An \emph{\RR-quantitative hyperdoctrine} is 
a first order \RR-\lip. 
The 2-full 2-subcategory of \EGMD{\RR} on \RR-quantitative hyperdoctrines is \ELLL{\RR}. 
\refToProp{fragments2} shows that (if \RR has finite suprema) the 2-adjunction between \GMD{\RR} and \EGMD{\RR} restricts to a 2-adjunction between 
\LLL{\RR} and \ELLL{\RR}. 
Then, we can extend the diagrams \ref{eq:diagrampp} to the first order case, obtaining the following: 
\[
\xymatrix{
&\ED\ar[ld]\ar[rr]<-1.2ex>&&\ar[ll]<-1ex>^-{\ \ \top}_{\GMEFun}\PD\ar[ld]\\
\EGMD{\RR}\ar[rr]<-1.2ex>&&\ar[ll]<-1ex>^-{\ \ \top}_{\GMEFun}\GMD{\RR}&\\
&\HD\ar[ld]\ar@{^(->}[uu]|(.62){\hole}\ar[rr]<-1.2ex>&&\ar[ll]<-1ex>^-{\ \ \top}_{\GMEFun\ \ \ \ \ \ \ }\FOD\ar[ld]\ar[uu]\\
\ELLL{\RR}\ar[uu]\ar[rr]<-1.2ex>&&\ar[ll]<-1ex>^-{\ \ \top}_{\GMEFun}\LLL{\RR}\ar[uu]&
}\] 
The embedding of \FOD into  $\LLL{\RR}$ follows from \refToProp{graded-primary}, as doctrines in \FOD are in particular primary and the rest of the structure does not interact with the modality. 
The same applies to the embedding of \HD into $\ELLL{\RR}$, but 
here one has also to notice that, since doctrines in \HD are in particular elementary, they trivially satisfy the axioms of \lip doctrines (\cf \refToDef{affine-lip}).   
\begin{exa}
The $\RRPos$-\lip doctrine \ple{\LIPDoc,\textbf{m},\delem} of \refToEx{liplip} is an $\RRPos$-quantitative hyperdoctrine because 
it is the \lip completion of \ple{\PP_{[0,\infty]},\textbf{m}} (see \refToExItem{provvisorio2}{1}), which is first order linear as $[0,\infty]$ is a quantale, 
and $\RRPos$ has finite suprema. 
\end{exa}

Finally, note that 
any \RR-quantitative hyperdoctrine  
provides a domain for interpreting the \IERLL calculus, it suffices to extend \refToLem{sound} to all the other connectives given the definition of $\gr$ at the beginning of this section.

%% file: future.tex

\section{Related and future works}\label{sect:future}

A syntactic study of equality in non-modal  Linear Logic can be found 
in \cite{Hodges1996LogicFF}. 
Affineness and replicability of equality are derived from the (non-quantitative) substitution rule. 
In \cite{coniglio}, full Linear Logic is considered. Equality is required to be intuitionistic, in the sense that  $x\Leq y\vdash \Lbang(x\Leq y)$  holds,  
implying again affineness and replicability. 
A solution to recover linearity is in \cite{Hodges1996LogicFF} and  
 is an instance of our calculus, when the semiring is $\NN^=$ and the grade of each symbol is $1$. 

The construction in \refToSect{comonad}, like that of \cite{PasqualiF:cofced,EPR}, originates from a series of works \cite{MaiettiME:quofcm,MaiettiME:exacf,MaiettiME:eleqc} on the notion of quotient in an elementary doctrine. 
An interesting direction for future work is the development of a quantitative counterpart of quotients in \lip doctrines. 

Any QET induces a monad on the category of metric spaces and non-expansive maps. 
For simple  QET, models are proved to be the algebras for such a monad  \cite{BacciMPP18,MardarePP16,MardarePP17,Adamek22}. 
This is a strong result attesting that QETs axiomatise the algebras of the monad induced by themselves. 
Similarly, in our setting, we would like to show that, under suitable assumptions, 
an equational theory in \ERPLL
 induces a monad on categories of the form \Met{\Doc}{\rmod{}}, and 
the category of models in $\DMet{\Doc}{\rmod{}}$ for such a theory is (equivalent to) the category of algebras for such a monad. 
That is,  equational theories in \ERPLL axiomatise quantitative algebras where operations are 
\lip maps.

In \cite{DahlqvistN23,DahlqvistN23mfps} the authors propose quantitative equational theories for linear $\lambda$-calculi also with graded modal types, proving these provide internal languages for autonomous categories enriched over certain metric spaces. 
We would like to see whether we can describe these theories in our calculus as well, possibly starting from the example of graded combinators in \refToSect{qet}.

Another possible application of \lip doctrines is the denotational semantics of various forms of \emph{``graded'' $\lambda$-calculi} (see, \eg, \cite{PetricekOM14,AbelB20}). 
In these calculi variables in typing contexts are annotated by grades from a semiring, describing how much the term uses each variable. 
Consequently, arrow types are labelled by a grade describing how much the function uses its input. 
These data can be interpreted in (the base of) a \lip doctrine: 
types are objects, contexts are products and terms are arrows whose \lip constants agree with the grades reported in the context. 
To interpret labelled arrow type, 
we can exploit the structure of a \lip doctrine to define a \emph{graded exponential}, that intuitively collects all \lip arrows having a fixed \lip constant. 

We also plan  to move the ideas discussed in this paper to a proof-relevant setting, studying connections with type theories. 
It is known that every type theory gives rise to a doctrine \cite{MaiettiME:quofcm,MaiettiME:exacf}. 
Similarly, linear dependent type theories, where linear and intuitionistic types are kept separated in contexts
\cite{CERVESATO200219}, 
give rise to a \pl doctrine. 
Studies of Id-Types in linear type theories are available in the literature, but, to our knowledge, they treat Id-Types as intuitionistic types (see \eg, \cite{KrishnaswamiPB15,Vakar}).
There are also dependent type theories using grading \cite{Atkey18,MoonEO21}, but they do not consider Id-Types. 
Therefore, an interesting  direction will be to use the machinery of the present paper, possibly extended to fibrations,  to study  quantitative Id-Types in linear type theories. 

Finally, in \cite{PimentelON21} it is shown that certain proof systems of intuitionistic Linear Logic with subexponentials can be used to model and reason abount concurrent programming under the processes-as-formulas interpretation. 
It would be interesting  to investigate to which extent this paradigm applies to our calculus for quantitative equality and whether \RR-\lip doctrines can be used to build new models for the associated computational paradigm.